\documentclass[opre,nonblindrev]{informs4}

\usepackage{xcolor}
\usepackage{booktabs} 
\usepackage{apptools}

\usepackage[ruled]{algorithm2e} 

\usepackage{natbib}
\usepackage{float}
\usepackage{multicol}
\bibpunct[, ]{(}{)}{,}{a}{}{,}%
\def\bibfont{\small}%
\usepackage{mathtools}
\usepackage{multirow}
\usepackage{amsbsy, amsfonts, amsgen, amsmath, amsopn, amssymb, amstext,
	amsxtra, bezier, bbm, color, enumerate, graphicx, latexsym, verbatim,
	pictexwd, supertabular, url, dsfont,leftidx, mathrsfs, mathtools, setspace, bbm}

\usepackage[page,header]{appendix}
\usepackage{titletoc}

\usepackage[ruled]{algorithm2e}

\newtheorem{theorem}{Theorem}
\newtheorem{lemma}[theorem]{Lemma}
\newtheorem{corollary}[theorem]{Corollary}

\newtheorem{eg}{Example}[section]

\newtheorem{proposition}{Proposition}
\newcommand{\mb}[1]{\ensuremath{\boldsymbol{#1}}}

\newcommand{\onenu}{\mathbbm{1}_\nu}
\newcommand{\onee}{\mathbbm{1}}
\newcommand{\ed}{\mathbb{E}_{\boldsymbol {\nu}}}
\newcommand{\eo}{\mathbb{E}_{\boldsymbol {\omega}}}
\newcommand{\eod}{\mathbb{E}_{\boldsymbol {\omega,\nu}}}

\newcommand{\pd}{\mathbb{P}_{\boldsymbol {\nu}}}

\newcommand{\nuk}{\nu_{k}}
\newcommand{\nukp}{\nu_{k^+}}

\newcommand{\bsigma}{\boldsymbol {\sigma}}

\def\opt{\textsc{OPT}}
\def\alg{\textsc{ALG}}
\def\galg{\textsc{G-ALG}}
\def\salg{Sample\textsc{ G-ALG}}
\def\astalg{Sample Assort\textsc{ G-ALG}}
\def\astgalg{Assort\textsc{ G-ALG}}
\def\dpg{\textsc{RBA}}

\DeclarePairedDelimiter{\floor}{\lfloor}{\rfloor}

\SetAlFnt{\small}
\SetAlCapFnt{\small}
\SetAlCapNameFnt{\small}
\SetAlCapHSkip{0pt}
\IncMargin{-\parindent}

\usepackage[noframe]{showframe}
\usepackage{framed}
\usepackage{lipsum}
{\endMakeFramed}
\definecolor{shadecolor}{gray}{0.90}

\ECRepeatTheorems

\EquationsNumberedThrough    

\MANUSCRIPTNO{}

\begin{document}
	
	\RUNTITLE{Optimal Competitive Ratio for Reusable Resource Allocation}
	\RUNAUTHOR{Goyal, Iyengar, Udwani}
	\TITLE{Asymptotically Optimal Competitive Ratio for Online Allocation of Reusable Resources}
	
	\ARTICLEAUTHORS{%
		\AUTHOR{Vineet Goyal}
		\AFF{
			Columbia University, \EMAIL{vgoyal@ieor.columbia.edu}} 
	\AUTHOR{Garud Iyengar}
	\AFF{
		Columbia University, \EMAIL{garud@ieor.columbia.edu}}
	\AUTHOR{Rajan Udwani}
	\AFF{
		UC Berkeley, \EMAIL{rudwani@berkeley.edu}}}

\ABSTRACT{
	We consider the problem of online allocation (matching, budgeted allocations, and assortments) of reusable resources where an adversarial sequence of resource requests is revealed over time and any allocated resource is used/rented for a stochastic duration drawn independently from a resource dependent usage distribution. Previously, it was known that a greedy algorithm is 0.5--competitive against the clairvoyant benchmark that knows the entire sequence of requests in advance (Gong et al.\ (2021)). 
		We give a novel algorithm that is $(1-1/e)$--competitive for arbitrary usage distributions when the starting capacity of each resource is large and the usage distributions are known. This is the best achievable competitive ratio guarantee for the problem, i.e., no online algorithm can have better competitive ratio. We also give a distribution oblivious online algorithm and show that it is $(1-1/e)$--competitive in special cases.   
	At the heart of our algorithms is a new quantity that factors in the potential of reusability for each resource by (computationally) creating an asymmetry between identical units of the resource. We establish the performance guarantee for our algorithms by constructing a feasible solution to a novel 
		system of inequalities that allows direct comparison with the clairvoyant benchmark instead of a linear programming (LP) relaxation of the benchmark. Our technique generalizes the primal-dual analysis framework for online resource allocation and may be of broader interest. 
}

\maketitle




\section{Introduction}\label{sec:intro}

The problem of online bipartite matching~ \citep{kvv}, and its various generalizations such as online assortment optimization~\citep{negin}, online budgeted allocations or Adwords~\citep{msvv} and others (see \cite{survey}), have played an important role in shaping the theory and practice of resource allocation in dynamic environments. Driven by applications such as internet advertising, personalized recommendations, crowdsourcing etc., 
these settings focus on online allocation of non-reusable resources, where resources can only be used once.  
	Emerging applications in sharing economies, data centers, and make-to-order services, 
	has prompted a surge of interest in understanding online allocation of resources that are reusable, i.e., \emph{resources that may be used/allocated multiple times}~\citep{dickerson, RST18, reuse}. 

\textbf{Online Bipartite Matching of Reusable Resources (OBMR):} We start by describing a fundamental model where reusable resources need to be matched to sequentially arriving demand. Consider a bipartite graph $G=(I,T,E)$, one side of which is a set of resources $i\in I$ with inventory/capacity $\{c_i\}_{i\in I}$. Vertices $t\in T$ arrive sequentially at times $\{a(t)\}_{t\in T}$. When vertex $t$ arrives, we see the edges incident on $t$ and must make an irrevocable decision to match $t$ to at most one resource with available capacity, without any knowledge about future arrivals (that could be set by an adversary). If we match $t$ to resource $i$, one unit of the resource is used by arrival $t$ for a random duration $d_{it}$ drawn independently from distribution $F_i$.  
At time $a(t)+d_{it}$, this unit comes back to the system and is immediately available for a re-match. While distributions $F_i$ are known a priori, the exact duration of any use is revealed to us only upon return. A match between $t$ and $i$ earns reward $r_i$ (later, we allow rewards to be an arbitrary function of the random usage duration). With the goal of maximizing the expected total reward, we seek online algorithms that \emph{compete} well against an optimal \emph{clairvoyant} benchmark on every arrival sequence. The clairvoyant knows the entire arrival sequence in advance and it matches arrivals in order, observing the realizations of usage durations only when resources return\footnote{In Appendix \ref{appx:stronger}, we show that the stronger offline benchmark which also knows realizations of usage durations in advance does not allow a non-trivial competitive ratio guarantee for any online algorithm. We also show that clairvoyant is equivalent to a natural offline LP benchmark in the settings of interest here.} (same as an online algorithm). Formally, let $\mathcal{G}=(G,\{c_i,r_i,F_i\}_{i\in I},\{a(t)\}_{t\in T})$ denote a problem instance and let $\mb{\mathcal{G}}$ denote the set of all instances. {\color{black} Let $\alg(\mathcal{G})$ and $\opt(\mathcal{G})$ denote the expected reward of an online algorithm \alg\ and the clairvoyant respectively.
	\[\text{Competitive ratio of \alg: } \qquad \inf_{\mathcal{G}\in \mb{\mathcal{G}}}\frac{\alg(\mathcal{G})}{\opt(\mathcal{G})}. \]}
This model generalizes classic online bipartite matching with non-reusable resources (OBM), which corresponds to the special case where every usage duration is larger than $a(T)$ (time of last arrival) with probability 1.  
Borrowing an example application from \cite{reuse}, modern cloud platforms such as Amazon Web Services, Google Cloud, and Microsoft Azure, commonly support large scale data storage. This service is widely used by online video platforms such as YouTube and Netflix. A given data file is typically stored in only a subset of servers that are part of the cloud. A user request for data can only be sent to the subset of servers with the required file(s). Each server can concurrently serve only a limited number of requests while meeting the stringent low latency requirements on such platforms. The online matching model described above captures this setting, with servers in the cloud modeled as reusable resources. The capacity of a server is the maximum number of requests that it can simultaneously serve. 
We are interested in developing the best possible 
online algorithm for allocating reusable resources in the \emph{large capacity/inventory regime}, i.e., we seek algorithms that attain the highest achievable 
competitive ratio as $c_i\to+\infty$ for every $i\in I$. 
The large inventory assumption is prominent in a variety of online resource allocation problems. This includes the seminal Adwords problem~\citep{msvv}, where this is called the small bids assumption, and settings such as display advertising~\citep{displayad}, online assortment optimization~\citep{negin}, online order fulfillment~\citep{fulfil}, two-sided online assortment~\citep{aouad2020online} and many others \citep{survey}. 
{ \color{black}  
For the sake of simplicity, we focus most of our discussion on the OBMR model. Our main results as well as results shown in prior work, hold in more general settings 
that we will discuss in Section \ref{sec:mtoasst}.}

\textbf{Prior work:} \cite{reuse} showed that the greedy algorithm that matches each arrival to the available resource with the highest reward, is 0.5--competitive, i.e., on every arrival sequence, greedy has expected total reward that is at least half the expected reward of clairvoyant. The result holds for arbitrary usage distributions and starting inventory. 
Typically, the main algorithm design question in online resource allocation is to find algorithms that outperform greedy. Prior to our work, no result beating the 0.5--competitive guarantee of greedy was known for our setting. 
When resources are non-reusable, 
a \emph{simple and scalable} algorithm that incorporates the fraction of remaining inventory of resources in allocation decisions, achieves the best possible competitive ratio guarantee of $(1-1/e)$ 
~\citep{msvv,negin}. This algorithm, often called \emph{Inventory Balancing} or simply \emph{Balance}, 
tracks the remaining capacity $y_i(t)$ of each resource and matches an arrival $t$ with edge set $S_t$ as follows,
\begin{equation}\label{Bal}
		\text{Balance: }\qquad \text{Match $t$ to   } \argmax_{i\in S_t}\,\, r_i \left(1-g\left(\frac{y_i(t)}{c_i}\right)\right),
	\end{equation}
where 
$g(x)$ is a non-increasing function. By choosing $g(\cdot)$ we can trade-off between conserving resources with low inventory for future arrivals and maximizing the revenue from each allocation. Setting $g(x)=e^{-x}$ leads to the algorithm with the best possible guarantee of $(1-1/e)$ and we refer to this particular choice as the Balance algorithm. We recover the greedy algorithm by setting $g(x)=\onee(x=0)$, where $\onee(x=0)$ is a binary indicator that is 1 only when $x=0$.  Finally, note that $(1-1/e)$ is also an upper bound on the best possible competitive ratio for our (more general) setting. 

\subsection{Our Contributions}\label{sec:ourcontrib}
{\color{black} 
	When resources are reusable, in addition to conserving resources for unknown future arrivals we may also need to account for the fact that resources that are currently in use may return ``soon'' (and not require conservation). However, recall that the time of return of an in-use resource and the types of future arrivals are both unknowns 
- a feature common in many practical settings. Under this combination of uncertainty, we show that previously known algorithms, such as Balance, have a competitive ratio that is strictly lower than $(1-1/e)$. 
\begin{theorem}\label{balup}
	The competitive ratio of Balance with $g(x)=e^{-x}$ is less than $0.626$ $(<(1-1/e))$ for OBMR. 
\end{theorem}
The proof is included in Appendix \ref{appx:twopoint}. We construct an instance where usage durations are sampled from two-point distributions. In concurrent work, \cite{feng3} showed that Balance is $(1-1/e)$--competitive for reusable resources when usage durations are deterministic. Our result demonstrates the need for new algorithms when durations are uncertain. }
We start with a discussion of our algorithmic ideas and our main results, followed by a brief description of the technical ideas. 
Appendix \ref{sec:challenges} substantiates these discussions with examples. 

\subsubsection*{Algorithmic Ideas.}

Our first algorithmic ingredient is a new quantity that factors in the potential of reusability for each resource by (computationally) creating an asymmetry between identical units of the resource. For resource $i\in I$, consider an arbitrary ordering over its $c_i$ units. 
Let indices $k\in[c_i]$ denote the ordering, with index $k=c_i$ representing the highest ranked unit. 
{\color{black}When a customer is matched to $i$ we treat this as a match to the \emph{highest ranked} unit of $i$ available at the time that the customer arrives. This unit is then unavailable until the customer returns the resource.} Over time, higher ranked units get matched more often and this induces a difference between units that are otherwise identical. Formally, let $z_i(t)$ denote the highest ranked unit of $i$ that is available just before $t$ arrives and let $S_t$ denote the set of available resources with an edge to $t$. Then our algorithm, \emph{Rank Based Allocation} (\dpg), matches $t$ to 
\begin{equation*}
\text{\dpg: }\qquad \argmax_{i\in S_t}\quad  \left(1-g\left(\frac{z_i(t)}{c_i}\right)\right)r_i,
\end{equation*}
{\color{black} with ties broken arbitrarily.} We refer to $ \left(1-g\left(\frac{z_i(t)}{c_i}\right)\right)r_i$ as the \emph{reduced price} for resource $i\in I$ at $t$. Observe that \dpg\ is a greedy algorithm on reduced prices. Similar to Balance, it turns out that choosing $g(x)=e^{-x}$ gives the best guarantee for a variety of usage distributions.  
Notice that $z_i(t)$ is always at least as large as the remaining inventory $y_i(t)$. Therefore, for a given value of $y_i(t)$, the reduced price of $i$ is higher in \dpg\ than in Balance. 
Examples \ref{passivity} -- \ref{passivity2} in  Appendix \ref{sec:challenges} numerically demonstrate the difference between RBA and Balance and also show how 
the reduced prices in \dpg\ can capture the ``effective" inventory of each resource -- the units that are available \emph{and} units that are in use but return ``soon". Remarkably, \dpg\ can capture this despite being completely oblivious to usage distributions. 

We establish that \emph{\dpg\ is $(1-1/e)$--competitive for OBMR over a broad class of usage distributions}, including two-point distributions, exponential distributions, and more generally, increasing failure/hazard rate (IFR) distributions.  
A formal statement of the theorem is deferred to Section \ref{sec:rba}. Despite the simplicity of \dpg, analyzing its performance is quite challenging. A major reason is that the output of \dpg\ is stochastic. 
{\color{black} We say that an online algorithm is \emph{realization dependent} when its output depends on the realizations of random usage durations. 
In \dpg, the quantities $z_i(t)$ are random variables that depend on usage durations of matches prior to $t$. In fact, varying a single past duration can sometimes change the value of $z_i(t)$ from $c_i$ to $0$. Note that realization dependence} presents a challenge in evaluating the performance of Balance as well, though the remaining capacity $y_i(t)$ is much more stable to small variations on the sample path. For online matching, we can control some of these stochastic dependencies in the output of \dpg\ 
to show that it achieves the optimal competitive ratio of $(1-1/e)$ for a broad class of usage distributions. 

With the interpretation of \dpg\ as a greedy algorithm on reduced prices, the algorithm generalizes quite naturally to more general settings such as online assortments and budgeted allocation {\color{black}(see Appendix \ref{sec:rbafail})}. However, in analyzing \dpg\ in these more general models we find that the stochastic element of reusability can interact with certain aspects of these models 
in non-trivial ways, causing a failure of 
the structures that enable the analysis of \dpg\ for online matching. 

To address these challenges and show a more general result, we consider {\color{black} \emph{realization independent} algorithms that are not directly influenced by the vagaries of stochastic parts of the problem. Our second algorithm uses \dpg\ at its core but \emph{wraps it in a realization independent framework}.} The final algorithm is not as simple as \dpg, {\color{black} nor oblivious to usage distributions,} but it address the difficulties with analyzing \dpg\ 
in a crisp and unified way. We introduce our realization independent algorithm for OBMR in Section \ref{sec:matching}. 

Let $c_{\min}=\min_{i\in I} c_i$. In the large inventory regime, we have $c_{\min}\to +\infty$. We establish the following guarantee for arbitrary usage distributions.
{\color{black}	
\begin{theorem}[Special Case of Main Result]\label{main}
For OBMR with arbitrary usage distributions, there exists a realization independent algorithm with competitive ratio $(1-1/e-\delta)$, for $\delta=O\left(\sqrt{\frac{\log c_{\min}}{c_{\min}}} \right)$. 
For $c_{\min}\to +\infty$, the guarantee approaches $(1-1/e)$. This is the best possible asymptotic guarantee, i.e., no online algorithm can have asymptotic competitive ratio better than $(1-1/e)$. 
\end{theorem}



\noindent \textbf{Remarks.} Related to Theorem \ref{main}:  
\begin{enumerate} [(i)]
	\item 
	We obtain the $(1-1/e)$ guarantee in a more general setting that is formally described in Section \ref{sec:mtoasst}. 
	\item Our realization independent algorithm requires knowledge of the usage distributions. In contrast, \dpg\ is realization dependent but distribution oblivious. 
		\item Our results also hold when rewards depend on duration of usage. 
	In fact, for resource $i\in I$ the per-unit reward can be any function $r_i(d)$ of the usage duration $d$, provided the expectation $r_i=\mathbb{E}_{d\sim F_i}[r_i(d)]$ is finite. Algorithmically, it suffices to consider just the expected reward $r_i$ in this more general setting (see Appendix \ref{appx:miscstatic}). 
	\item We focus on showing competitive ratio guarantees against clairvoyant but our results also hold against a linear programming (LP) benchmark that has been used in prior work. In fact, we show that the two benchmarks are equivalent in the large inventory regime (see Appendix \ref{appx:LP}).
\end{enumerate}


\subsubsection*{Technical Ideas.}\label{sec:techintro}
The primal-dual technique of \cite{buchbind} and \cite{devanur}, commonly used to analyze algorithms for online resource allocation, converts the problem of finding a lower bound on the competitive ratio of an algorithm to that of finding a feasible solution for the dual of an LP upper bound on the clairvoyant. As we demonstrate in Appendix \ref{primalimpossible}, due to a mix of combinatorial and stochastic elements the standard approach is insufficient for analyzing algorithms in case of reusability. 
Remarkably, this is true even with \emph{fluid} version of reusability where if a unit of some resource $i$ with usage distribution $F_i$ is matched at time $\tau$, then exactly $F_i(d)$ fraction of this unit returns by time $\tau+d$, for every $d\geq 0$. 
{\color{black}
We take a novel approach by developing a certificate of optimality that is \emph{LP free}, i.e., not based on LP duality. 
We propose a system of linear inequalities 
such that proving feasibility of the system certifies competitiveness of the online algorithm that is being analyzed. The inequalities in our system depend directly on the actions of clairvoyant and 
are \emph{easier to satisfy} than the dual of an LP upper bound on the clairvoyant. The increased flexibility of the LP free system allows us to prove competitive ratio results in settings where the classic primal-dual approach fails. 
We present this generalized certificate in Section \ref{sec:certificate}.	 To use the LP free certificate and prove performance guarantees, we require several new technical ideas and structural properties that are discussed in more detail later on. Our LP free analysis framework has also been used in 
subsequent work beyond the settings considered in this paper \citep{aouad2020online, unkads, pricing, manshadi2022online, osow}.} 
\subsection{Related Work}\label{related}
{\color{black}
\begin{table}{\color{black}
\centering
\begin{tabular}{|c|c|c|c|c|c|}
	\cline{3-6}
	\multicolumn{2}{c|}{}   & \multicolumn{2}{c|}{General Case} & \multicolumn{2}{c|}{Large Inventory}               
	                                                        \\
	                                                        \cline{3-6}
	\multicolumn{2}{c|}{}   & \multicolumn{1}{c|}{Lower Bound} & \multicolumn{1}{c|}{Upper Bound} & \multicolumn{1}{c|}{Lower Bound} & \multicolumn{1}{c|}{Upper Bound}                                                                       \\	
	\hline
	\multirow{2}{*}{Non-Reusable}& Matching & ${(1-1/e)}$ (R)      & \multirow{2}{*}{${(1-1/e)}$}  & \multirow{2}{*} {$(1-1/e) $ (D)}&\multirow{2}{*} {$(1-1/e) $}\\
\cline{2-3}
& Assortment & \multirow{1}{*}{${0.5}$ (D)} &&&  \\
\cline{1-6}
	\multirow{2}{*}{Reusable}& Matching &  \multirow{2}{*}{${0.5}$ (D)}     &  \multirow{2}{*}{${(1-1/e)}$} & \multirow{2}{*} {$0.5\to \mb{(1-1/e)} $ \textbf{(R)}}& \multirow{2}{*}{${(1-1/e)}$}\\
\cline{2-2}
& Assortment & &&&  \\
\hline
\end{tabular}
\caption{Summary of best known competitive ratio guarantees (lower and upper bounds) in relevant settings. We use ``Assortment" in the table to refer to settings that generalize matching (see Section \ref{sec:mtoasst} for the general model). The entry in bold represents our contribution. The lower bounds reported for large inventory are in the limit $c_i\to +\infty\,\, \forall i\in I$. The letters (D) and (R) refer to deterministic and randomized algorithms respectively.  }
\label{tab}
}
\end{table}
 We start by discussing prior work on non-reusable resources and then review prior work on reusable resources.  Our focus is mainly on settings that are directly related to our work and we refer to various other sources for missing details of other settings. 
 Table \ref{tab} summarizes our competitive ratio results in the context of previously known results. 
\smallskip

 \textbf{Non-reusable Resources.}  \cite{kvv} introduced the classic problem of online bipartite matching (OBM), a special case of OBMR, where each resource has identical reward (i.e., $r_i=1\,\, \forall i\in I$) and the usage durations are larger than $a(T)$, making the resources non-reusable. They observed that any deterministic online algorithm that generates a maximal matching is 0.5--competitive for OBM and no deterministic online algorithm has a better competitive ratio. Further, they showed that matching arrivals based on a random ranking over all resources has a competitive guarantee of $(1-1/e)$, and this is the highest possible competitive ratio guarantee achievable by an online algorithm. Thus, $(1-1/e)$ serves as an upper bound on the best possible guarantee in various settings that generalize OBM. The original analysis of \cite{kvv} was clarified and considerably simplified in \cite{baum, goel2, devanur}. 
 There are various generalizations of OBM, ranging from vertex-weighted matching \citep{goel} to online submodular welfare maximization \citep{kapralov}. We refer to \cite{survey} and \cite{newsurvey} for a detailed review of these settings. 

Under the assumption of large starting inventory, \cite{pruhs} showed that the Balance algorithm that matches each arrival to the resource with the largest amount of unmatched inventory, is asymptotically $(1-1/e)$--competitive for OBM when all resources have the same starting inventory. This setting is called online $b$-matching where $b$ denotes the starting inventory of each resource. Note that Balance is a deterministic algorithm that is 0.5--competitive without the large inventory assumption\footnote{\cite{kvv} note that all deterministic myopic algorithms that output a maximal matching are 0.5--competitive for online $b$-matching when $b=1$.}. 
\cite{pruhs} also showed that $(1-1/e)$ is the best possible asymptotic guarantee even under the large inventory assumption. This gives an upper bound of $(1-1/e)$ in the large inventory regime for 
various generalizations of OBM. Note that the Balance algorithm for $b$-matching is a special case of the Balance algorithm that we discussed in Section \ref{sec:intro} (see \eqref{Bal}). The generalized form \eqref{Bal} with the trade-off function $g(\cdot)$ was 
proposed by \cite{msvv}, who introduced the Adwords problem and showed that Balance is $(1-1/e)$--competitive for Adwords in the large inventory regime. Adwords generalizes online $b$-matching by allowing resources to have heterogeneous starting inventory and considering arrivals that may require multiple units of the resource that they are matched to  
(see Section \ref{sec:mtoasst} for more details). Finding the best possible competitive ratio for Adwords without the large inventory assumption  remains an open problem\footnote{Greedy is 0.5--competitive for Adwords. The first improvement over greedy was made very recently by \cite{huang2024adwords}.}. 


  More recently, \cite{negin} considered another generalization of OBM where each arrival is offered an assortment (set) of resources from which the arrival randomly chooses at most one resource according to a known choice model 
  (see Section \ref{sec:mtoasst} for more details). They showed that a natural (deterministic) generalization of Balance is $(1-1/e)$--competitive for online assortment optimization in the large inventory regime and 0.5--competitive in general. \cite{chan} showed that the greedy algorithm is 0.5--competitive for a general family of problems that includes online assortment optimization as a special case. 
  Note that $(1-1/e)$ is the optimal guarantee for this problem due to the upper bounds for OBM and online $b$-matching. 
Finding the best possible guarantee for online assortment optimization with arbitrary starting inventory remains an open problem even in a well-studied special case known as online matching with stochastic rewards~\citep{deb}. 
We refer the reader to \cite{negin} for a detailed review of prior work on online allocation with customer choice. 
\smallskip

\textbf{Reusable Resources.} 
In general, the best known guarantee for our setting is due to \cite{reuse}, who showed the greedy algorithm is 0.5--competitive for online assortment optimization with reusable resources. Recall that our basic model (OBMR)  generalizes OBM,  so $(1-1/e)$ is the best possible guarantee for our model even in the large inventory case.  
We focus on the large inventory regime and establish the guarantee of $(1-1/e)$ for a generalization of the setting in \cite{reuse}. Concurrently, 
\cite{feng3} showed that Balance is $(1-1/e)$--competitive for the special case of deterministic usage durations. We note that our result is based on a randomized algorithm and finding a deterministic $(1-1/e)$--competitive algorithm (for large inventory), or proving that no such algorithm exists, remains an open problem. We make some progress in this direction by analyzing \dpg\ (see Section \ref{sec:rba} for more details). 
\smallskip

\textbf{Other Related Settings.} Prior to \cite{reuse}, \cite{dickerson} and \cite{RST18} considered a bayesian setting where resources are reusable and the arrival sequence is not adversarial but sampled from a distribution that is known to the online algorithm. 
The state-of-the art in this model is due to \cite{feng2}, who gave an online algorithm with competitive ratio 1 (asymptotically) for large inventory and competitive ratio 0.5 in the general case. Their result generalizes prior work on OBM with stochastic arrivals (and non-reusable resources). Note that the results for stochastic arrivals hold even if the usage distributions are arrival dependent, i.e., arrival $t$ uses resource $i$ for a duration sampled independently from a distribution $F_{it}(\cdot)$. For adversarial arrivals, \cite{reuse} showed that no non-trivial competitive ratio result is possible when the usage distributions depend on arrivals. We show that this hardness holds even for the special case of \emph{deterministic} arrival dependent usage durations (see Appendix \ref{appx:miscdet}). We refer to \cite{survey} and \cite{feng2} for a detailed review of prior work in the stochastic arrival model for non-reusable resources and reusable resources respectively.  

 Recent works in queuing theory study a setting similar to ours but with the objective of latency minimization on a bipartite network, for instance see \cite{srikant} and \cite{jsq}. Our model ignores the queuing aspect and provides a complementary perspective from the point of view of maximizing the number of successful matches in a \emph{loss system} (where the latency of accepted jobs is zero).
Our setting also bears some resemblances with the stochastic online scheduling problem \citep{megow} and online load balancing problems \citep{azar1,azar2}. While the settings and objectives in these problems differ substantially from ours, it may be an interesting connection to explore in future work.
}
\subsubsection*{Outline: }  
{\color{black}	In the next section, we present the generalized certificate for proving competitive ratio guarantees. 
In Section \ref{sec:matching}, we introduce our realization independent algorithm for online matching and use the LP free certificate to prove Theorem \ref{main}. In Section \ref{sec:rba}, we present our theoretical results for \dpg\ and give an overview of the technical challenges in analyzing \dpg. The detailed analysis of \dpg\ is included in Appendix \ref{appx:rbaoverview}. In Section \ref{sec:mtoasst}, we introduce the more general setting of online assortment optimization with budgeted allocation and present the generalization of Theorem \ref{main}. 
In Section \ref{sec:numerics}, we numerically evaluate the performance of our algorithms and compare with existing algorithms.	In Section \ref{sec:conclusion}, we summarize the results and discuss possible directions for future work.} 


\section{Generalized Certificate of Optimality}\label{sec:certificate}
Inspired by the classical primal-dual analysis, we formulate a system of linear constraints such that finding a feasible solution to the system establishes a lower bound on the competitive ratio of an online algorithm. In primal-dual analysis, this linear system is given by the constraints in the dual that depend only on the natural problem parameters: the graph, starting capacity, usage distributions and per unit revenues. In a departure from this approach we work with a linear system that depends directly on the actions of the clairvoyant algorithm that we are comparing against. 
This results in a set of conditions that are weaker, i.e., easier to satisfy. 

Let \alg\ denote the online algorithm under consideration as well as the expected total reward of the algorithm. Let $\alg_i$ denote the expected reward generated from matching resource $i$ in \alg. Similarly, \opt\ refers to the clairvoyant and its expected reward and $\opt_i$ denotes the expected reward generated from matching resource $i$ in \opt.  
We generate a sample path $\omega$ for \opt\ as follows: For every unit $k\in[c_i]$ of every resource $i\in I$, construct a long list (say of size $T$) of independent samples from distribution $F_i$. When a unit is matched, we draw the first unused sample from the corresponding list of samples. This collection of lists and samples uniquely determines a sample path $\omega$. The matching output by \opt\ is a function of the sample path $\omega$ and we denote the set of arrivals matched to resource $i$ on sample path $\omega$ using $O(\omega,i)$. In the more general case of assortments, $\omega$ also includes randomness due to customer choice. To simplify the discussion and without loss of generality (w.l.o.g.), we let clairvoyant (\opt) be a deterministic algorithm, so that $\omega$ captures all the randomness in \opt\ (see Appendix \ref{appx:miscfrac} for further justification). We use $\nu$ to refer to the sample path of usage durations in the online algorithm (\alg). Let $\eod[]$ denote the expectation over the randomness in sample paths $\omega$ and $\nu$.

At a high level, our main idea is to define non-negative \emph{pseudo-rewards} $\text{PR}_i$ for every resource $i\in I$ such that,
\[\text{PR}_i\geq \alpha\, \opt_i \quad \forall i\in I,\]
here $\alpha$ is the desired competitive ratio. We also want the sum of pseudo-rewards to be upper bounded by the total expected revenue of \alg, i.e.,
\[\sum_{i\in I} \text{PR}_i \leq \alg.\]
Clearly, if we can define such pseudo rewards for every instance then we have a competitive ratio guarantee of $\alpha$. Observe that if we set $\text{PR}_i=\opt_i\,\, \forall i\in I$, then the second constraint may not hold. Setting $\text{PR}_i=\alg_i\,\, \forall i\in I$ will violate the first constraint because $\alg_i$ could be 0 on some instances even if $\opt_i$ is non-zero. We define $\text{PR}_i$ 
by using non-negative variables $\lambda_t(\omega,\nu)$ and $\theta_i$, 
\begin{equation*}
	 \text{PR}_i\,=\,\theta_i + \eod\Big[ \sum_{t\in O(\omega,i) }\lambda_t(\omega,\nu)\Big]\quad  \forall i\in I. 
\end{equation*}
This novel definition allows the pseudo-rewards to be an instance dependent combination of \alg\ and \opt. In particular, we often set the values of these variables to reflect the online actions of \alg\ and \opt. Observe that the pseudo-reward also depends on \opt\ through the sets $O(\omega,i)$. Since \opt\ matches each arrival to at most one resource, we have, $O(\omega,i)\cap O(\omega,j)=\emptyset$ for every $i\neq j$ and the union $\cup_{i\in I} O(\omega,i)$ includes the set of all arrivals matched by \opt\ on sample path $\omega$.
Now, our (LP free) linear system is as follows, 
\begin{eqnarray}
		\sum_{i}\theta_i +\sum_t \eod[\lambda_t(\omega,\nu)]&\leq &\beta\,  \alg,\label{cert1}\\
		\theta_i + \eod\Big[ \sum_{t\in O(\omega,i) }\lambda_t(\omega,\nu)\Big]& \geq &\alpha_i \, \opt_i 
		\quad \forall i\in I ,\label{cert2}\label{cert'2} 
\end{eqnarray}
here $\alpha_i, \beta\geq 0$ are constants that define the competitive ratio. 
Notice that constraints \eqref{cert2} depend directly on the actions of \opt. 
The following lemma states that these conditions are sufficient to certify approximate optimality. See Appendix  \ref{appx:certificate} for a proof. 
\begin{lemma}\label{certificate}
	Given an online algorithm \alg, non-negative values $\{\lambda_t(\omega,\nu)\}_{t,\omega,\nu}$ and $\{\theta_i\}_{i}$ such that conditions \eqref{cert1} and \eqref{cert2} hold, 
	we have
	\[\alg \geq \frac{\min_{i\in I} \alpha_i}{\beta}  \opt.\] 
\end{lemma}

{\color{black}
	\textbf{Remarks:}	
		In the classic primal-dual approach for analyzing online matching algorithms, the dual has a constraint for every edge in the graph (see Appendix \ref{primalimpossible}). In contrast, the LP free system has a constraint for every resource. In other words, the dual constraints corresponding to all the edges incident on a resource are replaced by a single constraint in the LP free system.  
		It can be verified that any feasible solution to the dual of the natural LP for reusable resources, provides a feasible solution to the linear system given by \eqref{cert1} and \eqref{cert'2},  but not vice-versa. In fact, in Appendix \ref{appx:lpfreedual}, 
		we show that the converse of Lemma \ref{certificate} is true, i.e., if $\alg\geq \frac{\min_{i\in I}\alpha_i}{\beta} \opt$, then there exists a feasible solution to the linear system. 
		In other words, \alg\ is $\frac{\min_{i\in I}\alpha_i}{\beta}$--competitive \emph{if and only if} our linear system is feasible. 
		This (stronger) statement is not true for standard primal-dual where, typically, the primal LP is a relaxation of the clairvoyant benchmark. In contrast, the LP free framework is not tied to the tightness of a relaxation and the increased flexibility makes it easier to find a feasible solution in this framework. We note that the LP free framework is a generalization of the path based analysis approach of \cite{stochrew}, that was designed for a specific setting of online matching with non-reusable resources.  
	}
	
	\section{$\left(1-\frac{1}{e}\right)$ Competitive Algorithm for Arbitrary Usage Distributions}\label{sec:matching}
	
	
	{\color{black} Let us recall the setting of online matching with reusable resources. We have a bipartite graph $G=(I,T,E)$, one side of which is a set of resources $i\in I$ with starting inventory/capacity $\{c_i\}_{i\in I}$. Vertices $t\in T$ arrive sequentially at times $\{a(t)\}_{t\in T}$. When vertex $t$ arrives, we see the edges incident on $t$ and must make an irrevocable decision to match $t$ to at most one resource with available capacity, without any knowledge about future arrivals. If we match $t$ to resource $i$, one unit of the resource is used by arrival $t$ for a random duration $d_{it}$ drawn independently from distribution $F_i$.  
		At time $a(t)+d_{it}$, this unit comes back to the system and is immediately available for a re-match. Distributions $F_i$ are known to the online algorithm but the exact duration of use is revealed only when the unit returns. A match between $t$ and $i$ earns reward $r_i$. The objective is to maximize the expected total reward.
		
		\textbf{\dpg\ algorithm: } For resource $i\in I$, consider an arbitrary ordering over its $c_i$ units. 
		Let indices $k\in[c_i]$ denote the ordering, with index $k=c_i$ representing the highest ranked unit. 
		When a customer is matched to $i$ we treat this as a match to the \emph{highest ranked} unit of $i$ available at the time that the customer arrives. This unit is then unavailable until the customer returns the resource. Let $z_i(t)$ denote the highest available unit of resource $i\in I$ just before the arrival of $t\in T$. Note that, $z_i(1)=c_i\,\, \forall i\in I$. Let $S_t$ denote the set of resources that have an edge to $t$ and that also have at least one unit available when $t$ arrives. Then, \dpg\ matches $t$ to 
		\begin{equation*}
			\qquad \argmax_{i\in S_t}\quad  r_i\,\left(1-g\left(\frac{z_i(t)}{c_i}\right)\right).
		\end{equation*}
		Choosing $g(x)=e^{-x}$ gives the best guarantee for a variety of usage distributions (see Section \ref{sec:rba}).
		As we discussed in Section \ref{sec:ourcontrib}, \dpg\ is hard to analyze for general usage distributions. 
		Now, we introduce the realization independent algorithm that we referred to in Theorem \ref{main} and show that 
		it is $(1-1/e)$--competitive for arbitrary usage distributions. 
	Our algorithm uses \dpg\ at its core but \emph{wraps it in a realization independent framework}. 
	{\color{black}	We discuss the generalization of this algorithm to online assortments and budgeted allocations in Section \ref{sec:mtoasst}, with more details included in Appendix \ref{appx:asst}.} 
}

We introduce the algorithms in two parts. The first part (Algorithm \ref{galg}), which is the main new object, is a relaxed online algorithm that is powered by \dpg\ but only subject to a \emph{fluid version of reusability} such that if a unit of some resource $i$ with usage distribution $F_i$ matched at time $\tau$, then $F_i(d)$ fraction of this unit returns by time $\tau+d$, for every $d\geq 0$. 
The output of this relaxed algorithm guides the final matching decision in a realization independent way (second part). 
We refer to the relaxed algorithm as \galg, since it serves as a guide.

\begin{algorithm}[h]
	\SetAlgoLined
	\textbf{Output:} Fractional matching given by values $x_{it}\in[0,1]$\;
	Let $g(t)=e^{-t}$, and 
	initialize $Y(k_i)=1$ for every $i\in I, k_i\in[c_i]$\;
	\For{every new arrival $t$}{
		For every $i\in I, k_i\in[c_i]$ and $t\geq 2$, update values 	\text{{\color{black}\tcp{Fluid update of inventory to capture the fractions of unit $k_i$ returning at $t$}}}\;
		\[ 
			Y(k_i)=Y(k_i)+\sum_{\tau=1}^{t-1} \Big(F_i\big(a(t)-a(\tau)\big)-F_i\big(a(t-1)-a(\tau)\big)\Big)y(k_i,\tau);\]	
			Initialize $S_t=\{i \mid (i,t)\in E\}$, values $\eta=0$, $y(k_i,t)=0$ and $x_{it}=0$ for all $i\in S_t,k_i\in [c_i]$\;
			\While{$\eta<1$ and $S_t\neq \emptyset$}{
				\For{$i\in S_t$}{
					\textbf{if}  $Y(k_i)=0$ for every $k_i\in[c_i]$ \textbf{ then} remove $i$ from $S_t$\;
					\textbf{else }$ z_i=\underset{k_i\in[c_i]}{\arg\max}\, \{k_i \mid Y(k_i)>0\}; \text{ {\color{black}\tcp{Highest\,\, available\,\, unit of $i$} }} $	
				}
				$ i(\eta)=\underset{i\in S_t }{\arg\max} \,\, r_i \Big(1-g\big(\frac{z_i}{c_i}\big)\Big);\, \text {{\color{black}\tcp{
							Rank based allocation}}}$\
				
				$y(z_{i(\eta)},t)=\min\{Y(z_{i(\eta)}),1-\eta\}$; \text{{\color{black} \tcp{Fractional\,\, match}}}\\
				Update $x_{i(\eta)t}\to x_{i(\eta)t}+y(z_{i(\eta)},t)$; \quad $\eta\to \eta+ y(z_{i(\eta)},t)$;\quad $Y(z_{i(\eta)})\to Y(z_{i(\eta)})- y(z_{i(\eta)},t) $\;
		}}	
		\caption{\galg}
		\label{galg}
	\end{algorithm}
	\begin{algorithm}[h]
		\SetAlgoLined
		Initialize inventory $y_i(0)=c_i$ and values $\delta_i=\sqrt{\frac{\log c_i}{c_i}}$ for $i\in I$\;
		Start a parallel (online) execution of \galg\;
		\For{every new arrival $t$}{
			For every $i\in I$ and $t\geq 2$, compute inventory 
			\[y_i(t)=y_i(t-1)+ \text{\# units of $i$ returned at $t$}; \]
			\text{{\color{black}\tcp{\# units returned depends on realization of usage durations}}}\;
			Read $(x_{it})_{i\in I}$ from the parallel execution of \galg\;
			Independently sample a resource from $I$ with probabilities $p_i=\frac{1}{1+\delta_i}x_{it}$ $\forall i\in I$\;\, \text{{\color{black} \tcp{Probability that no resource is sampled may be non-zero}}}\\
			Match $t$ to $i$ if $y_{i}(t)>0$ and update remaining capacity\;
			
		}
		\caption{\salg}
		\label{alg}
	\end{algorithm}
\smallskip
	
	\noindent \textbf{Description of the algorithm:} In \galg, we use $k_i$ to refer to unit $k$ of resource $i$. The variable $Y(k_i)$ keeps track of the fraction of unit $k_i$ available. 
	Each arrival may be matched to numerous units of different resources and the variable $y(k_i,t)$ is the fraction of unit $k_i$ matched to arrival $t$. {\color{black} Before matching $t$, $Y(k_i)$ is updated based on a fluid view of reusability whereby, if fraction $y(k_i,\tau)$ of $k_i$ is matched to $\tau$ ($\leq t$) then fraction $\big(F_i\big(a(t)-a(\tau)\big)-F_i\big(a(t-1)-a(\tau)\big)\big)y(k_i,\tau)$ returns after the arrival of $t-1$ and before arrival of $t$. After updating inventory, \galg\ uses \dpg\ to perform its fractional matching.   After matching $t$, the inventory is updated again and $Y(k_i)$ is reduced by $y(k_i,t)$ (which may be 0). 
		
		The values $x_{it}$ represent the total fraction of resource $i$ matched to arrival $t$. Notice that,
		\[\sum_{i\in I} x_{it}\leq 1\; \qquad x_{it}\geq0\quad \forall i\in I.\]
		At arrival $t$, \salg\ uses the fractional matching $(x_{it})_{i\in I}$ of \galg\ to make integral matches by independently sampling a resource from $I$ with the probability distribution $(\frac{1}{1+\delta_i}x_{it})_{i\in I}$. If the sampled resource is unavailable then $t$ departs unmatched. The value $\delta_i=\sqrt{\frac{\log c_i}{c_i}}$ is set so that the probability of sampling $i$ is reduced by a small amount that is sufficient to ensure that \galg\ and \salg\ have similar expected revenue from matching $i$.  Our overall algorithm consists of running and updating the states of both \galg\ and \salg\ in parallel. The algorithm is realization independent as allocation decisions do not dependent on realized usage durations.
		
		Overloading notation, let \galg\ denote the total revenue of \galg\ in a hypothetical scenario where (i) resources are not subject to true (stochastic) reusability but only to the (deterministic) fluid version of reusability and (ii) arrivals can be matched fractionally. Clearly, \galg\ is allowed more power than standard online algorithms (such as \dpg\ and Balance).}  The next lemma states that with this additional power, \galg\ achieves the best possible asymptotic guarantee against clairvoyant (\opt).
	
	\begin{lemma}[Performance of \galg] \label{galgua}\label{galgvopt} 
		For every instance of the vertex-weighted online matching problem with reusable resources we have, 
		\[\galg \geq\,  
		e^{-\frac{1}{c_{\min}}}\, (1-1/e)\,  \opt.\] 
	\end{lemma}
	{\color{black}This is the main technical ingredient for proving Theorem \ref{main}, where we also use concentration bounds to show that the expected performance of \salg\ is asymptotically the same as that of \galg. }
		\smallskip
		
		\noindent \textbf{The role of \dpg:} Recall, \galg\ matches fractionally by using \dpg. Indeed, even with a fluid version of reusability, the ability to measure ``effective" capacity by means of inducing asymmetry between units is key to obtaining our results. For instance, if we use the greedy rule in \galg, we 
		get a guarantee of $1/2$ (same as standard greedy). Using Balance in \galg\ has a competitive ratio of at most 0.626, same as standard Balance (see Appendix \ref{appx:twopoint}). 
		\medskip
		
		\noindent \textbf{The role of fluid guide:} The idea of using a fluid-relaxed online algorithm as a guide gives us an approach to convert realization dependent online algorithms such as \dpg, into realization independent ones. 
		This leads to a provably good general algorithm at a higher computational cost. In more structured instances, namely online matching with IFR usage distributions (see Section \ref{sec:rba}), we also analyze the performance of the faster and distribution oblivious \dpg\ policy. The analysis is more challenging but we show that in this special case \dpg\ has the same asymptotic competitive ratio guarantee as \salg. 
		
		Note that the idea of using a fractional solution to guide online decisions is commonly used 
		in online matching with stochastic arrivals \citep{devsivec,alaei,devsiv}, where 
		the fractional solution often comes from solving an LP and encodes 
		information about the optimal matching. 
		In our setting, the relaxed solution is generated by 
		an online algorithm which encounters fluid reusability 
		and this serves to ``smooth out" the effects of uncertainty in usage durations. In a broader context, the idea of maintaining and rounding an online fractional solution has enjoyed successful applications in several fundamental online decision making problems including (but not limited to) the online $k$-server problem \citep{bansal}, online covering problem \citep{alon}, and online paging \citep{blum}. %
		\smallskip
		
		\noindent \textbf{Outline of the analysis:} To prove Theorem \ref{main}, the key step is to establish Lemma \ref{galgua} by finding a feasible solution for the linear system introduced in Section \ref{sec:certificate}. We start by proposing a candidate solution  for the system. 
		In order to prove feasibility, we transform the main parts of the analysis into statements involving an explicitly defined random process that is introduced in the next section. This random process is crucial to the analysis of all algorithms in the paper -- \galg\ and its generalization, as well as, 
	 \dpg -- and may be of independent interest. Given Lemma \ref{galgua}, we show the claimed performance guarantee for \salg\ by (carefully) using standard concentration bounds. 
		\subsection{The Random Process Viewpoint} 
We start by defining a random process that 
is free of the complexities of any online or offline algorithm and serves as an intermediate object that allows a comparison between online and offline algorithms throughout this paper. After defining the process, we give some insight into its usage and discuss useful properties of the process.
\medskip

			\noindent	\textbf{$(F,\bsigma,\mb{p})$ random process:} Consider an ordered set of points $\bsigma=\{\sigma_1,\cdots,\sigma_{T}\}$ on the positive real line such that, $0<\sigma_1< \sigma_2< \cdots< \sigma_T$. 
			Let $\mb{p}=(p_1,\cdots,p_T)$ denote a vector of probabilities. 
			We are given a single unit of a resource that is in one of two states at every point in time: \emph{free/available} or \emph{in-use/unavailable}. 
			
			\paragraph{}The unit starts at time $0$ in the available state. The state of the unit evolves with time as follows. At any point $\sigma_t\in \bsigma$ where the unit is available, with probability $p_t$ the unit becomes in-use for an independently drawn random duration $d\sim F$. The unit stays in-use during $(\sigma_t,\sigma_t+d)$ and switches back to being available at time $\sigma_t+d$. Each time the unit switches from available to in-use we earn a unit reward.
\medskip

Observe that the process above captures a setting of the reusable resource problem where we have a single reusable item with unit reward and a sequence of arrivals with an edge to this item. We also have the aspect of stochastic rewards, as each arrival comes with an independent probability that the match would succeed. It is not hard to show that the best strategy to maximize total reward in this simple instance is simply to try to match the unit whenever possible, regardless of the arrival sequence or the probabilities and usage distributions. In a general instance of the problem an algorithm must choose between various resources to match each arrival and decisions can depend on past actions and realizations. However, working at the level of individual units of resources, we show that the total reward of \galg\ from an individual unit can be captured through a suitably defined $(F,\bsigma,\mb{ p})$ process. Similarly, for every unit we find a random process that provides a tractable upper bound on the number of times \opt\ allocates the unit. A key piece of the analysis is then converted into statements that compare the expected reward in one random process with another. 
Next, consider a fluid version of the random process.
\medskip

			\noindent \textbf{Fluid $(F,\bsigma,\mb{p})$ process:} We are given a single unit of a resource and ordered set of points/arrivals $\bsigma=\{\sigma_1,\cdots,\sigma_{T}\}$ on the positive real line such that, $0<\sigma_1< \sigma_2< \cdots< \sigma_T$. We also have a sequence of fractions $\mb{p}=(p_1,\cdots,p_T)$. The resource is fractionally consumed at each point in $\bsigma$ as follows.
			
			\paragraph{} If $\delta_t$ fraction of the resource is available when $\sigma_t\in \bsigma$ arrives, then $p_t\delta_t$ fraction of the resource is consumed by $\sigma_t$, generating reward $p_t\delta_t$. The $p_t\delta_t$ fraction consumed at $\sigma_t$ returns fluidly in the future according to the distribution $F$, i.e., exactly $F(d)$ fraction of $p_t\delta_t$ is available again by time $\sigma_t+d$, for every $d\geq 0$. 
			
	\smallskip
	

	\subsubsection{Useful Properties of $(F,\bsigma,\mb{p})$ Random Process.}
	We start with a crucial monotonicity property (Lemma \ref{monotone}) of the random process. The subsequent property (Lemma \ref{zeroset}) states that adding ``zero probability" points to the set $\bsigma$ does not change the process. The final property (Lemma \ref{equiv}) describes the equivalence between a fluid $(F,\bsigma,\mb{p})$  process and its random counterpart. Proofs of these properties are included in Appendix \ref{appx:randomprocess}. 
	
	Given an $(F,\bsigma,\mb{p})$ process, we refer to the points in $\bsigma$ as arrival times and arrivals interchangeably. Let $r(F,\bsigma,\mb{p})$ denote the total expected reward of the random process.  Let $\mb{1}_{\bsigma}$ denote the sequence of probabilities $p_t=1$ for every $\sigma_t\in\bsigma$. We also use the succinct notation $(F,\bsigma)$ to refer to $(F,\bsigma,\mb{1}_{\bsigma})$. Similarly, $r(F,\bsigma):= r(F,\bsigma,\mb{1}_{\bsigma})$. 
	Given two probability sequences $\mb{p}_1$ and $\mb{p}_2$ for the same set of arrivals, we use $\mb{p}_1\vee \mb{p}_2$ to denote their point-wise maximum. 
	
	\begin{lemma}[{Monotonicity Property}]\label{monotone}
		Given a distribution $F$, arrival set $\bsigma=\{\sigma_1,\cdots,\sigma_T\}$, and probability sequences $\mb{p}_1=(p_{11},\cdots,p_{1T})$ and $\mb{p}_2=(p_{21},\cdots,p_{2T})$ such that, $p_{1t}\leq p_{2t}$ for every $t\in[T]$, we have,
		\[r(F_i,\bsigma,\mb{p}_1)\leq r(F_i,\bsigma,\mb{p}_2) .\]
	\end{lemma}
	\begin{lemma}[{Adding Zero Probability Points}]\label{zeroset}
	Given an $(F,\bsigma,\mb{p})$ random process, let $\bsigma' \subset \bsigma$ be a subset of arrivals where the resource is unavailable with probability 1. 
Then, 
at every arrival $\sigma_t\in \bsigma$, the probability that the resource is available at $\sigma_t$ is identical in $(F,\bsigma,\mb{p})$ and $(F,\bsigma,\mb{p}\vee \mb{1}_{\bsigma'})$. 
	\end{lemma}
\begin{lemma}[{Equivalence}]\label{equiv}
		
		The following statements are true for every $(F,\bsigma,\mb{p})$ random process:
		\begin{enumerate}[(i)]
			\item For every $\sigma_t\in \bsigma$, the probability that the resource is available at $\sigma_t$ equals the fraction of the resource available at $\sigma_t$ in the fluid $(F,\bsigma,\mb{p})$  process.
			\item The expected reward $r(F,\bsigma,\mb{p})$, equals the total reward in the fluid $(F,\bsigma,\mb{p})$ process.
		\end{enumerate} 

\end{lemma}

\subsection{Analysis of \galg\ and \salg} \label{sec:solution}
We now show that \galg, the relaxed online algorithm which guides \salg\ but encounters a fluid form of reusability, is asymptotically $(1-1/e)$--competitive (Lemma \ref{galgua}). Recall that we use \galg\ and \opt\ to also denote the expected rewards of the respective algorithms. 
Our proof uses the LP free framework and we start by defining a candidate solution for the system. 

\noindent \textbf{Candidate solution:} We refer to unit $k$ of resource $i$ as $k_i$. 
Recall that $y(k_i,t)$ is the fraction of unit $k_i$ matched to $t$ in \galg. Inspired by \cite{devanur}, we propose the following candidate solution for the generalized certificate, 
\begin{eqnarray}
\lambda_t(\omega,\nu)&&=\sum_{i\in I} r_i\sum_{k_i\in[c_i]} y(k_i,t) \bigg(1-g\Big(\frac{k_i}{c_i}\Big)\bigg),\label{lambda2}\\
\theta_i&&= c_i\left(e^{\frac{1}{c_i}}-1\right)r_i\sum_{t} \sum_{k_i\in[c_i]} y(k_i,t) g\Big(\frac{k_i}{c_i}\Big) \label{theta2},
\end{eqnarray}
{\color{black}	here $\omega$ and $\nu$ represent the sample path of usage durations in offline (\opt) and online (\galg) respectively. Notice that $\lambda_t(\omega,\nu)$ is deterministic and independent of both $\omega$ and $\nu$. From here onward, we use $\lambda_t$ to denote $\lambda_t(\omega,\nu)$. We have set $\lambda_t$ equal to the total reduced price based reward from matching arrival $t$ in \galg. 
$\theta_i$ is defined such that the sum $\sum_t \lambda_t +\sum_i \theta_i \left(c_i\left(e^{\frac{1}{c_i}}-1\right)\right)^{-1}$, is the total revenue of $\galg$. Note that, when $\lambda_t$ is deterministic the simplified LP free certificate 
is  as follows, 
\begin{eqnarray}
\theta_i + \eo\Big[ \sum_{t\in O(\omega,i) }\lambda_t\Big] &\geq&\alpha_i \opt_i \quad \forall i\in I, \label{scert1}\\
\sum_{t\in T}\lambda_t + \sum_{i\in I} \theta_i &\leq &\beta \, \galg. \label{scert2}
\end{eqnarray}	

To prove that our candidate solution satisfies \eqref{scert1} and \eqref{scert2}, we need some additional notation. We often need to refer to the state of \galg\ just after some arrival $t$ departs the system. In order to avoid any inconsistencies when making such a reference, such as in a boundary case where $F_i(0)=0$, we formally consider the following sequence of events at every moment $a(t)$ when an arrival occurs,
\begin{enumerate}[(i)]
\item Units that were in use and set to return at $a(t)$ are made available and inventory is updated.
\item Arrival $t$ is (fractionally) matched and inventory of matched units is reduced accordingly.
\end{enumerate}   We refer to the state of the system after 
step $(ii)$ is performed as the state at $t^+$.
Correspondingly, for each arrival $t$ with an edge to $i$ and every unit $k_i$, we define (deterministic) indicator, 
\[\onee(\neg k_i,t^+)=\begin{cases}
1 \quad \text{ if no fraction of $k_i$ is available at $t^+$,}\\
0 \quad \text{ otherwise}.
\end{cases} \]
Recall that $g(x)=e^{-x}$. Let \[\Delta g(k_i):= g\Big(\frac{k_i-1}{c_i}\Big)-g\Big(\frac{k_i}{c_i}\Big)\qquad \forall k_i\in[c_i].\]
The following properties of $\Delta g(\cdot)$ will be quite useful.
\begin{eqnarray}
\Delta g(k_i)&= & g\Big(\frac{k}{c_i}\Big)(e^{\frac{1}{c_i}}-1)\geq 0 \qquad \forall k_i\in[c_i],\label{delg1}\\
1-g\Big(\frac{k_i}{c_i}\Big)&= &1-g(1)-\sum_{k> k_i} \Delta g(k)\qquad \forall k_i\in[c_i].\label{delg2}
\end{eqnarray}
Now, we show that the candidate solution satisfies \eqref{scert2} with $\beta$ close to 1 for large starting inventory.
\begin{lemma}\label{pscert2} For the candidate solution given by \eqref{lambda2} and \eqref{theta2}, we have,
\[\sum_{t\in T} \lambda_t +\sum_{i\in I} \theta_i \leq  \,e^{\frac{1}{c_{\min}}}\,\galg.\]
\end{lemma}
\begin{proof}{Proof.}
Recall that $\theta_i$ is defined such that,
\begin{equation}\label{thetdef}
\sum_t \lambda_t +\sum_i \theta_i \left(c_i\left(e^{\frac{1}{c_i}}-1\right)\right)^{-1}=\galg.
\end{equation}
Let $\beta= \max_i  \left(c_i\left(e^{\frac{1}{c_i}}-1\right)\right)\, \forall i\in I$. Using the fact that $e^x\geq 1+x\,\, \forall x\in\mathbb{R}$, we have, $\beta\geq 1$. Multiplying both sides of inequality \eqref{thetdef} by $\beta$ we get,
\[\sum_t \lambda_t +\sum_i \theta_i\leq \beta\left[ \sum_t \lambda_t +\sum_i \theta_i \left(c_i\left(e^{\frac{1}{c_i}}-1\right)\right)^{-1} \right]\leq \beta\, \galg.\]
It remains to show that $\beta\leq e^{\frac{1}{c_{\min}}}$, where $c_{\min}\geq 1$. This follows from the following chain of inequalities, \[c_i(e^{\frac{1}{c_i}}-1)\leq  c_i(1+\frac{1}{c_i}+\frac{1}{c_i^2}-1)= 1+\frac{1}{c_i}\leq e^{\frac{1}{c_i}} \leq e^{\frac{1}{c_{\min}}} \quad \forall c_i\geq 1.\]
\hfill\Halmos	\end{proof}
It remains to show that the candidate solution satisfies \eqref{scert1} with $\alpha_i=(1-1/e)$. We show this in two parts. First, in Lemma \ref{decompose}, we lower bound $\eo\Big[\sum_{t\in O(\omega,i)} \lambda_t\Big]$. Then in Lemma \ref{thetarel}, we lower bound $\theta_i$. Recall that, $\onee(\neg k_i,t^+)$ equals 1 if
if no fraction of $k_i$ is available at $t^+$, and equals 0 otherwise.
\begin{lemma}\label{decompose}
For the candidate solution given by \eqref{lambda2}, we have,
\[	\eo\Big[\sum_{t\in O(\omega,i)} \lambda_t\Big]\geq  (1-1/e)\opt_i-r_i\eo\Big[\sum_{t\in O(\omega,i)}\sum_{k_i\in[c_i]} \Delta g (k_i) \onee(\neg k_i,t^+)\Big].\]

\end{lemma}
\begin{proof}{Proof.}
Given an edge $(i,t)$, let $z_i(t^+)$ denote the highest available unit of $i$ that has non-zero fraction available at $t^+$. Let $z_i(t^+)=0$ if no fraction of any unit of $i$ is available at $t^+$. We have, $\onee(\neg k_i,t^+)=1$ for all $k_i>z_i(t^+)$ and $\onee(\neg z_i(t^+),t^+)=0$. We have, 
\begin{eqnarray*}
\lambda_t  &=&\sum_{i\in I} r_i\sum_{k_i\in[c_i]} y(k_i,t) \bigg(1-g\Big(\frac{k_i}{c_i}\Big)\bigg),\\
&\geq &r_i\Big(1-g\Big(\frac{z_i(t^+)}{c_i}\Big)\Big),\\
&= &r_i \Big(1-g(1)-\sum_{k_i> z_i(t^+)} \Delta g(k_i)\onee(\neg k_i,t^+)\Big),\\	
&\geq &r_i\Big(1-1/e-\sum_{k_i\in[c_i]} \Delta g (k_i) \onee(\neg k_i,t^+)\Big).
\end{eqnarray*} 
The first inequality uses the fact that if $y(k_i,t)>0$ then $k_i\geq z_i(t^+)$ and $1-g\big(\frac{k}{c_i}\big)\geq 1- g\big(\frac{z_i(t^+)}{c_i}\big)$. The second equality follows from \eqref{delg2}. The second inequality follows from the non-negativity of $\Delta g(k_i)$ (see \eqref{delg1}) and the fact that $g(1)=1/e$.	Thus,
\begin{eqnarray*}
\eo\Big[\sum_{t\in O(\omega,i)} \lambda_t\Big]&\geq &r_i \eo\Big[\sum_{t\in O(\omega,i)} \Big(1-1/e-\sum_{k_i\in[c_i]} \Delta g (k_i) \onee(\neg k_i,t^+)\Big)\Big],\\ 
&= & (1-1/e) \eo\Big[\sum_{t\in O(\omega,i)} r_i\Big] - r_i \eo\Big[\sum_{t\in O(\omega,i)}\sum_{k_i\in[c_i]} \Delta g (k_i) \onee(\neg k_i,t^+)\Big],
\end{eqnarray*}
here the equality follows by linearity of expectation. Finally, observe that, $\opt_i=\eo\Big[\sum_{t\in O(\omega,i)} r_i\Big]$.

\hfill\Halmos\end{proof}
\begin{lemma}\label{thetarel} For the candidate solution given by \eqref{theta2}, we have,
\begin{eqnarray}		
\frac{1}{r_i}\theta_i\geq \eo\Big[\sum_{t\in O(\omega,i)}\sum_{k_i\in[c_i]} \Delta g (k_i) \onee(\neg k_i,t^+)\Big] ,\label{interim}
\end{eqnarray}
\end{lemma}}
\begin{proof}{Proof.} 
{\color{black}	Fix an arbitrary resource $i$.  For brevity, let $k$ denote unit $k_i$. Let $k_O$ be some unit of $i$ in \opt ($k_O$ need not be the same as $k$). Given a sample path $\omega$ of usage durations in \opt, let $O(\omega,k_O)$ denote the set of arrivals matched to unit $k_O$ of $i$ in \opt.
The road map for the proof is as follows: First, we convert each side of the inequality into the expected reward of a suitable random process. This turns inequality \eqref{interim} into a crisp statement that only involves comparison between two different random processes. Then, we use Lemma \ref{monotone} through Lemma \ref{equiv} to compare the resulting random processes. Specifically, we claim that for every unit $k\in[c_i]$, there exists an $(F_i,\bsigma_1(k),\mb{p}_1(k))$ random process and an $(F_i,\bsigma_2(k),\mb{p}_2(k))$ fluid process such that,
\begin{eqnarray}
r(F_i,\bsigma_1(k),\mb{p}_1(k))&\geq &\eo\Big[\sum_{t\in O(\omega,i)} \Delta \onee(\neg k,t^+)\Big]\label{comp1},\\
\frac{1}{r_i}\theta_i&=& 
c_i(e^{\frac{1}{c_i}}-1)\,\sum_{k\in[c_i]} g\Big(\frac{k}{c_i}\Big)r(F_i,\bsigma_2(k),\mb{p}_2(k)),\label{comp2}\\
r(F_i,\bsigma_2(k),\mb{p}_2(k))&\geq & r(F_i,\bsigma_1(k),\mb{p}_1(k)).\label{comp3}
\end{eqnarray}
Then, 
combining \eqref{comp1}-\eqref{comp3} and using linearity of expectation we get, 
\begin{eqnarray*}
\eo\Big[\sum_{t\in O(\omega,i)}\sum_{k\in[c_i]} \Delta g (k) \onee(\neg k,t^+)\Big]&&=  \sum_{k_O \in [c_i]} \eo\Big[\sum_{t\in O(\omega,k_O)}\sum_{k\in[c_i]} g\Big(\frac{k}{c_i}\Big)(e^{\frac{1}{c_i}} -1) \onee(\neg k,t^+)\Big],\\
&&\leq \sum_{k_O \in [c_i]}\sum_{k \in [c_i]}g\Big(\frac{k}{c_i}\Big)(e^{\frac{1}{c_i}} -1) \, r\big(F_i,\bsigma_1(k),\mb{p}_1(k)\big),\\
&&\leq \sum_{k_O \in [c_i]}\sum_{k \in [c_i]}g\Big(\frac{k}{c_i}\Big)(e^{\frac{1}{c_i}} -1) \, r\big(F_i,\bsigma_{2}(k),\mb{p}_2(k)\big),\\
&&= c_i\sum_{k \in [c_i]} g\Big(\frac{k}{c_i}\Big)(e^{\frac{1}{c_i}} -1)\, r\big(F_i,\bsigma_{2}(k),\mb{p}_2(k)\big),\\
&& = \frac{1}{r_i} \theta_i. 
\end{eqnarray*}
}
{\color{black}\noindent \emph{Proof of \eqref{comp1}:}	 Let $\mb{\sigma}_1(k)$ denote the ordered set of all arrival times} $a(t)$ such that $\onee(\neg k,t^+)=1$ in \galg, i.e., unit $k$ is not available after arrival $t$ is matched. Arrival times in $\mb{\sigma}_1(k)$ are ordered in ascending order. For convenience, we 
often refer to arrival times simply as arrivals. Recall that $r(F,\bsigma)$ denotes the expected reward of an $(F,\bsigma,\mb{1}_{\bsigma})$ random process. 
Recall that $O(\omega,k_O)$ is the set of arrivals matched to unit $k_O$ of $i$ in \opt. To prove \eqref{comp1}, we show that,
\begin{equation} 
\eo\Big[\sum_{t\in O(\omega, k_O)} \onee(\neg k,t^+)\Big]=\eo\Big[\big| O(\omega, k_O)\cap \mb{\sigma}_1(k)\big| \Big] \leq r(F_i,\mb{\sigma}_1(k)).\label{interim2}
\end{equation}
The equality follows by definition. To prove the inequality in \eqref{interim2}, arbitrarily fix all the randomness in \opt\ except the usage durations of unit $k_O$. Consider the following coupling between usage durations of $k_O$ and the $(F_i,\mb{\sigma}_1(k))$ random process: We start by generating two distinct lists of independent samples from distribution $F_i$. We make a copy of List 1 called  List $1_O$. List 1 is used to obtain samples for the $(F_i,\mb{\sigma}_1(k))$ random process. 
When a duration is required by the random process, we draw the first sample from List 1 that has not previously been provided to the random process. On the other hand, when $k_O$ is matched in \opt\ we draw usage duration from List $1_O$ or List 2, depending on the arrival that $k_O$ is matched with. If the arrival is in $\mb{\sigma}_1(k)$, we draw the first unused sample from List $1_O$. For all other arrivals, we draw the usage duration from List 2 by picking the first unused sample. 
Observe that a sample is drawn from List $1_O$ only if its counterpart in List 1 is drawn. Thus, the number of samples drawn from List $1_O$ is upper bounded by the number of samples drawn from List 1 and inequality \eqref{interim2} follows. 
\smallskip

{\color{black}	\noindent \emph{Proof of \eqref{comp2}:} Next, we compare $r(F_i,\mb{\sigma}_1(k))$ with the expected reward of \galg\ from matching $k$. }To make this connection we interpret the actions of \galg\ through a fluid process. Let $\mb{T}$ be the ordered set of all arrivals 
(arrival times, to be precise). 
For every $t\in \mb{T}$, we define probability $p(k,t)$ as follows: If no fraction of $k$ is matched to $t$ we set $p(k,t)=0$. If some non-zero fraction $\eta(k,t)$ of $k$ is available for match to $t$ and fraction $y(k,t) (\leq \eta(k,t))$ is actually matched to $t$, we set $p(k,t)=\frac{y(k,t)}{\eta(k,t)}$.  
Let $\mb{p}(k)$ denote the ordered sequence of probabilities $p(k,t)$ corresponding to arrivals in $\mb{T}$. Consider the $(F_i,\mb{T},\mb{p}(k))$ random process and its fluid version. By definition of $\mb{p}(k)$, the fluid $(F_i,\mb{T},\mb{p}(k))$ process corresponds exactly to the matching of unit $k$ in \galg. 
From the equivalence between the fluid process and its random counterpart (Lemma \ref{equiv}), we have, 
\[ r\big(F_i,\mb{T},\mb{p}(k)\big) = \sum_t y(k_i,t).  \]
By definition of $\theta_i$, 
\[	 \frac{1}{r_i} \theta_i=c_i (e^{\frac{1}{c_i}}-1)\sum_{k \in [c_i]} \Big[g\Big(\frac{k}{c_i}\Big) \sum_t y_i(k,t)\Big]=	 c_i(e^{\frac{1}{c_i}}-1)\,\sum_{k\in[c_i]} g\Big(\frac{k}{c_i}\Big)r(F_i,\mb{T},\mb{p}(k)).\]
\smallskip

{\color{black}	\noindent \emph{Proof of \eqref{comp3}:} Let $\bsigma_0(k)=\{t\mid p(k,t)=1\}$ denote the set of all arrivals 
where some nonzero fraction of $k$ is available and \emph{fully} matched in \galg.} Therefore, $\onee(\neg k,t^+)=1$ for every $t\in\bsigma_0(k)$. Recall that $\bsigma_1(k)=\{t \mid \onee(\neg k,t^+)=1\}$. Therefore, $\bsigma_0(k)\subseteq \mb{\sigma}_1(k)$. We claim that the fraction of unit $k$ available in \galg\ at arrivals in $\mb{\sigma}_1(k)\backslash \bsigma_0(k)$ is zero. The claim essentially follows by definition; recall that $\onee(\neg k,t^+)=1$ for every $t\in \mb{\sigma}_1(k)$, so if a 
non-zero fraction of $k$ is available at $t\in\mb{\sigma}_1(k)$ then $t$ must be in $\bsigma_0(k)$. So at every arrival in $\mb{\sigma}_1(k)\backslash \bsigma_0(k)$, no fraction of $k$ is available to match. 
Then, from Lemma \ref{equiv} we have that in the $(F_i,\mb{T},\mb{p}(k))$ random process, the probability of matching to any arrival $t\in\mb{\sigma}_1(k)\backslash \bsigma_0(k)$ is zero. Now, consider the augmented probability sequence $\mb{p}(k)\vee \mb{1}_{\mb{\sigma}_1(k)\backslash \bsigma_0(k)}$, which denotes that we set a probability of one for every arrival in $\mb{\sigma}_1(k)\backslash \bsigma_0(k)$. Notice that,
\[\mb{p}(k)\vee \mb{1}_{\mb{\sigma}_1(k)}=\mb{p}(k)\vee \mb{1}_{\mb{\sigma}_1(k)\backslash \bsigma_0(k)}.\]
Applying Lemma \ref{zeroset}, 
\[r\big(F_i,\mb{T},\mb{p}(k)\vee \mb{1}_{\mb{\sigma}_1(k)}\big)=r\big(F_i, \mb{T},\mb{p}(k)\big).\]
Finally, using the monotonicity Lemma \ref{monotone}, we have,
\[ r\big(F_i,\mb{\sigma}_1(k)\big)\leq r\big(F_i,\mb{T},\mb{p}(k)\vee \mb{1}_{\mb{\sigma}_1(k)}\big)= r\big(F_i,\mb{T},\mb{p}(k)\big),\] 

\hfill\Halmos\end{proof}
\begin{proof}{Proof of Lemma \ref{galgua}.}
Combining Lemma \ref{decompose} and Lemma \ref{thetarel} we have that the candidate solution given by \eqref{lambda2} and \eqref{theta2} satisfies,
\[ \theta_i + \eo\Big[ \sum_{t\in O(\omega,i) }\lambda_t\Big] \geq(1-1/e) \opt_i \quad \forall i\in I.\]
Thus, \eqref{scert1} is satisfied with $\alpha_i=(1-1/e)$. From Lemma \ref{pscert2}, we have that \eqref{scert2} is satisfied with $\beta=e^{\frac{1}{c_{\min}}}$. This establishes the feasibility of the candidate solution and a direct application of Lemma \ref{certificate} gives us that $\galg \geq e^{-\frac{1}{c_{\min}}} (1-1/e) \opt$, as desired. 
\hfill\Halmos
\end{proof}

{\color{black}Now, we show that the performance of \salg\ is close to that of \galg. When $t$ arrives, we first let \galg\ fractionally match $t$. Then 
\salg\ independently rounds the fractional matching. In addition to making integral decisions, \salg\ experiences true stochasticity in usage.  Let $x_{it}$ denote the fraction of arrival $t$ matched to resource $i$ by \galg. Let $\delta_i=\sqrt{\frac{\log c_i}{c_i}}$. \salg\ randomly selects at most one resource and matches $t$ to this resource subject to availability of the resource. In particular, resource $i$ is selected w.p.\ $\frac{1}{1+\delta_i}x_{it}$, independent of the availability of resources. Observe that,
\[\sum_{i\in I}\frac{x_{it}}{1+\delta_i}\leq 1,\]
so this is a valid distribution over resources. 
If the sampled resource is unavailable, then $t$ is left unmatched. In the next lemma, we show that at any given arrival, the sampled resource is available with high probability (w.h.p.). }
\begin{lemma}\label{algvgalg}
Let $\salg_i$ denote the expected reward generated from matching resource $i$ in \salg. We have, $\salg_i \geq \frac{1-c^{-0.5}_i}{1+\delta_i}\, \galg_i\,\, \forall i\in I$.
\end{lemma}
\begin{proof}{Proof.}
The proof rests on showing that for every $(i,t)\in E$, $i$ is available at $t$ with probability (w.p.) at least $1-c^{-0.5}_i$. This implies a lower bound of $\frac{1-c^{-0.5}_i}{1+\delta_i}x_{it}$ on the expected reward from matching $i$ to $t$, completing the proof. To prove the claim, for any edge $(i,t)\in E$, let $\onee( i,t)$ indicate the event that some unit of $i$ is available in \salg\ when $t$ arrives. {\color{black} Let $\onee(i\to t)$ indicate the event that \salg\ selects $i$ to match to $t$. The occurrence of this event is independent of the availability of resources and
\[\mathbb{E}[\onee(i\to t)]=\frac{1}{1+\delta_i}x_{it},\]
here the expectation is w.r.t.\ the independent rounding that occurs at $t$.} W.l.o.g., we independently pre-sample usage durations for every possible match and let $\onee(d_t> a(\tau)-a(t))$ indicate that the duration of usage pre-sampled for (a potential) match of $i$ to arrival $t$ is at least $a(\tau)-a(t)$. Now, the event that a unit of $i$ is available when $t$ arrives is equivalent to the following event,
\[ \sum_{\tau=1}^{t-1} \onee (i,\tau) \onee(i\to \tau)\onee(d_{\tau}> a(t)-a(\tau)) <c_i.\]
The probability that this event occurs is lower bounded by the probability of the following event occurring,
\[\sum_{\tau=1}^{t-1} \onee(i\to \tau)\onee(d_{\tau}> a(t)-a(\tau)) < c_i. \]
Define Bernoulli random variables $X_{\tau}= \onee(i\to \tau)\onee(d_{\tau}> a(t)-a(\tau))$ for all $\tau\leq t-1$. Random variables $X_{\tau}$ are independent of each other as both the randomized rounding and the durations of usage are independent across time. Further, the total expectation is upper bounded as follows, 
\[\mu:=\mathbb{E}\Big[\sum_{\tau=1}^{t-1}X_{\tau}\Big]=\frac{1}{1+\delta_i} \sum_{\tau=1}^{t-1} x_{i\tau} \big(1-F_i( a(t)-a(\tau))\big)\leq\frac{c_i}{1+\delta_i}.\]
{\color{black}  To see this, observe that in \galg, $x_{i\tau}$ is the fraction of $i$ matched to arrival $\tau$ and $\big(1-F_i( a(t)-a(\tau))\big)$ is the fraction of $x_{i\tau}$ that has not been returned by time $a(t)$. Thus, $\sum_{\tau=1}^{t-1} x_{i\tau} \big(1-F_i( a(t)-a(\tau))\big)$ is the total amount of resource $i$ that is unavailable (or in use) in \galg\ when $t$ arrives, and the total (fractional) inventory of $i$ used in \galg\ at any time is at most $c_i$. 
To complete the proof, we apply the Chernoff bound stated in Lemma \ref{chernoff}.} 
\hfill	\Halmos\end{proof}

\begin{lemma}[{From Multiplicative Chernoff}]\label{chernoff}
Given integer $\tau>0$, independent indicator random variables $\onee(t)$ for $t\in [\tau]$ and $\delta=\sqrt{\frac{\log c}{c}}$ such that, $\sum_{t=1}^{\tau} \mathbb{E}\left[\onee(t)\right]\leq \frac{c}{1+\delta}.$
We have,
\[\mathbb{P}\left(\sum_{t=1}^{\tau}\onee(t) \geq c \right)\leq \frac{1}{\sqrt{c}}. \]
\end{lemma}
A formal proof Lemma \ref{chernoff} is included in Appendix \ref{appx:chernoff}.
We believe that a tighter $ 1-\frac{1}{\sqrt{2\pi c_{\min}}}$ factor may be shown by generalizing the balls and bins type argument used in \cite{devsivec} and \cite{alaei}. The original arguments are in the context of non-reusable resources but a balls and bins setting where bin capacity is reusable over time could be an interesting object to study in future work.
\begin{proof}{Proof of Theorem \ref{main}.}
From Lemma \ref{galgvopt} and Lemma \ref{algvgalg} we have that \salg\ is at least,
\[(1-1/e)e^{-1/c_{\min}}\left(\frac{1-c^{-0.5}_{\min}}{1+\delta_{\min}}\right)\, \opt,\]
competitive. For large $c_{\min}$ this converges to $(1-1/e)$ with convergence rate $O\left( \sqrt{\frac{\log c_{\min}}{c_{\min}}}\right)$. It only remains to argue that this is best possible. As we have already seen in examples, the case of non-reusable resources is included as a special case. It is well known due to \cite{pruhs,msvv}, that the best achievable result for online matching with non-reusable resources in the large capacity regime is $(1-1/e)$.
\hfill\Halmos\end{proof}

\section{Faster and Distribution Oblivious Algorithms}\label{sec:rba}

We introduced two algorithms for online allocation of reusable resources, \dpg\ and \salg. The latter achieves best possible guarantee for a general model and arbitrary usage distributions. However, it requires (exact) knowledge of the usage distributions (which may not be available in practice) and has slower runtime due to the guiding algorithm \galg, which takes $O\left(\sum_{i\in I} c_i\right)$ steps to fractionally match an arrival.  

The runtime of \galg\ can be exponentially improved at the loss 
a small factor in the performance guarantee. In Appendix \ref{appx:faster}, we present an algorithm with asymptotic guarantee $(1-\epsilon)(1-1/e)$ and runtime $O\left(\frac{1}{\epsilon} \sum_{i\in I} \log c_{i}\right)$ per arrival, for any choice of $\epsilon>0$. 
This runtime improvement carries over to more general models (such as online assortments), 
however, the algorithm continues to rely on exact knowledge of usage distributions.

While \dpg\ serves as a crucial ingredient in \salg, it is itself a promising algorithm and a natural alternative to fluid guided algorithms. 
It is much simpler to implement, computationally efficient, and \emph{oblivious to usage distributions}. {\color{black}Further, \dpg\ matches every arrival that can be matched to some resource. In contrast, due to its realization independent nature, \salg\ may leave an arrival unmatched even if there are resources available with edges to the arrival. }
In the rest of this section, we show that \dpg\ also attains the best possible performance guarantee of $(1-1/e)$ for online matching and a number of different families of usage distributions.  Theorem \ref{genmatch} states the result, followed by a discussion of the types of usage distributions captured by the theorem. We also briefly discuss the difficulties in analyzing \dpg\ for general usage distributions. {\color{black} In Appendix \ref{sec:rbafail}, we discuss additional challenges with analyzing natural extensions of \dpg\ in more general settings.} 
\begin{theorem}\label{genmatch}
\dpg\ is $(1-1/e)-\delta$--competitive for OBMR when usage distributions belong to one of the families in Table \ref{summary}. The table also presents the convergence rate $\delta$ corresponding to each family.  
\end{theorem}

\begin{table}{}	 

\centering
\begin{tabular}{|l|c|}
\cline{1-2}
\multirow{2}{*}{{\bf Distributions}}   & 
\multirow{2}{*}{$\mb{\delta=O(\cdot)}$}
\\       
&\\                                                      
\hline
Two point - $\{d_i,+\infty\}$ & $O\Big(\frac{\log c_{min}}{c_{\min}}\Big)$ \\
\cline{1-2}                              
Exponential              & $O\Big(\frac{1}{\sqrt{c_{\min}}}\Big)$\\ 
\cline{1-2}                                  
\multirow{2}{*}{IFR  } &\multirow{2}{*}{ $\tilde{O}\big( c^{-\frac{\eta}{1+\eta}}_{\min} \big)$}\\
&      \\
\cline{1-2}      
\multirow{2}{*}{IFR with mass at $+\infty$        }       & \multirow{2}{*}{$\tilde{O}\big( c^{-\frac{\eta}{1+2\eta}}_{\min} \big)$ }         \\
&
\\
\cline{1-2}                          
\end{tabular}
\caption{ 
Table presents the rate of convergence to $(1-1/e)$ for various distributions. For IFR distributions, convergence rate is characterized using functions $L_i(\epsilon)$, as defined in \eqref{defineL}. Typically, $L_i(\epsilon_i)\epsilon_i=O(\epsilon_i^{\eta})$ for some $\eta>0$. For IFR with non-increasing densities we have $L_i(\epsilon_i)=1$ and thus, $\eta=1$. For non-monotonic IFR families, such as the Weibull distribution, $\eta=1/k$ where $k>1$ is the shape parameter for IFR Weibull distributions. 
Note that $\tilde{O}$ hides a factor of $\log c_{min}$.} 
\label{summary}
\end{table}	


\noindent \textbf{Bounded Increasing Failure Rate (IFR) distributions:} 
Theorem \ref{genmatch} implies a $(1-1/e)$ guarantee for IFR distributions that are bounded in the following sense. Let $f_i(\cdot)$ denote the p.d.f.\ and $F_i(\cdot)$ denote the c.d.f..\ For values $\epsilon> 0$, 
define function, 
\begin{eqnarray}
L_i(\epsilon)=\max_{ x\geq 0}\frac{F_i\big(x+F_i^{-1}(\epsilon)\big)-F_i(x)}{\epsilon}. \label{defineL}
\end{eqnarray} 
\dpg\ is asymptotically $(1-1/e)$--competitive whenever,
\[\epsilon L_i(\epsilon)\to 0, \text{ for } \epsilon\to 0. \] 
Note that for non-increasing densities $f_i(\cdot)$, $L_i(\epsilon)$ is identically 1. Therefore, as a special case we have $(1-1/e)$--competitiveness for IFR distributions with non-increasing density function. This includes two common families - exponential distributions and uniform distributions. The former 
is commonly used to model usage and service time distributions \citep{besbes, srikant,jsq}. 
Some common IFR distributions (with non-monotonic density function) that satisfy the boundedness condition include 
truncated Normal, Gamma, and Weibull distributions. For more details and the rate of convergence to $(1-1/e)$ for these families, see Appendix \ref{ifrexamples}. 
\smallskip

\noindent \textbf{Two point distributions - $\{d_i,+\infty\}$:} 
While seemingly simple, the family of two point distributions with support $\{d_i,+\infty\}$ for $i\in I$,  presents a non-trivial challenge in both small and large inventory regimes (see Section \ref{sec:stochrew} and Appendix \ref{sec:challenges}). \dpg\ is $(1-1/e)$--competitive for this family. 
\smallskip

\noindent\textbf{Bounded IFR with mass at $+\infty$:} 
An important consideration 
that is not modeled by IFR distributions is the possibility that units may become unavailable over time due to resource failures or 
when resources are live agents they may depart from the platform after some number of matches. We account for this by allowing an arbitrary probability mass at $+\infty$ to be mixed with an IFR distribution. More concretely, consider usage distributions where 
for resource $i$, a duration takes value $+\infty$ w.p.\ $p_i$ and with probability $1-p_i$ the duration is drawn from a bounded IFR distribution with c.d.f.\ $F_i(\cdot)$. \dpg\ is $(1-1/e)$-competitive for this family. 



\subsection{Analysis of \dpg: Overview of Main Challenges}
In this section, we briefly discuss the main new concepts needed for analyzing the performance of \dpg. Further details are included in Appendix \ref{appx:rbaoverview}. 
We start by introducing (and recalling) important notation. We use $\omega$ to denote the sample path of usage durations in \opt\ and $O(\omega,i)$ to denote the set of all arrivals matched to $i$ in \opt\ on sample path $\omega$. Similarly, $\nu$ denotes a sample path of usage durations in \dpg.  Recall that \dpg\ tracks the highest available unit for each resource, denoted as $z_i(t)$ for resource $i$ at arrival $t$. In fact, since \dpg\ is realization dependent, $z_i(t)$ is more accurately written as $z_i(\nu,t)$ but we continue to use the shorthand $z_i(t)$. Given set of resources $S_t$ with an edge to $t$, \dpg\ matches according to the following simple rule,
\[\argmax_{i\in S_t} r_i \left(1-g\left(\frac{z_i(t)}{c_i}\right)\right). \]
Let $D(t)$ (technically $D(\nu,t)$) denote the resource matched to $t$ by \dpg\ on sample path $\nu$. Let $z_{D(t)}$ denote the highest available unit of resource $D(t)$ at $t$. Finally, recall that $\Delta g(k)=g(\frac{k-1}{c_i})-g(\frac{k}{c_i})=e^{-\frac{k}{c_i}}(e^{\frac{1}{c_i}}-1)$.

To show the desired guarantee for \dpg\ we use the following candidate solution for the LP free certificate, 
\begin{eqnarray}
\theta_i = r_i \, \ed\Big[\sum_{t\mid D(t)=i} g\Big(\frac{z_i(t)}{c_i}\Big) \Big],\, \forall i\in I 
\text{   and   }
\lambda_t= \ed \Big[ r_{D(t)} \Big(1-g\Big(\frac{z_{D(t)}}{c_{D(t)}}\Big)\Big)\Big], \, \forall t\in T. \label{theta}\label{lambda}
\end{eqnarray}
Similar to the analysis of \galg\ in Section \ref{sec:solution}, we defined $\lambda_t$ and $\theta_i$ as deterministic quantities. Since $\lambda_t$ is independent of $\omega$ and $\nu$, it suffices to use the simplified system of inequalities given by \eqref{scert1} and \eqref{scert2}. The candidate solution above satisfies \eqref{scert1} by definition, with $\beta=1$.  The main step is to prove that \eqref{scert2} is satisfied with $\alpha_i=(1-1/e)$ for every $i\in I$. 
The following lemma brings out the key source of difficulty in proving this claim. See Appendix \ref{appx:rbaproofs} for the proof. 


\begin{lemma}\label{decompose1} Given $\lambda_t$ as defined by \eqref{lambda}, we have for every resource $i$,
\begin{equation}
\eo\Big[\sum_{t\in O(\omega,i)} \lambda_t\Big]\geq (1-1/e) \opt_i - r_i \eo\Big[ \sum_{t\in O(\omega,i)} \sum^{c_i}_{k=1} \Delta g (k) \pd \big[k> z_i(t)\big]\Big]. \label{lam4}
\end{equation}
\end{lemma}
{\color{black}
Consider a resource $i\in I$ and let $k$ denote a unit of $i$. Contrast decomposition \eqref{lam4} with its counterpart given in Lemma \ref{decompose}. 
Deterministic binary valued quantities $\onee(\neg k,t^+)$, that signified unavailability of a unit $k$ after arrival $t$ is matched in \galg, have been replaced by (non-binary) probabilities $\pd[k>z_i(t)]$. For any given arrival $t$, probability $\pd [k> z_i(t)]$ represents the likelihood that unit $k$ and all higher units of $i$ are unavailable in \dpg\ when $t$ arrives. This dependence on other units, combined with the sensitivity of $z_i(t)$ to small changes on the sample path, makes it challenging to bound these probabilities in a meaningful way. 
To address these challenges we introduce two new ingredients. 
The first ingredient simplifies the stochastic dependencies by introducing a conditional version of the probability $\pd[k>z_i(t)]$. The second ingredient builds on the first one and addresses the non-binary nature of these probabilities. We discuss these ingredients in more detail in Appendix \ref{appx:rbaoverview}. These ingredients use special structure that is not available in settings beyond online matching. Analyzing \dpg\ in more general settings remains a challenging open problem  (see Appendix \ref{sec:rbafail} for further discussion).

Recall that an $(F,\bsigma, \mb{p})$ random process with probabilities $p_t=1$ for every $\sigma_t\in \bsigma$ is called an $(F,\bsigma)$ random process. Using the new ingredients outlined above and some general steps (similar to the analysis of \galg), we reduce the entire analysis to proving a certain \emph{perturbation property} of $(F,\bsigma)$ random processes  (see Proposition \ref{relatexp} in Appendix \ref{appx:rbaoverview}). We show this property for the families of usage distributions described in Table \ref{summary}. We believe that the property holds for arbitrary usage distributions but proving this remains a challenging open problem. This is the only missing step towards proving a general $(1-1/e)$--competitive ratio for \dpg\ in the context of online matching. In other words, our analysis shows that \dpg\ is $(1-1/e)$--competitive for every family of usage distributions that satisfies the perturbation property. 

\section{Extensions to Online Assortments and Budgeted Allocation} \label{sec:mtoasst}

Our algorithms -- \dpg, \salg\ and its variants --  extend quite naturally to more general settings. In this section, we formally introduce a more general model and discuss the extension of our results as well as the new ideas needed to generalize our analysis. In a nutshell, we show that a natural generalization of \salg\ is $(1-1/e)$--competitive in the following model.  We outline the key difficulty with proving performance guarantee for \dpg\ in this model in Appendix \ref{sec:rbafail}.
\smallskip

\noindent \textbf{Online Assortment Optimization with Multi-unit Demand:} 
In this setting, we do not directly match arrivals to resources but instead offer a set of available resources (an assortment) to each arriving customer. Customer $t\in T$, probabilistically chooses at most one resource from the set according to a choice model $\phi_t: 2^{I}\times I\to [0,1]$, where $\phi_t(S,i)$ is the probability that resource $i$ is chosen from assortment $S$. The choice model $\phi_t$ is revealed to the online algorithm when customer $t$ arrives but the customer's choice is realized after the assortment decision is made. 
We consider a multi-unit demand setting where customer $t$ is interested in up to $b_{it}\geq 0$ units of resource $i\in I$. 
If $t$ chooses $i$ from the offered assortment and $y_i(t)$ units of $i$ are available, then $t$ takes $\min\{y_i(t),\, b_{it}\}$ units of $i$, generating a reward $r_i\,\min\{y_i(t),\, b_{it}\}$. The choice model $\phi_t$ and unit requirements $(b_{it})_{i\in I}$ are revealed to the online algorithm on arrival of $t$. Resources are reusable and different units of resource $i$ allocated to a given customer may all be used for the same random duration $d\sim F_i$, or each unit for an independently drawn random duration (with distribution $F_i$). Let $\gamma=\min_{(i,t)\in E} \frac{c_{i}}{ b_{it}}$. In the large inventory regime for this setting we have,  $\gamma\, \to +\infty$. 

The setting generalizes two important models from online allocation of non-reusable resources: (i) Online assortment optimization~\citep{negin}, where $b_{it}\in\{0,1\}$ for every $i\in I, t\in T$, and (ii) Online budgeted allocation or Adwords~\citep{msvv}, where values $b_{it}$ can be arbitrary but choice is deterministic (as in	 online matching). 	We show the following guarantee for arbitrary usage distributions\footnote{Guarantees are shown against clairvoyant that knows $\phi_t$ for all customer $t\in T$ at the start of the planning horizon but observes the realization of customer $t'$s choice only after it offers an assortment to the customer (same as the online algorithm). Clairvoyant can offer at most one assortment to each customer (same as the online algorithm).}. 

\begin{theorem}[{Generalization of Theorem \ref{main}}]\label{main2}
For online assortment with multi-unit demand and arbitrary usage distributions, there is a $(1-1/e-\delta)$--competitive realization independent algorithm for $\delta=O\left(\sqrt{\frac{\log \gamma}{\gamma}} \right)$. For $\gamma\to+\infty$, the guarantee approaches $(1-1/e)$ and this is the best possible asymptotic guarantee. 
\end{theorem}
}
\galg\ and \salg\ admit a natural generalization for the setting of online assortment optimization with multi-unit demand. We defer a formal generalization of the algorithm and analysis to Appendix \ref{appx:asst}. At a high level, we generalize \galg\ to obtain a fluid guide that outputs a fractional solution over assortments at each arrival. To generalize \salg, we could independently round the fractional solution given by the guide to get a candidate assortment at each arrival. However, it is possible that some items in the candidate assortment are unavailable. In fact, for large assortments the probability that at least one item is unavailable can be close to 1. In light of this, perhaps a natural approach is to offer the available subset of items from the candidate assortment. However, due to substitution, this can increase the probability of individual items being chosen, prohibiting a direct application of concentration bounds. To address this we introduce an algorithm that constructs new assortments with lower choice probability for the individual items in order to precisely control the probability of each resource being chosen. The following lemma and the subsequent \emph{Probability Matching} algorithm formalizes these ideas. 


\begin{lemma}\label{probmatch}
Consider a choice model $\phi: 2^{N}\times N \to [0,1]$ satisfying the weak substitution property\footnote{Weak substitution: $\phi(S,i)\geq \phi(S\cup\{j\},i),\, \forall i\in S, j\not\in S$.}, an assortment $A\subseteq N$ belonging to a downward closed feasible set $\mathcal{F}$, a subset $S\subseteq A$ and target probabilities $p_s$ such that, $p_s\leq \phi(A,s)$ for every  $s\in S$. There exists a collection $\mathcal{A}=\{A_1,\cdots,A_m\}$ of $m=|S|$ assortments along with weights $\big(u_i\big) \in[0,1]^{m}$, such that the following properties are satisfied:
\begin{enumerate}[(i)]
\item For every $i\in[m]$, $A_i \subseteq S$ and thus, $A_i\in \mathcal{F}$.
\item Sum of weights, $\sum_{i\in[m]} u_i \leq 1$.
\item For every $s\in S$, $\sum_{A_i\ni s} \, u_i\, \phi(A_i,s) = p_s$.
\end{enumerate}
Algorithm \ref{pmatch} computes such a collection $\mathcal{A}$ along with weights $(u_i)$ in $O(m^2)$ time.
\end{lemma}
\begin{algorithm}[H]
\SetAlgoNoLine
\textbf{Inputs:} set $S$, choice model $\phi(\cdot,\cdot)$, target probabilities $p_s\leq \phi(S,s)$ for $s\in S$\;
\textbf{Output:} Collection $\mathcal{A}=\{A_1,\cdots,A_m\}$, weights $\mathcal{U}=\{u_1,\cdots,u_m\}$ with $\underset{A_i\ni s}{\sum}u_i \phi(A_i,s)=p_s$ $\forall s\in S$\;
\For{$j=1$ to $|S|$}{
Compute values $\zeta_s=\frac{p_s}{\phi(S,s)}$ for all $s\in S$\;
Define $s^*=\arg\min_{s\in S} \zeta_s$\;
Set $A_j=S$ and $u_j=\zeta_{s^*}$\;
Update $S \to S\backslash \{s^*\}$ and $p_s\to p_s-u_j\,\phi(S,s),\, \forall s\in S$\;
}

\caption{Probability Match $(S,\{p_s\}_{s\in S})$}
\label{pmatch}
\end{algorithm}

{\color{black}
\textbf{Remarks:} It can be shown that this subroutine runs in $O(m\log m)$ time for the commonly used MNL choice model.  We note that \cite{feng2} concurrently developed a randomized version of our probability matching subroutine (Algorithm \ref{pmatch}) in a setting with \emph{stochastic arrivals} (see sub-assortment sampling in \cite{feng2}). While probability matching outputs a distribution over assortments, sub-assortment sampling outputs a single randomized assortment from the same distribution. 
Prior work by \cite{chen}, which focuses on a setting with non-reusable resources, also proposes an idea that has similarities with our subroutine. They designed an algorithm that computes a distribution over assortments by iteratively eliminating items that have reached a kind of capacity limit\footnote{See discussion on breakpoints on page 5 and Procedures 1 and 2 in \cite{chen}.}. Their procedure has a different goal and they do not need (or prove) the result stated in Lemma \ref{probmatch} above.
Finally, recent works show that a generalization of this subroutine also has applications in settings with product retirement \citep{elmachtoub2022revenue} as well as pricing of substitutable goods \citep{pricing}, where the fact that probability match outputs a distribution over nested assortments turns out to quite be important. 
}

{\color{black}

\section{Numerical Experiments}\label{sec:numerics}

In this section, we describe results from a numerical experiment where we compare different algorithms on a synthetically generated data set. In order to obtain a clear comparison between the algorithms, we focus on the OBMR setting where there is no sampling noise from the randomness in customer choice. Overall, we find that \dpg\ tends to outperform \salg, Balance, and the greedy algorithm in many of the scenarios that we consider. 

Before describing our experimental setup, we recall a family of ``hard'' instances for online $b$-matching~\citep{pruhs}. Consider an instance with $n$ resources, $n$ different types of arrivals, and $n^2$ arrivals in total. There are $n$ units of each resource and resources have identical rewards ($r_i=1\,\, \forall i\in [n]$). The first $n$ arrivals have an edge to every resource but arrivals get \emph{pickier} over time. The next $n$ arrivals have an edge to every resource in the set $[n-1]$. In general, for $k\in[n]$, arrivals $\{(k-1)n+1,\cdots, kn\}$ are of the same type and each of them has an edge to every resource in $[n-k+1]$. It is not hard to see that the optimal offline solution has value $n^2$. For $n\to+\infty$, greedy (with the worst tie breaking rule) matches $\approx 0.5n^2$ arrivals and Balance matches $\approx (1-1/e)n^2$ arrivals.

\subsection{Experimental Setup} 
Our setup is inspired by the family of hard instances above. We use four `parameters' to generate different scenarios: the number of resources $n$, the starting inventory of resources $c$, the type of usage distribution $F$, and a parameter $\kappa$ that influences the arrival sequence. We describe the setup in terms of these parameters and then discuss the different parameter settings that we consider. 
\smallskip

\noindent \textbf{Resources:} We use identical starting inventory ($c$) and identical rewards ($r_i=1\,\,\forall i\in [n]$) for all resources. The resources also have the same usage distribution $F$. For the subsequent discussion, let $F$ be the two-point distribution where a matched unit returns after 1 unit of time w.p.\ 0.5 and w.p.\ 0.5 the unit does not return until the end of the planning horizon. In our experiments, a unit that returns at time $t$ can be matched to an arrival at time $t$. We discuss other usage distributions that we experiment with later on. 
\smallskip

\noindent \textbf{Arrivals:} There are $2n$ different types of arrivals and $2cn$ arrivals in total.  An arrival of type $i\in[n]$ has an edge to every resource in $[i]$. Arrivals of type $i\in[2n]\backslash[n]$ have an edge to resource $i-n$ only.  Observe that for the two-point usage distribution, the expected value of the number of times we can match a single unit of resource is 2, and the maximum total expected number of matches is $2cn$.

We describe the sequence of arrivals by splitting it into $n$ consecutive phases with $2c$ arrivals in each phase. Phase 1 is spread over time interval $[1,c+1]$ and starts with a ``burst'' of $c$ arrivals at time 1 followed by one ``normal'' arrival at every time $t\in\{2,\cdots,c+1\}$. Every arrival at time 1 is of type $2n$, i.e., has an edge to resource $n$. 
For every normal arrival, the type is sampled independently from distribution $D_1$ over arrival types that we will define later. In general, for $k\in[n]$, phase $k$ is spread over time interval $[(c+1)(k-1)+1,\, (c+1)k]$ and consists of a burst of $c$ arrivals of type $2n-k+1$ at time $(c+1)(k-1)+1$, followed by a (single) normal arrival at every time $t\in\{(c+1)(k-1)+2,\cdots, (c+1)k\}$. The type of each normal arrival is sampled independently from distribution $D_k$ (to be described). Overall, let $T_1$ denote the set of bursty arrivals 
and let $T_2$ denote the set of normal arrivals.
\smallskip

\noindent \emph{Distributions $D_k$:} A normal arrival in phase $k$ is type $i$ with probability proportional to $e^{-\kappa |i-(n-k+1)|}$, here $|\cdot|$ denotes the absolute value. Observe that the probability peaks for type $n-k+1$ and for $\kappa\to+\infty$, every normal arrival in phase $k$ is of type $n-k+1$ w.p. 1. This mimics the worst case arrival sequence for non-reusable resources where the arrivals become pickier over time. On the other hand, when $\kappa=0$, every normal arrival is independently sampled from the uniform distribution over $[n]$. 
\smallskip

\noindent \emph{Remarks on the arrival sequence:} The distribution of normal arrivals closely mimics the distributions over arrival types that have been used for numerical experiments in prior work~\citep{RST18,reuse}. Suppose that we remove the bursty arrivals ($T_1$) and consider the sub-sequence of normal arrivals $T_2$ without changing the arrival times. In the resulting setup, we find that all the online algorithms that we test have similar performance and are close to optimal in most scenarios. We include a sample of the results in Table \ref{suptab1} in Appendix \ref{appx:exp}. 
This outcome matches with the experimental results of \cite{reuse}, who considered a very similar setup. 
By mixing normal arrivals with bursty arrivals, we obtain some insights on scenarios where \dpg\ (and \salg) may outperform Balance and Greedy. To see this, let $\kappa\to+\infty$ and consider the arrivals in phase 1. Ideally, we should match all arrivals in this phase to resource $n$ (if we can) because future arrivals do not have an edge to resource $n$. The algorithms under consideration will all match the bursty arrivals at time 1 to resource $n$. However, it can be shown that \dpg\  will match a relatively larger number of normal arrivals to resource $n$ as well (for more details, see Examples A.1, A.2, and A.3 in Appendix \ref{sec:challenges}).  
\smallskip

\noindent {\bf Algorithms and Benchmark}: We compare the performance of \dpg, \salg, Balance, and greedy with an LP benchmark. We use an LP that upper bounds the expected total revenue of clairvoyant, where the expectation is over the randomness in arrival sequence as well as usage durations (see Appendix \ref{appx:exp} for more details). In our implementation of \salg, we set $\delta_i=0\,\, \forall i\in I$. Tuning this parameters did not change the overall performance of \salg\ significantly. For all algorithms, we use a deterministic tie breaking rule that chooses the resource with the lowest index \footnote{This tie breaking rule is without loss of generality because one could, for example, set $r_i=1-\epsilon i\,\, \forall i\in [2n]$ for a suitably small value $\epsilon>0$. This would eliminate ties in the online algorithms we consider without significantly changing any of the experimental results.}. 

\subsection{Experimental Results}
We experiment with $n\in\{5,20\}$, starting inventory $c\in\{5,15,25\}$, and $\kappa\in\{0,1\}$. 
In addition to the two-point distribution, we consider the following two usage distributions ($F$): the exponential distribution with mean $c$, i.e., $F(x)=1-e^{-x/c}$, and the non-IFR Weibull distribution with mean $c$, shape parameter 0.5, and scale parameter $c/2$, i.e., $F(x)=1-e^{-\left(\frac{2x}{c}\right)^{0.5}}$. 
The Weibull distribution allows us to examine the performance of \dpg\ in a setting that is not included in our theoretical result for \dpg. Note that, setting a mean value of $c$ (approximately) calibrates the value of the LP benchmark across all three types of usage distributions in many of the scenarios that we consider. 

\begin{table}[h]{\color{black}
	\centering
	\begin{tabular}{|c|c|c|c|c|c|}
		\hline
		  \textbf{$c$} &\textbf{$F$}&\text{$\quad$\textbf{\dpg}$\quad$} & \text{$\quad$\textbf{Balance}$\quad$} & \text{$\quad$\textbf{Greedy}$\quad$} &\textbf{\salg} \\ \hline
				\multirow{3}{*}{5} & Two-point  & 0.781 & 0.754 & 0.727 & 0.613  \\ \cline{2-6}
		 & Exponential   & 0.883  & 0.876  & 0.889 & 0.646  \\ \cline{2-6}
		 & Weibull  & 0.942 & 0.936 & 0.936 & 0.672 \\ \hline
		\multirow{3}{*}{15} & Two-point  & 0.814 & 0.771 & 0.737 & 0.728  \\ \cline{2-6}
		 & Exponential & 0.896 & 0.878  & 0.899 & 0.793\\ \cline{2-6}
		 & Weibull & 0.937 & 0.923 & 0.931 & 0.846 \\ \hline
		 	\multirow{3}{*}{25} & Two-point & 0.827 & 0.781 & 0.737 & 0.816 \\ \cline{2-6}
		 & Exponential & 0.896 & 0.878  & 0.897 & 0.884 \\ \cline{2-6}
		 & Weibull & 0.934 & 0.918 & 0.923 &0.925 \\ \hline
	\end{tabular}
	\caption{Average performance of algorithms for \( n = 5 \) and \( \kappa = 1 \) in comparison to the LP benchmark. For \dpg, Balance, and greedy, the standard deviation of the ratio is less than $0.0001$ in every scenarios. For \salg, the standard deviation is less than $0.01$. All values have been rounded up to three decimal places.}\label{maintab}}
\end{table}
Let us consider $n=5$ and $\kappa=1$. Table 1 summarizes our results for different values of $c$ and different types of usage distributions $F$. For each setting of $c$ and $F$, we randomly generate 100 instances (arrival sequences) and simulate all four algorithms 20 times on each instance. The performance of each algorithm (\alg) in the table is reported as the ratio of the empirical average performance of \alg\ (based on $20\times100$ trials) and the optimal value of the LP benchmark. Recall that we use an LP that upper bounds the \emph{expected} performance of clairvoyant (see Appendix \ref{appx:exp}). In other words, the LP is independent of the random arrival sequence and usage durations only needs to be solved once\footnote{One could also use an instance dependent LP relaxation (say ILP) that depends on the (random) arrival sequence and then compute the empirical average of the ratio of \alg\ and ILP. We believe that this would not significantly change the overall trends that follow from our results.}. The standard deviation of the values is small for all algorithms  and we report the values below the table.
 Observe that \dpg\ dominates the other algorithms in most scenarios, earning up to $ 5\%$ more reward than Balance and up to $10\%$ more than greedy in some scenarios. Balance dominates greedy for the two-point distribution but greedy performs slightly better for exponential and Weibull distributions. \salg\ performs poorly when the starting inventory is small ($c=5$) but the performance improves as the inventory increases. We believe that the unimpressive numerical performance of \salg\ is partly due to the realization independent nature of the algorithm that can lead the algorithm to propose matches to resources that are unavailable. This may be more likely to occur when the starting inventory is small.  

In Appendix \ref{appx:exp}, we also include results for $n=5, \kappa=0$ (Table \ref{suptab2}), $n=20, \kappa=1$ (Table \ref{suptab3}), and $n=20, \kappa=0$ (Table \ref{suptab4}). Together, these tables sufficiently cover all the scenarios that we considered. Overall, we find that \dpg\ continues to lead the other algorithms across these different scenarios. Recall that our $(1-1/e)$ lower bound on the performance of \dpg\ does not extend to settings where the usage durations are non-IFR. However, the numerical performance of \dpg\ indicates the possibility that it may have a strong performance guarantee in general. 

%

}
\section{Summary and Future Directions}\label{sec:conclusion}
We considered settings in online resource allocation when resources are reusable. 
Focusing on the large capacity regime, we proposed a new distribution oblivious algorithm called Rank Based Allocation (\dpg). By developing a scheme to turn realization dependent online algorithms into realization independent ones, and using a new LP free system for certifying competitiveness, 
we gave a randomized algorithm that achieves the best possible guarantee of $(1-1/e)$ for arbitrary usage distributions. To show this guarantee in the setting of online assortments we also developed a novel probability matching subroutine that gives us full control over the substitution behavior in assortments. Finally, we showed that the much simpler (and deterministic) \dpg\ algorithm also achieves the optimal guarantee of $(1-1/e)$ for online matching when the usage distributions are IFR (roughly speaking). 

{\color{black}
A natural question that remains open for future work is whether there exists an algorithm that performs better than greedy 
when resources have unit capacity\footnote{Very recently, \cite{delong} (subsuming \cite{periodic})  made progress in the unit inventory setting for the special case of online matching with \emph{deterministic} and \emph{identical} usage durations i.e., when every matched resource is used for a fixed duration $d\geq 0$.} (the most general setting). In fact, many such questions remain open in the low capacity regime and in Appendix \ref{apx:connection} we shed some light on the difficulty of this problem by establishing a connection between a very special case of reusability and a long standing open problem in the setting of online matching with stochastic rewards. 
On the more technical side, it would be interesting to optimize the convergence rates for the asymptotic guarantees or to find a deterministic algorithm (like \dpg) is asymptotically $(1-1/e)$ for general usage distributions (beyond IFR). Finally, it would be interesting to see if our algorithmic and analytical framework can be fruitfully applied to other related settings\footnote{Recent works such as \cite{aouad2020online, manshadi2022online,unkads, pricing, osow} have demonstrated new applications of the LP free framework}.} 
{\small
\bibliographystyle{informs2014.bst}
\bibliography{./../../bib}
}

\begin{APPENDICES}
\newpage
\appendixpage
\startcontents[sections]
\printcontents[sections]{l}{1}{\setcounter{tocdepth}{2}}
\newpage
\setcounter{theorem}{0}
\renewcommand{\thetheorem}{\Alph{section}\arabic{lemma}}
\renewcommand{\thetheorem}{\Alph{section}\arabic{theorem}}
\section{Limitations of Balance and Primal-dual for Stochastic Usage Durations} \label{sec:challenges}
Recall that in case of non-reusable resources, the Balance algorithm combined with primal-dual analysis leads to the best possible $(1-1/e)$ guarantee in a variety of settings. Through simple examples we now demonstrate some of the challenges with applying these ideas to the more general case of reusable resources. These examples also illustrate the ability of our new algorithm and analysis approach to address uncertainty in reusability. 

\subsection{Performance of Balance for Two-point usage distributions}\label{appx:twopoint}
We consider settings where for every resource $i\in I$, the usage durations are stochastic but take only one of two values - a value $d_i>0$ or $+\infty$. This gives us a simple and natural generalization of the setting of non-reusable resources as well as the setting of deterministic usage durations. We show that Balance has a competitive ratio strictly less than $(1-1/e)$ in this setting. This indicates that quantities such as the time interval between arrivals and the probability of matched units returning before the next arrival may play an important role. 

We start with a simple example that demonstrates the difference between Balance and \dpg.
\begin{eg}\label{passivity}
\emph{Consider an instance with two resources, labeled $1$ and $2$. Resources have the same reward and a large capacity $n$. We use a two-point usage distribution for both resources, with support $\{1,+\infty\}$ and probability of return $0.5$. Consider a sequence of $4n$ arrivals as follows: At time 0 we have a burst of $2n$ arrivals all within a very short amount of time $\epsilon\to 0$. These arrivals can only be matched to resource 1. 
This is followed by $n$ regularly spaced arrivals at time epochs $\{2,4,\dots,2n\}$, each of which can be matched to either resource. Finally, we have a burst of $n$ arrivals that can only be matched to resource 2, at time $2n+2$. For large $n$, clairvoyant can match $\sim 3n$ arrivals with high probability (w.h.p.). }

\end{eg}  

\noindent \textbf{Comparing Balance with \dpg:} Balance will match to the resource with highest fraction of remaining capacity. So it matches the first $n$ arrivals $t\in[n]$ (half of the first burst), to resource 1. With high probability (w.h.p.), nearly $n/2$ units of resource 1 return by time 2 and the rest never return. W.h.p., Balance matches most of $n$ the arrivals $t\in[2n]\backslash[n]$ occurring between time $2$ and $2n$ to resource 2 and thus, can only match half of the final burst of $n$ arrivals $t\in[3n]\backslash[2n]$ at time $2n+2$. Balance fails to recognize that due to reusability the $n/2$ remaining units of resource 1 could all be matched to the second set of $n$ arrivals and, in this way, the ``effective" remaining capacity of resource 1 is  $n$. 

Consider the decisions of \dpg\ in this instance. \dpg\ coincides with Balance over the first $2n$ arrivals. 
After the first burst, w.h.p., the highest available unit of resource 1 in \dpg\ is no lower than $n-O(\log n)$. In fact, w.h.p., \dpg\ manages to match a constant fraction ($\sim 1/3$) of the $n$ spread out arrivals to resource 1, successfully gauging ``effective" inventory of the resource. Overall, \dpg\ matches $\sim n/6$ more arrivals than Balance. In general, the metric $z_i(t)$ is very sensitive to reusability, i.e., it tends to have a high value when arrivals are spaced out and units return ``frequently enough", and when this is not the case it acts closer to Balance and protects resources with low inventory.
\smallskip

\noindent \textbf{What about fluid reusability?} Does the framework of turning realization dependent algorithms into realization independent ones by means of using a fluid guide (see Section \ref{sec:matching}), addresses the above issue with Balance? Notice that even when we consider the fluid versions of usage distributions in Example \ref{passivity}, the actions of Balance do not change. 
{\color{black}	In fact, by modifying Example \ref{passivity}, we establish an upper bound of 0.626 on the competitive ratio of Balance.

In Example \ref{passivity}, let $r_t$ denote the reduced price computed by Balance for the resource matched to arrival $t\in[4n]$.
\begin{eg}\label{ega2} \emph{We consider the 
instance in Example \ref{passivity} and augment it as follows. 
For every $t$, we have a new low reward non-reusable resource $i_t$ with price $\max\{0,\frac{r_t}{(1-1/e)} -\delta\}$, for some small $\delta>0$. Resource $i_t$ does not have an edge to any arrival except $t$. For every $t$, let $i_t$ have large capacity and usage duration of $+\infty$. We also augment the arrival sequence in Example \ref{passivity} with additional $n/2$ arrivals at the end (time $2n+2$), making a total of $1.5n$ arrivals in the final burst and $4.5n$ arrivals in total. The final $n/2$ arrivals have an edge to resource 2 only.}
\end{eg} 


\begin{repeattheorem}[Theorem \ref{balup} (restated)]
Consider the family of instances given in Example \ref{ega2}. 
Suppose we allow fractional matching and experience fluid versions of usage distributions such that, if an $\epsilon$ amount of resource $i$ is matched to $t$ then $0.5\epsilon$ returns at time $t+1$ and the other half never returns.  
Then, the total revenue of fractional Balance algorithm is $< 0.626\, \opt$ $(<(1-1/e)\, \opt)$. 
\end{repeattheorem}
\begin{proof}{Proof.}
Notice that, the matching output by Balance does not change with the addition of the low reward resources. Further, Balance does not match the additional $n/2$ arrivals added in the final burst. Therefore, the total reward of Balance is $2.5n$. We now compute a lower bound on the clairvoyant. Let \opt\ match the first $n$ arrivals (from 1 to $n$) to the low reward resources. This gives a reward of at least \[\frac{\sum_{i=1}^{n} 1-e^{-i/n}}{1-1/e}-O(n\delta),\] here $\frac{ 1-e^{-t/n}}{1-1/e}-\delta$ is the revenue from matching arrival $t$ to $i_t$. Let \opt\ match the next $2n$ arrivals (from $n+1$ to $3n$) to resource 1 (as much as possible). Using concentration bounds for large $n$, we have that \opt\ matches $2n-o(n)$ of the arrivals to resource 1 w.h.p..\ In the final burst of arrivals, \opt\ matches first $n/2$ arrivals (from $3n+1$ to $3.5n$) to low reward resources and final $n$ arrivals (from $3.5n$ to $4.5n$)  to resource 2. This gives additional reward of at least $ \frac{\sum_{i=0}^{0.5n-1}1-e^{-0.5+i/n}}{1-1/e}-O(n\delta)+n$. The total value of \opt\ is at least,
\begin{eqnarray*}
&&\frac{\sum_{i=1}^{n} 1-e^{-i/n}}{1-1/e}+2n+\frac{\sum_{i=0}^{0.5n-1}1-e^{-0.5+i/n}}{1-1/e}+n-O(n\delta),\\
&&\qquad \geq n\left(3+\frac{1/e}{1-1/e}+\frac{0.5-(1/\sqrt{e} -1/e)}{1-1/e}-O(\delta)\right),\\ 
&&\qquad \overset{\delta\to0}{\geq }\frac{1}{0.626}\times 2.5\,n.
\end{eqnarray*}
\hfill\Halmos
\end{proof}
}
%
%
\subsection{Comparison for Exponential usage distributions}
We consider a different example below to further demonstrate \dpg's ability to 
adapt to arrival sequence and usage distribution, without explicitly using the distribution.

\begin{eg}\label{passivity2}
\emph{
Consider a setting with two resources that have rewards $r_1=1$ and $r_2=2$ and large starting capacities $n$. The usage time distribution of both resources is exponential with rate $\mu$ (mean $1/\mu$). We receive an arrival sequence where the first $n-1$ arrivals come in a very short span of time $\epsilon\to 0$ and are only interested in a unit of resource 2. Suppose that the next arrival, call it $t_0$, is 1 unit of time later and can be matched to either resource. 
}\end{eg}
\noindent \textbf{Comparing Greedy, Balance, and \dpg:} All algorithms match the first $n-1$ arrivals to resource 2. Greedy matches $t_0$ to resource 2 for every value of $\mu$. On the other hand, Balance will be quite risk-averse and protect resource 2 from being matched unless $\mu>0.5$ (roughly). \dpg, on the other hand, responds quite nimbly to $\mu$. It will protect resource 2 when $\mu\to 0$ i.e., the 1 unit time interval is insignificant for inventory to replenish. For every non-infinitesimal value of $\mu$, the highest available unit of resource 2 will have index at least $n-O(\log n)$ with high probability and therefore, \dpg\ will act greedily and match the arrival to resource 2.  Figure \ref{ex2fig} summarizes these differences between the algorithms.
\begin{figure}[h]
\centering
\includegraphics[scale=0.5]{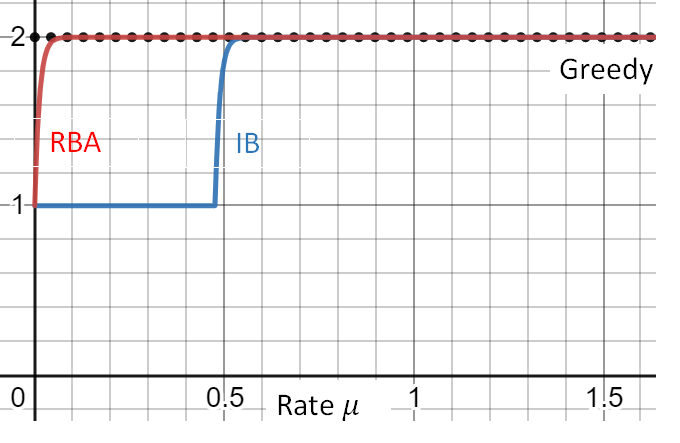}
\caption{Comparison of Greedy, Balance, and \dpg\ on Example \ref{passivity2}. Numbers on the $y$ axis denote the resource matched to $t_0$ w.h.p..}
\label{ex2fig}
\end{figure}

\subsection{Variants of Balance}\label{apx:balvar}
Examples \ref{passivity} and \ref{passivity2} hint at the following ``switching" behavior: In phases of arrivals where resources return ``frequently" relative to the arriving demand, it is better to be greedy. On the other hand, when arrivals occur in a large batch and resources can not return in time, we should follow Balance to protect diminishing resources. Since we have no information about future arrivals it is not clear how an online algorithm can make this ``switch" optimally. While \dpg\ manages to perform this switch quite nimbly, the following natural variants of Balance fail to do so. 
\begin{enumerate}[(i)]
\item  An algorithm that uses distributional knowledge to deduce when some items are not going to return and refreshes the maximum capacity instead of just the remaining capacity. In the context of Example \ref{passivity}, this algorithm would realize after 1 unit of time that the units of resource 1 that have not returned, will never return. Subsequently, it computes a new maximum capacity of $n/2$ for resource 1 at time 1, and treats the resource as if it were at full capacity. It is not hard to see that in general this ends up converging to the greedy algorithm (when the return probabilities approach 0 for instance).
\item An algorithm that anticipates that items are going to return in the future and  considers a more optimistic remaining capacity level. In Example \ref{passivity}, this algorithm would deduce that there are no further items returning after time 1. Therefore, it makes the same decisions as Balance on the instance in Example \ref{passivity}.
\end{enumerate}
\subsection{Challenges with Standard Primal-dual Analysis}\label{primalimpossible}
{\color{black}	From an analysis standpoint, 
it has previously been observed that using the primal-dual framework of \cite{devanur} and \cite{buchbind}, a cardinal technique of analysis in case of non-reusable resources, presents non-trivial challenges 
and typical dual fitting arguments do not seem to work \citep{RST18,reuse}. For a concrete demonstration, consider the LP upper bound on clairvoyant and its dual,
\begin{eqnarray}
\textbf{Primal}\quad  \max &&  \sum_{(i,t)\in E} r_i y_{it}\nonumber\\
s.t.  &&\sum_{t=1}^{\tau}[1-F_i(a(\tau)-a(t))]y_{it}\leq c_i \quad \forall \tau\in\{1,\dots,T\},\forall i\in I \nonumber\\
&& \sum_{i\in I} y_{it} \leq 1\quad \forall t\in T\nonumber\\
&& 0\leq y_{it}\leq 1 \quad \forall t\in T, \quad i\in I
\end{eqnarray}

\begin{eqnarray}
\textbf{Dual}\quad &\min &\sum_{t}\lambda_t +\sum_{(i,t)\in E} c_i \theta_{it}\label{dualnat}\\
&s.t.\ & \lambda_{t} +\sum_{\tau \mid a(\tau)\geq a(t)}[1-F_i(a(\tau)-a(t))]\,\, \theta_{i\tau} \geq r_i\quad \forall (i,t)\in E\nonumber\\
&& \lambda_t, \theta_{it}\geq 0 \quad \forall t\in T,i\in I	\nonumber
\end{eqnarray}
Given an online algorithm \alg\ (with expected revenue \alg), if one can find a dual fitting $\lambda_t\geq 0$ and $\theta_{i\tau}\geq 0$ such that,
\begin{enumerate}[(i)]
\item  $\lambda_{t} +\sum_{\tau \mid a(\tau)\geq a(t)}[1-F_i(a(\tau)-a(t))]\,\, \theta_{i\tau} \geq \alpha \,\, r_i,\, \forall (i,t)\in E.$
\item $\alg \geq \sum_{t}\lambda_t +\sum_{(i,t)\in E} c_i \theta_{it}.$
\end{enumerate}
Then, by weak duality \alg\ is $\alpha$--competitive. 
For this certificate, we demonstrate that the standard dual fitting approach for finding a feasible dual solution  
cannot be used to obtain tight guarantees for stochastic (or even fluid) reusability. 

In the standard procedure for defining a candidate solution to the dual, we initialize all dual variables to 0 and then update the values in tandem with the matching decisions made by the online algorithm \citep{buchbind,devanur}. For simplicity, consider a deterministic online algorithm \alg. When \alg\ matches arrival $t$ to resource $i$, the reward $r_i$ is split into two parts and the dual variables are updated as follows,
\[\lambda_t=r_i\alpha_{it};\quad \sum_{\tau\in T}\theta_{it}=\sum_{\tau\in T}\theta_{it}+\frac{r_i}{c_i}(1-\alpha_{it}),\]
here $\alpha_{it}\in[0,1]$ decides how the gain of $r_i$ is split between $\lambda_t$ and $\sum_{\tau\in T}\theta_{it}$. This procedure ensures that at the end of the planning horizon we have,
\[\sum_{t\in T}\lambda_t+\sum_{i\in I}\sum_{t\in T}c_i \theta_{it}=\alg. \]
The values of splitting parameters $\alpha_{it}$ need to be carefully tuned to obtain the desired performance guarantee. We show that when usage durations are stochastic there is, in general, no setting of the splitting parameters that can be used to certify the actual performance of Balance.

For family of instances given in Example \ref{ega2}, to show that Balance has total reward $\alpha \opt$ using this approach, we need to find splitting parameters that satisfy the following system of inequalities,
\begin{eqnarray*}
&& \lambda_{t} +\sum_{\tau \in[2n]\backslash [t-1]}  \theta_{1,\tau} + \sum_{\tau \in [3n]\backslash[2n]} 0.5\theta_{1,\tau} \geq \alpha\quad \forall t\in[2n]\\
&& \lambda_{t} +\theta_{i,t}+\sum_{\tau \in\{t+1,\cdots,(3+1.5(i-1))n\}} 0.5 \theta_{i,\tau} \geq \alpha\quad \forall i\in\{1,2\},\, t\in\{2n+1,\cdots,3n\}\\
&& \lambda_{t} +\sum_{\tau \in\{t,\cdots,4.5n\}}  \theta_{2,\tau} \geq \alpha\quad \forall t\in\{3n+1,\cdots,4.5n\}\\
&& \lambda_t + \theta_{i_t}\geq \alpha r_{i_t}\quad \forall t\in[4.5n]\\
&& \lambda_t, \theta_{it}\geq 0, 
\end{eqnarray*}
here we use a single variable $\theta_{i_t}$ for low reward resource $i_t$ since resource $i_t$ does not have an edge to any arrival other than $t$. Note that for $t\in\left([2n]\backslash[n]\right)\cup\left(\{3.5n,\cdots, 4.5n\}\right)$, the low reward resources have price 0 and can be ignored. Now, under a fluid version of reusability, Balance does not match any arrival to low reward resources (with non-zero price). Based on the standard dual fitting procedure, we set, 
\[\theta_{i_t}=0\quad \forall t\in[4.5n].\]
Thus, $\lambda_t\geq \alpha r_{i_t}$ for all $t\in[4.5n]$. Similarly, we have $\lambda_t=0$ for $t\in[2n]\backslash[n]$ and $t\geq 3.5n+1$. Now, consider a simplified and reduced system of inequalities,
\begin{eqnarray*}
\theta_{1,2n}+ \sum_{\tau \in [3n]\backslash[2n]} 0.5\theta_{1,\tau}&\geq &\alpha,\\
\theta_{1,t}+\sum_{\tau \in [3n]\backslash[t]} 0.5\theta_{1,\tau}&\geq &\alpha-\lambda_t\quad \forall t\in\{2n+1,\cdots,3n\},\\
\theta_{2,4.5n}&\geq &\alpha,\\
\lambda_t&\geq & \alpha r_{i_t}\quad \forall t\in[4.5n],\\
\sum_{t\in [4.5n]}\lambda_t +n\sum_{i\in \{1,2\}, t\in [4.5n]} \theta_{it}&=&2.5n.
\end{eqnarray*}
For $\alpha=0.626$, it can be verified that this system does not have a feasible solution, i.e., there is no setting of the splitting parameters $\alpha_{it}$ that gives a feasible solution to the inequalities above.}

\section{Validity of Generalized Certificate}\label{appx:certificate}
\begin{repeatlemma}[Lemma \ref{certificate}.]
	Given an online algorithm \alg, non-negative values $\{\lambda_t(\omega,\nu)\}_{t,\omega,\nu}$ and $\{\theta_i\}_{i}$ such that conditions \eqref{cert1} and \eqref{cert2} hold, 
we have
\[\alg \geq \frac{\min_{i\in I} \alpha_i}{\beta}  \opt.\] 
\end{repeatlemma}
\begin{proof}{Proof.}
We start by summing both sides of constraint \eqref{cert2} over $i\in I$, 
\begin{eqnarray}
\sum_{i\in I}\alpha_i\, \opt_i \leq \sum_i \Big(\theta_i + 
\eod \Big[ \sum_{ t\in O(\omega,i)}\lambda_t(\omega,\nu) \Big]\Big) 
&&=\sum_{i\in I}\theta_i + \eod \Big[ \sum_{i\in I} \sum_{  t\in O(\omega,i)}\lambda_t(\omega,\nu) \Big]\nonumber \\
&&\overset{(a)}{\leq} \sum_{i\in I}\theta_i + \eod\Big[\sum_{t\in T} \lambda_t(\omega,\nu)\Big]\nonumber\\
&&\overset{(b)}{=}\sum_{i\in I}\theta_i +\sum_{t\in T} \eod\Big[ \lambda_t(\omega,\nu)\Big]\nonumber\\
&&\leq \beta \alg. \nonumber
\end{eqnarray}
Inequality $(a)$ follows from the fact that $\{O(\omega,i)\}_{i\in I}$ is collection of disjoint subsets of $T$ and also that $\lambda_t(\omega,\nu)\geq0\,\, \forall t\in T$. 
Equality $(b)$ follows by exchanging the order of the sum and expectation. The final inequality follows from \eqref{cert1}. 
\hfill\Halmos\end{proof}
\subsection{Tightness of LP free Certificate}\label{appx:lpfreedual}
Let $\omega$ denote a sample path w.r.t.\ all the randomness in \opt. Let $\nu$ denote a sample path w.r.t.\ all the randomness in \alg.	
The LP free certificate is given by the following inequalities,
\begin{eqnarray*}
\sum_{t\in T} \eod[\lambda_t(\omega,\nu)] + \sum_{i\in I} \theta_i&\leq &\beta \alg\\
\eod\left[\sum_{t\in O(\omega,i)} \lambda_t(\omega,\nu )\right] +\theta_i&\geq &\alpha \opt_i\quad \forall i\in I,\\
&&\lambda_t,\, \theta_i\geq 0.
\end{eqnarray*}
We view the linear system above as an LP with a trivial objective of minimizing 0 (a constant). 
The dual of this LP is,
\begin{eqnarray*}
\max && -\beta\alg\,y+\alpha\,\sum_{i\in I}  \opt_i\, x_i\\
\text{s.t. }&& 0\leq x_i\leq y\quad \forall i\in I,\\
&& y\geq 0,
\end{eqnarray*}
here we use the fact that for any given arrival $t$ and sample paths $\omega$ and $\nu$, there is at most one set $O(\omega,i)$ that contains $t$.
The optimal value of this LP is $\max_{y\geq 0}\left[(-\beta\alg+\alpha\,\sum_{i\in I} \opt_i)\, y\right]$. Thus, from strong duality we have that $\alg\geq \frac{\alpha}{\beta} \opt$, if and only if our LP-free system has a feasible solution. This implies that our certificate is tight, i.e., if $\alg$ has a competitive ratio guarantee of $\gamma\in(0,1]$, then there exists a feasible solution to our linear system with $\alpha_i=1\, \forall i\in I$ and $\beta=1/\gamma$. 
\section{Missing Details from Proof of Theorem \ref{main}}
\subsection{Properties of $(F,\bsigma,\mb{p})$ Random Process}\label{appx:randomprocess}
\begin{repeatlemma}[Lemma \ref{monotone}.] 
Given a distribution $F$, arrival set $\bsigma=\{\sigma_1,\cdots,\sigma_T\}$, and probability sequences $\mb{p}_1=(p_{11},\cdots,p_{1T})$ and $\mb{p}_2=(p_{21},\cdots,p_{2T})$ such that, $p_{1t}\leq p_{2t}$ for every $t\in[T]$, we have,
\[r(F_i,\bsigma,\mb{p}_1)\leq r(F_i,\bsigma,\mb{p}_2) .\]
\end{repeatlemma}
\begin{proof}{Proof.} 
Suppose the resource is available at some arrival $\sigma_t\in\bsigma$. Recall that independent of all other randomness in the random process, w.p.\ $p_t$ we switch the resource to in-use at $\sigma_t$ and w.p.\ $1-p_t$ the resource stays available till at least the next arrival. Consider an alternative random process where given a set $\bsigma$ and probability sequence $\mb{p}$, we first sample a random subset $\bsigma_{\mb{p}}$ of $\bsigma$ as follows: independently for each arrival $\sigma_t\in\bsigma$, include the arrival in the subset w.p.\ $p_t$. Taking expectation over this random sampling we claim that,
\[\mathbb{E}[r(F_i,\bsigma_{\mb{p}})]=r(F_i,\bsigma,\mb{p}).\]
This is a direct implication of the fact that an available resource is switched to unavailable independent of other events. 

Now, consider random processes $(F_i,\bsigma,\mb{p}_1)$ and $(F_i,\bsigma,\mb{p}_2)$ as given in the statement of the lemma. We establish the main claim by using the alternative viewpoint defined above to couple the two processes. To be more precise, we couple the subset sampling stage in the two processes. First, sample a random subset $\bsigma_{\mb{p}_1}$ by including each arrival $\sigma_t$ with corresponding probability $p_{1t}$. Next, sample subset $\bsigma_{\mb{p}_2\backslash \mb{p}_1}$ of $\bsigma$ by independently including arrival $\sigma_t$ w.p.\ $\frac{1}{1-p_{1t}}(p_{2t}-p_{1t})$, for every $\sigma_t\in \bsigma$. Finally, let $\bsigma_{\mb{p}_2}=\bsigma_{\mb{p}_1}\cup \bsigma_{\mb{p}_2\backslash \mb{p}_1}$. Since for every $t$, $\sigma_t\in \bsigma_{\mb{p}_2}$ with probability $p_{2t}$ independent of other points in $\bsigma_{\mb{p}_2}$, we have,
\[\mathbb{E}[r(F_i,\bsigma_{\mb{p}_2})]=r(F_i,\bsigma,\mb{p}_2) \text{ and } \mathbb{E}[r(F_i,\bsigma_{\mb{p}_1})]=r(F_i,\bsigma,\mb{p}_1).\]
Observe that $\bsigma_{\mb{p}_1}\subseteq \bsigma_{\mb{p}_2}$ on every sample path. Thus, to finish the proof it suffices to argue that $r(F,\bsigma_{1})\leq r(F,\bsigma_{2})$ if $\bsigma_1\subseteq \bsigma_2$.
Consider $(F,\bsigma_{1})$ and $(F,\bsigma_{2})$ random processes and the straightforward coupling of usage durations where we have a list of i.i.d.\ samples drawn according to distribution $F$ and each process independently parses this list in order, moving to the next sample whenever the current sample is used and never skipping samples. On any coupled path, the number of transitions made on arrivals in $\bsigma_2$ is lower bounded by the number of transitions made on arrivals in $\bsigma_1$, giving us the desired. 
\hfill\Halmos\end{proof}

\begin{repeatlemma}[Lemma \ref{zeroset}.]
	Given an $(F,\bsigma,\mb{p})$ random process, let $\bsigma' \subset \bsigma$ be a subset of arrivals where the resource is unavailable with probability 1. 
Then, 
at every arrival $\sigma_t\in \bsigma$, the probability that the resource is available at $\sigma_t$ is identical in $(F,\bsigma,\mb{p})$ and $(F,\bsigma,\mb{p}\vee \mb{1}_{\bsigma'})$. 
\end{repeatlemma}
\begin{proof}{Proof.}
It suffices to show the lemma for a subset $\bsigma'$ consisting of a single arrival. The result for general $\bsigma'$ then follows by repeated application. Now, let $\sigma_t$ denote an arbitrary arrival in $\bsigma$ such that w.p.\ $1$, the resource is in-use when $\sigma_t$ arrives. Observe that we can change the probability $p_t$ associated with $\sigma_t$ arbitrarily, but this does not change the probability of resource being available at $\sigma_t$. In particular, w.p.\ 1, the resource is in-use when $\sigma_t$ arrives in the $(F,\bsigma,\mb{p}\vee \mb{1}_t)$ random process as well. Consequently, the probabilities at other arrivals are unchanged in going from $(F,\bsigma,\mb{p})$ to $(F,\bsigma,\mb{p}\vee \mb{1}_t)$.
\hfill\Halmos\end{proof}
\begin{repeatlemma}[Lemma \ref{equiv}.]
The following statements are true for every $(F,\bsigma,\mb{p})$ random process:
\begin{enumerate}[(i)]
	\item For every $\sigma_t\in \bsigma$, the probability that the resource is available at $\sigma_t$ equals the fraction of the resource available at $\sigma_t$ in the fluid $(F,\bsigma,\mb{p})$  process.
	\item The expected reward $r(F,\bsigma,\mb{p})$, equals the total reward in the fluid $(F,\bsigma,\mb{p})$ process.
\end{enumerate} 

\end{repeatlemma}
\begin{proof}{Proof.}
Note that, statement $(ii)$ of the lemma is a direct consequence of statement $(i)$. The proof of $(i)$ hinges on the fact that in the $(F,\bsigma,\mb{p})$ random process, the duration of every state transition is independent of past randomness. Using this we write a recursive equation for the probability of reward at every arrival.  Let $\eta(\sigma_t)$ denote the probability that the resource is available when $\sigma_t\in \bsigma$ arrives. We have, 
\[ \eta(\sigma_t)=\eta(\sigma_{t-1})\big(1-p_{t-1}\big)+\sum_{\tau=1}^{t-1} \eta(\sigma_{\tau}) p_{\tau}\big(F(\sigma_t-\sigma_{\tau})-F(\sigma_{t-1}-\sigma_{\tau})\big), \]
where $\eta(\sigma_1)=1$. By forward induction, it is easy to see that this set of equations has a unique solution. {\color{black} Now, consider the fluid $(F,\bsigma,\mb{p})$ process and let 
$\eta'(\sigma_t)$ denote the fraction of resource available at $\sigma_t$ in the fluid process. Clearly, $\eta'(\sigma_1)=1$ and we have,
\[ \eta'(\sigma_t)=\eta'(\sigma_{t-1})\big(1-p_{t-1}\big)+\sum_{\tau=1}^{t-1} \eta'(\sigma_{\tau}) p_{\tau}\big(F(\sigma_t-\sigma_{\tau})-F(\sigma_{t-1}-\sigma_{\tau})\big). \]
Thus, 
\[\eta(\sigma_t)=\eta'(\sigma_t) \quad \forall t\leq T.\]
The expected reward from a match occurring at $\sigma_t$ in the $(F,\bsigma,\mb{p})$ random process is $p_t\eta(\sigma_t)$, which, using the equality above, is equal to the fraction of resource consumed at $\sigma_t$ in the fluid process.} 
\hfill\Halmos\end{proof}

\subsection{Applying Chernoff and Missing Pieces of Theorem \ref{main}}\label{appx:chernoff}
\begin{repeatlemma}[Lemma \ref{chernoff}.]
Given integer $\tau>0$, real value $c>0$, independent indicator random variables $\onee(t)$ for $t\in [\tau]$ and $\delta=\sqrt{\frac{\log c}{c}}$ such that, $\sum_{t=1}^{\tau} \mathbb{E}\left[\onee(t)\right]\leq \frac{c}{1+\delta}.$
We have,
\[\mathbb{P}\left(\sum_{t=1}^{\tau}\onee(t) \geq c \right)\leq \frac{1}{\sqrt{c}}. \]
\end{repeatlemma}
\begin{proof}{Proof.}			
Let $\mu=\frac{c}{1+\delta}$. 
Let the condition on total mean hold with equality, i.e.,
\[\sum_{t=1}^{\tau} \mathbb{E}\left[\onee(t)\right]= \mu\]	
This is w.l.o.g., as we can always add some number of dummy independent binary random variables to make the condition hold with equality. 
{\color{black}			Applying the standard multiplicative Chernoff bound for independent Bernoulli random variables $\{\onee(t)\}_{t\in[\tau]}$, we have
\[ \mathbb{P}\Big( \sum_{t=1}^{\tau}\onee(t) \geq (1+\delta)\mu = c\Big)\leq e^{-\frac{\mu\delta^2}{2+\delta}}<\frac{1}{\sqrt{c}},\]
here $\frac{\mu\delta^2}{2+\delta}=\frac{\log c}{1+2\delta+\delta^2}\geq 0.5\log c$ for $c\geq 1$.}
\hfill\Halmos\end{proof}

\section{From Matching to Multi-unit Assortments}\label{appx:asst}
In this section we generalize the $(1-1/e)$ result for the following model. 

\noindent \textbf{Online Assortments with Multi-unit demand:} Customer $t$ requires up to $b_{it}\geq 0$ units of resource $i\in I$. Given an assortment $S$, the customer chooses at most one resource from $S$ with probabilities given by choice model $\phi_t$. 
Let $y_i(t)$ denote the number of units of resource $i$ when $t$ arrives. Selection of resource $i$ results in $\min\{y_i(t),\, b_{it}\}$ units of resource $i$ being used for an independently drawn duration $d\sim F_i$, and a reward $\min\{y_i(t),\, b_{it}\}r_i$ (results hold even if each of the $b_{it}$ units is used for an independent random duration). The assortment $S$ that we offer must belong to a downward closed feasible set $\mathcal{F}_t$. Choice model $\phi_t$ and quantities $b_{it}$ are revealed when $t$ arrives. 

Similar to the online matching problem, we focus on the large capacity regime. Due to multi-unit (budgeted) allocations this is more accurately the large budget to bid ratio regime, where the parameter,
\[\gamma:=\min_{i\in I,\, t\in T} \frac{c_i}{ b_{it}}, \text{ approaches } +\infty.\]
We compare online algorithms against a clairvoyant algorithm that knows the choice models and quantities $b_{it}$ for all arrivals in advance but makes assortment decisions in order of the arrival sequence and observes (i) realizations of customer choice after showing the assortment and (ii) realizations of usage duration when used units return (same as an online algorithm). Further, in case of assortments we make the standard assumptions (\cite{negin,RST18,reuse}) that for every arrival $t\in T$, 
choice model $\phi_t$ satisfies the weak substitution property, i.e.,
\begin{equation}\label{subs}
\phi_t(S,i)\geq \phi_t(S\cup\{j\},i),\quad \forall i,j\not\in S,\, \forall t\in T.
\end{equation}
We also assume access to an assortment optimization oracle that takes a choice model $\phi$ and set of feasible solutions $\mathcal{F}$ as input and outputs a feasible revenue maximizing assortment 
More generally, an $\alpha$ approximate oracle is also acceptable and in this case the competitive ratio guarantee is $(1-1/e)\,\alpha$.

Since we now need to think in terms of sets of resources offered to arrivals, a relatively straightforward way to generalize \galg\ will be to fractionally ``match" every arrival to a collection of revenue maximizing assortments/sets, consuming constituent resources in a fluid fashion in accordance with the choice probabilities. 

\begin{algorithm}[h]
\SetAlgoLined
\textbf{Output:} For every arrival $t$, collection of assortments and probabilities $\{A(\eta,t),u(\eta,t)\}_{\eta}$\;
Let $g(t)=e^{-t}$, and 
initialize $Y(k_i)=1$ for every $i\in I,k_i\in[c_i]$\;
\For{every new arrival $t$}{
For every $k_i\in[c_i]$ and $t\geq 2$, update values 
\[ 
Y(k_i)=Y(k_i)+\sum_{\tau=1}^{t-1} \Big(F_i\big(a(t)-a(\tau)\big)-F_i\big(a(t-1)-a(\tau)\big)\Big)y(k_i,\tau)
\] \text{ {\color{black}\tcp{Fluid\,\, update of returning capacity}}}

Initialize $S_t=\{i \mid (i,t)\in E\}$, values $\eta=0$, and $y(k_i,t)=0$ for all $i\in S_t,k_i\in [c_i]$\;
\While{$\eta<1$ and $S_t\neq \emptyset$}{
\For{$i\in S_t$}{
\textbf{if}  $Y(k_i)=0$ for every $k_i\in[c_i]$ \textbf{ then} remove $i$ from $S_t$\;
\textbf{else }$ z_i=\underset{k_i\in[c_i]}{\arg\max}\, \{k_i \mid Y(k_i)>0\}; \text{ {\color{black}\tcp{Highest\,\, available\,\, unit} }} $	
}

$ A(\eta,t)=\underset{S\in S_t }{\arg\max} \, \sum_{i\in S} b_{it}r_i \phi_t(S,i)\left(1-g\left(\frac{z_i}{c_i}\right)\right)$ 
{\color{black}\tcp{Optimal\,\, assortment\,\, with\,\, \dpg}}\

$u(\eta,t)= \min\left\{1-\eta, \min_{i\in A(\eta,t)} \frac{Y(z_i)}{b_{it}\phi_t(A(\eta,t),i)}\right\}$ \text{{\color{black} \tcp{Fractional assortment}}}
\ \\~\\

Update $\eta\to \eta + u(\eta,t)$\;

\For{$i\in A(\eta,t)$}{
Update $y(z_{i},t)\to y(z_{i},t)+  u(\eta,t)\,b_{it}\,\phi_t(A(\eta,t),i)$; \quad $Y(z_i)\to Y(z_i)-y(z_i,t)$\; \text{{\color{black} \tcp{Inventory update after fluid customer choice}}}
}
}}	
\caption{\astgalg}
\label{astgalg}
\end{algorithm}
\noindent \textbf{Description of \astgalg\ (Algorithm \ref{astgalg}):}	Observe that the stochasticity due to choice has been converted to its fluid version. 
Specifically, arrival $t\in T$ is fractionally ``matched" to assortments $A(1,t),\cdots,A(m,t)$ for some $m\geq 0$. The weight/fraction of assortment $A(j,t)$ is given by $u(j,t)>0$ and we have, $\sum_{j=1}^m u(j,t)\leq 1$. The amount of resource $i$ (fluidly) consumed as a result of this is given by $\sum_{A(j,t) \ni i} u(j,t)\, \phi\big(A(j,t),i\big)$. The collection of assortments is found by computing the revenue maximizing assortment with reduced prices computed according to \dpg\ rule, as in the case of matching. The values $y(k_i,t)$ in the algorithm correspond to the total fraction of unit $k_i$ that is fluidly chosen by arrival $t$. The weights $u(j,t)$ are chosen to ensure that $y(k_i,t)$ does not exceed the fraction of $k_i$ that was available when $t$ arrived. We assume w.l.o.g.\ that the oracle that outputs revenue maximizing assortments never includes resources with zero probability of being chosen in the assortment. 
Interestingly, the performance guarantee of the relaxed online algorithm \astgalg, depends only on $c_{\min}=\min_{i\in I} c_i$. 
\begin{lemma}
For every instance of the online budgeted assortment problem we have,
\[\astgalg \geq (1-1/e)\,e^{-\frac{1}{c_{\min}}} \, \opt.\]
\end{lemma}
\begin{proof}{Proof.}
Note that the sample path $\omega$ now also includes the randomness due to customer choice. Let $O(\omega,i)$ denote the set of all arrivals on sample path $\omega$ in \opt\ where some units of $i$ are \emph{chosen}. Since each arrival chooses (possibly multiple units of) \emph{at most one} resource, we interpret $O(\omega,i)$ as the set of arrivals that choose $i$. Let $b(\omega,i,t)$ denote the number of units of $i$ chosen in \opt\ at arrival $t\in O(\omega,i)$.  
Let $z_i(t^+)$ be the highest index unit of resource $i$ that has a non-zero fraction available in \astgalg\ at $t^+$. We use the generalized certification with sample path based variables $\lambda_t(\omega)$. Given a sample path $\omega$ and resource $i\in I$, we set,
\begin{eqnarray}
\lambda_t(\omega)=  b(\omega,i,t) r_i \bigg(1-g\Big(\frac{z_i(t^+)}{c_i}\Big)\bigg),\label{astlambda}
\end{eqnarray}
for every arrival $t\in O(\omega,i)$. Let $\lambda_t=\eo[\lambda_t(\omega)]$. Recall that $y(k_i,t)$ is the total fraction of unit $k_i$ that is fluidly chosen by arrival $t$. By definition of \astgalg, we have 
\[\lambda_t\leq  \sum_{i\in I,\, k_i\in[c_i]} y(k_i,t)\, r_i \left(1-g\left(\frac{k_i}{c_i}\right)\right),\]
i.e., $\lambda_t$ is at most the expected total reward at $t$ in \astgalg\ calculated with units at their reduced price. Setting $\theta_i$ as before (in the proof of Lemma \ref{galgvopt}), i.e.,
\[\theta_i= c_i\left(e^{\frac{1}{c_i}}-1\right)r_i\sum_{t} \sum_{k\in[c_i]} y(k_i,t) g\Big(\frac{k}{c_i}\Big),\] 
we have that condition \eqref{cert1} of the certificate is satisfied with $\beta=e^{1/c_{\min}}$. 

The remaining analysis now mimics the proof of Lemma \ref{galgua}. We generalize the basic setup to demonstrate this formally. Let $\onee(\neg k,t^+)$ indicate that no fraction of unit $k$ is available in \astgalg\ right after $t$ has been matched. Recall that $\Delta g(k)= g\big(\frac{k-1}{c_i}\big)-g\big(\frac{k}{c_i}\big)$. 
By definition of $\lambda_t$ (see \eqref{astlambda}),  we have 
\begin{eqnarray*}
\eo\Big[\sum_{t\in O(\omega,i)} \lambda_t(\omega) \Big]\geq (1-1/e)\opt_i - r_i \eo\Big[\sum_{t\in O(\omega,i)}b(\omega,i,t) \sum_{k\in[c_i]} \Delta g (k) \onee(\neg k,t^+)\Big].
\end{eqnarray*}
Fix an arbitrary unit $k_O$ of $i$ and let $O(\omega,k_O)$ denote the set of arrivals on sample path $\omega$ in \opt\ where $k_O$ is one of the chosen units of $i$. Note that $O(\omega,k_O)$ is a subset of $O(\omega,i)$. Using the decomposition above, it suffices to show that, 
\begin{eqnarray*}
\eo\Big[\sum_{t\in O(\omega,k_O)}\sum_{k\in[c_i]}\Delta g(k) \onee(\neg k,t^+)\Big]\leq \frac{1}{c_ir_i}\theta_i.
\end{eqnarray*}
The proof of this inequality follows Lemma \ref{galgua} verbatim. Crucially, we have the following inequalities that complete the proof,
\[	\eo\Big[\sum_{t\in O(\omega, k_O)} \onee(\neg k,t^+)\Big] \leq r(F_i,\mb{s}(k))\leq r(F_i,\mb{T},\mb{p}(k)),\]
where $\mb{s}(k),\mb{T},$ and $\mb{p}(k)$ are as defined in Lemma \ref{galgua}, i.e., $\mb{s}(k)$ is the ordered set of all arrivals $t$ (arrival times $a(t)$ to be precise) such that $\onee(\neg k,t^+)=1$ in \astgalg. $\mb{T}$ is the ordered set of all arrivals. Finally, probabilities $p(k,t)\in\mb{p}(k)$ are defined as follows: $p(k,t)=0$ if $y(k,t)=0$, otherwise $p(k,t)= \frac{y(k,t)}{\eta(k,t)}$, where $\eta(k,t)$ is the fraction of $k$ available in \astgalg\ when $t$ arrives. 
\hfill\Halmos\end{proof}
The main new challenge in turning \astgalg\ into \astalg\ is that we must deal with scenarios where \astgalg\ directs some mass towards a set $A$ but only a subset of resources in $A$ are available in \salg. Recall that in case of matching, if the randomly chosen resource is unavailable we simply leave $t$ unmatched. We could consider a similar approach here whereby if any unit of sampled set $A$ is unavailable then we do not offer $A$. However, this will not preserve the overall revenue in expectation as the probability of every resource in $A$ being available simultaneously is likely to be small. 
If it were acceptable to offer an assortment with items that are not available in \salg\ then we could also offer the set $A$ as is. The underlying assumption in such a case is that if the arrival chooses an unavailable item then we earn no reward and the arrival simply departs (called static substitution in \cite{ma2021dynamic}). However, in many applications it may not be possible or desirable to offer an assortment where some items are unavailable. 

An alternative approach is to offer the subset $S$ of $A$ that is available in \salg\ at $t$. However, this can affect the choice probability for resources $i\in S\cap A$ in non-trivial ways, and thus, affect future availability of resources in a way that is challenging to control. In other words, the concentration bounds that show \salg\ has the same performance as \galg\ for large inventory, will not apply here. 
Consequently, we need to find a way to display some subsets of $A$ such that, (i) the overall probability of any given resource being allocated is no larger than if we offered $A$ itself and (ii) we do not rely on multiple resources in $A$ being available simultaneously. The main novelty of our approach to tackle this problem will be to switch our perspective from sets of resources back to individual resources. Specifically, for each resource we find the overall probability that the resource is chosen by a given arrival and then use these probabilities as our guideline, i.e., given the subset $S\subseteq A$ of resources that is available, we find a new collection of assortments so that for every available resource, the overall probability of the resource being chosen matches this probability in the original collection of assortments in \astgalg. The main idea here is a probability matching, made non trivial by the fact that we are restricted to choice probabilities given by the choice model. We find an iterative polytime algorithm (Algorithm \ref{pmatch}) that 
ensures that the probability of a resource being chosen by an arrival in \salg, matches that in \astgalg. 

Recall that $\gamma:=\min_{(i,t)\in E} \frac{c_i}{ b_{it}}$. We are interested in the case where $\gamma\to +\infty$. Note that we assume knowledge of a lower bound on $\gamma$ in \astalg. Overloading notation, we denote this lower bound also as $\gamma$.

\begin{algorithm}[h]
\SetAlgoNoLine
Initialize capacities $y_i(0)=c_i$ and let $\delta=\sqrt{\frac{\log \gamma}{\gamma}}$\;
\For{every new arrival $t$}{
Update capacities $\{y_i(t)\}_{i\in I}$ for resources with returning units and let $S_t=\{i\mid y_i(t)\geq b_{it}\}$\;
Get collection of assortments $\{A(\eta,t),u(\eta,t)\}_{\eta}$ from \astgalg\; 
Randomly sample collection $\eta$ w.p.\ $u(\eta,t)$\; 
For sampled $\eta$, let $\mathcal{\hat{A}},\, \mathcal{\hat{U}}=\text{Probability Match } \Big(A(\eta,t)\cap S_t,\{
\frac{1}{1+\delta}\phi_t(A(\eta,t),s)\}_{s\in A(\eta,t)\cap S_t }\Big)$\;
Randomly sample asssortment $\hat{A}_j\in \mathcal{\hat{A}}$ with distribution $\mathcal{\hat{U}}$\; 
\text{{\color{black} \tcp{Assortment may be empty with non-zero probability}}}

Offer $\hat{A}_j$ to $t$ and update capacity after $t$ chooses\;
}
\caption{\astalg}
\label{astalg}
\end{algorithm}

\begin{repeatlemma}(Lemma \ref{probmatch} restated)
Consider a choice model $\phi: 2^{N}\times N \to [0,1]$ satisfying the weak substitution property (see \eqref{subs}), an assortment $A\subseteq N$ belonging to a downward closed feasible set $\mathcal{F}$, a subset $S\subseteq A$ and target probabilities $p_s$ such that, $p_s\leq \phi(A,s)$ for every  $s\in S$. There exists a collection $\mathcal{A}=\{A_,\cdots,A_m\}$ of $m=|S|$ assortments along with weights $\big(u_i\big) \in[0,1]^{m}$, such that the following properties are satisfied:
\begin{enumerate}[(i)]
\item For every $i\in[m]$, $A_i \subseteq S$ and thus, $A_i\in \mathcal{F}$.
\item Sum of weights, $\sum_{i\in[m]} u_i \leq 1$.
\item For every $s\in S$, $\sum_{A_i\ni s} \, u_i\, \phi(A_i,s) = p_s$.
\end{enumerate}
Algorithm \ref{pmatch} computes such a collection $\mathcal{A}$ along with weights $(u_i)$ in $O(m^2)$ time.

\end{repeatlemma}
\begin{proof}{Proof of Lemma \ref{probmatch}.}
We give a constructive proof that outlines Algorithm \ref{pmatch} introduced earlier. 
Let, 
\[q^0_s=\phi( S,s) \text{ and }\zeta^0_s=\frac{p_s}{q^0_s} \text{ for every $s\in S$.}\] 
Observe that $q^0_s\geq \phi(A,s)\geq p_s$, due to substitutability. Thus, $\zeta^0_s\leq 1$ for every $s\in S$.

Let $s_1$ be an element in $S$ with the smallest value $\zeta^0_{s_1}$. 
Let $A_1=S$ be the first set added to collection $\mathcal{A}$ with $u_1=\zeta^0_{s_1}$, so that $u_1 \phi(A_1,s_1)=p_{s_1}$. We will ensure that all subsequent sets added to $\mathcal{A}$ do not include the element $s_1$ and this will guarantee condition $(iii)$ for element $s_1$. 
Next, define the set $S^1=S\backslash \{s_1\}$. Let,
\[q^1_s=\phi(S^1,s)\geq q^0_s\text{ and }\zeta^1_{s}=\frac{p_s-u_1q^0_s}{q^1_s} \text{  for every $s\in S^1$.}\]
Observe that $\zeta^1_s\in[0,1]$ for every $s\in S^1$. Let $s_2$ denote the element with the smallest value $\zeta^1_{s_2}$, out of all elements in $S^1$. 
If $\zeta^1_{s_2}=0$ we stop, otherwise we now add the second set $A_2=S^1$ to the collection with $u_2=\zeta^1_{s_2}$. Inductively, after $i$ iterations of this process, we have added $i$ nested sets $A_i\subset A_{i-1}\subset \cdots\subset A_1$ to the collection and have the remaining set $S^{i}= A_{i}\backslash \{s_{i}\}$ of $|A|-i$ elements. Define values, 

\[q^{i}_s=\phi(S^{i},s)\text{ and }\zeta^{i}_s=\frac{p_s-\sum_{k=1}^{i}u_kq^{k-1}_{s}}{q^{i}_s} \text{ for every $s\in S^i$.}\]
Let $s_{i+1}\in S^i$ be the element with the smallest value $\zeta^i_{s_{i+1}}$. If $\zeta^i_{s_{i+1}}>0$, we add the set $A_{i+1}=S^i$ to the collection with $u_{i+1}=\zeta^i_{s_{i+1}}$ and continue.

Clearly, this process terminates in at most $m=|S|$ steps, resulting in a collection of size at most $m$. Each step involves updating the set of remaining elements, computing the new values $\zeta^{(\cdot)}_s$ and finding the minimum of these values. Thus every iteration requires at most $O(m)$ time and the overall algorithm takes at most $O(m^2)$ time. Due to the nested nature of the sets and downward closedness of $\mathcal{F}$, condition $(i)$ is satisfied for every set added to the collection. It is easy to verify that condition $(iii)$ is satisfied for every element by induction. We established the base case for element $s_1$ in the first iteration. Suppose that the property holds for all elements $s_1,s_2,\cdots,s_{i-1}$. Then, by the definition of $u_{i}$ we have for element $s_i$,
\[\sum_{j=1}^{i} u_j \phi(A_j,s_i) =  \big(p_s- \sum_{j=1}^{i-1} u_j q^{j-1}_{s_i} \big)+\sum_{j=1}^{i-1} u_j q^{j-1}_{s_i}=p_{s_i}. \]
Since $s_i$ is excluded from all future sets added to the collection, this completes the induction for $(iii)$. Finally, to prove property $(ii)$ it suffices to show that,
\[u_{m}\leq \zeta^0_{s_{m-1}}-\sum_{i=1}^{m-1} u_i, \]
as this immediately implies, $\sum_{i\in [m]}u_i\leq \zeta^0_{s_{m}}\leq 1$.
The desired inequality follows by substituting $u_{m}$ and using the following facts: (i) $q^{j}_{s_m}$ is non-decreasing in $j$ due to substitutability, (ii) $u_i\geq 0$ for every $i\in[m]$ since we perform iteration $i$ only if $u_i=\zeta^{i-1}_{s_i}>0$. Therefore,
\[ u_{m}=\frac{p_{s_m}-\sum_{i=1}^{m-1}u_{i}q^{i-1}_{s_m}}{q^{m-1}_{s_m}}\leq \frac{p_{s_{m}}-q^{0}_{s_{m}}\sum_{i=1}^{m-1}u_i}{q^{0}_{s_{m}}} =\zeta^0_{s_{m}}-\sum_{i=1}^{m-1}u_i.\]
\hfill	\Halmos\end{proof}

The probability matching algorithm be executed more efficiently for the commonly used MNL choice model. Using properties of MNL it suffices to sort the resources in order of values $\zeta^0_{s}$ in the beginning and this ordering does not change as we remove more and more elements. Each iteration only takes $O(1)$ time and so the process has runtime dominated by sorting a set of size $m=|S|$, i.e., $O(m\log m)$.

\begin{proof}{Proof of Theorem \ref{main2}.}

Recall that $\gamma=\min_{(i,t)\in E} \frac{c_i}{ b_{it}}$, and $y(k_i,t)$ is the total fraction of unit $k_i$ that is fluidly chosen by arrival $t$ in \astgalg. The proof rests simply on showing that for every $i\in I$ and $t\in T$, at least $c_i/\gamma$ $(\geq b_{it})$ units of $i$ are available at $t$ w.p.\ at least $1-1/\gamma$. Conditioned on this, from Lemma \ref{probmatch} we have that in \astalg, $i$ is offered to and chosen by arrival $t$ w.p.\ $\frac{1}{b_{it}(1+\delta)}\sum_{k \in [c_i]}y(k,t)$. This implies a lower bound of $\frac{1-1/\gamma}{1+\delta}\sum_{k\in[c_i]} y(k,t)$ on the expected reward from $t$ choosing $i$, completing the proof. 

Let $x_{it}:=\frac{1}{b_{it}(1+\delta)}\sum_{k \in [c_i]}y(k,t)$. To show that $i$ is available at $t$ w.h.p., we first draw out some hidden independence in the events of concern. Recall that if less than $b_{it}$ units of $i$ are available at $t$ then \astalg\ does not offer $i$ in any (randomized) assortment to $t$. Otherwise, Probability Matching (Algorithm \ref{pmatch}) ensures that $i$ is \emph{offered and chosen} by arrival $t$ w.p.\ exactly $x_{it}$. Now, consider the following alternative process at every arrival, 
\begin{enumerate}[1.]
\item Given collection $\{A(\eta,t),u(\eta,t)\}_{\eta}$ from \astgalg, sample assortment $A(\eta,t)$ independently w.p.\ $\frac{1}{1+\delta}u(\eta,t)$.
\item Offer $A(\eta,t)$ and it customer chooses resource $i$ with insufficient inventory, reject the customer request.
\end{enumerate} 
We refer to this alternative process as the \emph{static} process. Observe that the static process is probabilistically identical to \astalg, which first checks the inventory of resources and then offers an assortment (after running probability matching). Thus, it suffices to show that in the static process, for every $i\in I$ and $t\in T$, at least $c_i/\gamma$ units of $i$ are available at $t$ w.p.\ at least $1-1/\gamma$. 

Now, consider an arbitrary resource $i$ in the static process and let $\onee(i\to t)$ indicate the event that $i$ is offered to and chosen by $t$. 
W.l.o.g., we independently pre-sample usage durations for every possible match and let $\onee(d_t> a(\tau)-a(t))$ indicate that the duration of usage pre-sampled for (a potential) match of $i$ to arrival $t$ is at least $a(\tau)-a(t)$. Let $\gamma_i=\frac{c_i}{\max_{t'\in T}b_{it'}}$. Observe that $\gamma_i\leq \gamma\,\,\forall i\in I$. 
The static process never fails to satisfy customer request for resource $i$ if, 

\[\sum_{\tau=1}^{t}b_{i\tau} \onee(i\to \tau)\onee(d_{\tau}> a(t)-a(\tau)) \leq  c_i(1-1/\gamma_i) \quad \forall t\in T. \]
Define binary (not necessarily Bernoulli) random variables $X_{\tau}= \frac{b_{i\tau}}{\max_{t'\in T}b_{it'}}\onee(i\to \tau)\onee(d_{\tau}> a(t)-a(\tau))$ for all $\tau\leq t-1$. Random variables $X_{\tau}$ are independent of each other as the assortment sampled, customer choice, and the duration of usage are all independently sampled at each arrival (in the static process). From \astgalg, we have the following upper bound on the total expectation, 
\[\mu:=\mathbb{E}\Big[\sum_{\tau=1}^{t} X_{\tau}\Big]= \sum_{\tau=1}^{t} \frac{b_{i\tau}}{\max_{t'\in T}b_{it'}}\, x_{i\tau} \big(1-F_i( a(t)-a(\tau))\big)\leq\frac{\gamma_i}{1+\delta}.\] 
Applying the generalized Chernoff bound as stated in Lemma \ref{chernoff2}, 
completes the proof.
\hfill\Halmos\end{proof}
\begin{lemma}\label{chernoff2}
Given integer $\tau>0$, real value $\gamma>0$, independent binary random variables $X_t\in\{0,x_t\}$ with $x_t\in(0,1]\,\, \forall t\in [\tau]$ and $\sum_{t=1}^{\tau} \mathbb{E}\left[X_t\right]\leq \frac{\gamma}{1+\delta}$, where $\delta=\sqrt{\frac{\log \gamma}{\gamma}}$.
We have,
\[\mathbb{P}\left(\sum_{t=1}^{\tau}X_t > \gamma-1 \right)\leq \frac{1}{\sqrt{\gamma}}. \]
\end{lemma}
\begin{proof}{Proof.} Consider Bernoulli random variables $Y_t$ such that $E[Y_t]=x_tp_t\,\, \forall t\in [\tau]$ and $\sum_{t=1}^\tau E[Y_t] \leq \frac{\gamma}{1+\delta}$. From Lemma \ref{chernoff}, we have
\[\mathbb{P}\left(\sum_{t=1}^{\tau}Y_t > \gamma-1 \right)\leq \frac{1}{\sqrt{\gamma}}.\]
The key step in deriving the Chernoff bound that underlies this inequality is the following upper bound on the moment generating function of $Y_t$ \citep{goemans}. For $s\geq 0$,
\begin{equation}\label{moment}
E[e^{sY_t}]=(x_tp_t)\, e^{s}+(1-x_tp_t)=1+x_tp_t\,(e^{s}-1)\leq e^{x_tp_t\,(e^s-1)}.
\end{equation}			
Using this upper bound along with Markov's inequality and independence of random variables, 
gives the desired bound. 
Thus, to show the desired bound for independent (non-Bernoulli) random variables $X_t\in\{0,x_t\}$ (where $E[X_t]=x_tp_t$), it suffices to establish the upper bound on moment generating functions given by \eqref{moment}. For $s\geq 0$,	we have
\[E[e^{sX_t}]=p_t\,e^{sx_t}+(1-p_t)=1+p_t\,(e^{sx_t}-1).\]
Unlike \eqref{moment}, here we have the term $x_t$ in the exponent. However, for $s\geq 0$, and $x_t,\, p_t\in[0,1]$, we have, $p_t\,(e^{sx_t}-1)\leq e^{p_tx_t\,(e^s-1)}-1$, giving us the same upper bound as \eqref{moment}. This completes the proof.
\hfill\Halmos\end{proof}

\section{Missing Details for \dpg}\label{appx:rbamissing}

\subsection{Boundedness of Common IFR Distributions}\label{ifrexamples}
Given,
\[L_i(\epsilon)=\max_{ x\geq 0}\frac{F_i\big(x+F_i^{-1}(\epsilon)\big)-F_i(x)}{\epsilon},\]
we show that \dpg\ is asymptotically $(1-1/e)$--competitive if $L_i(\epsilon)\epsilon\to$ as $\epsilon\to 0$. In particular, if $\epsilon L_i(\epsilon)$ decreases to 0 as strongly as $O(\epsilon^{\eta})$ for some $\eta>0$, then we have a polynomial convergence rate of $\tilde{O}(c_i^{-\frac{\eta}{1+\eta}})$. We show that this holds for many common IFR families (Chapter 2 of \cite{barlow}). An approximation for $L(\epsilon)\epsilon$ that eases calculations is given by $\{F^{-1}(\epsilon)\max_{x\geq 0}f(x)\}$ and $F(x)\to O(1)x f(x)$ for small $x$.
\begin{itemize}
\item \textbf{For exponential, uniform and IFR families with non-increasing density:} it is easy to see that $L_i(\epsilon)=1$ and thus, $\eta=1$. This implies a $\tilde{O}(c_i^{-0.5})$ convergence.
\item \textbf{Weibull distributions:} This family is characterized by two non-negative parameters $\lambda,k$ with c.d.f.\ given by, $F(x)=1-e^{(-x/\lambda)^k}$ for $x\geq 0$. The family is IFR only for $k\geq1$. It is easy to see that,
\[L(\epsilon)\epsilon = O(1) \epsilon^{1/k}. \]
Hence, for any finite $\lambda, k$ we have $\tilde{O}(c_i^{-\frac{1}{1+k}})$ convergence to $(1-1/e)$.
\item \textbf{Gamma distributions:} This family is given by non-negative parameters $k,\theta$ such that the c.d.f.\ is $F(x)=\frac{1}{\Gamma(k)}\gamma(k,x/\theta)$. Here $\Gamma,\gamma$ are the upper and lower incomplete gamma functions resp..\ The family is IFR only for $k\geq 1$. It can be shown that for small $\epsilon$, $L(\epsilon)\epsilon\to O(1)\epsilon^{1/k}$. Thus, we have $\tilde{O}(c_i^{-\frac{1}{1+k}})$ convergence for any finite set of parameter values.
\item \textbf{Modified Extreme Value Distributions:} This family is characterized by the density function, $f(x)=\frac{1}{\lambda}e^{-\frac{e^x -1}{\lambda}+x}$ where parameter $\lambda>0$. For small $\epsilon$ we have that $L(\epsilon)\epsilon\to O(1)\epsilon$, for any finite $\lambda$. Giving a a $\tilde{O}(c_i^{-0.5})$ convergence.	
\item \textbf{Truncated Normal:} For truncated normal distribution with finite mean $\mu$, finite variance $\sigma^2$, and support $[0,b]$ where $b$ could $\to\infty$, it is easy to see that $L(\epsilon)=O(1)$,  leading to a $\tilde{O}(c_i^{-0.5})$ convergence. Note that if the support is given by $[a,b]$ for some $a>0$, then we do not have convergence as $L(\epsilon)\epsilon\to O(1)$. 
\end{itemize}

\subsection{Detailed Overview of Analysis of \dpg}\label{appx:rbaoverview}
We devote this section to setting up the overall framework for analyzing \dpg\ and postpone detailed proofs to Appendix \ref{appx:rbaproofs}. 
We start by introducing (and recalling) important notation. We use $\omega$ to denote the sample path of usage durations in \opt\ and $O(\omega,i)$ to denote the set of all arrivals matched to $i$ in \opt\ on sample path $\omega$. Similarly, $\nu$ denotes a sample path of usage durations in \dpg.  Recall that \dpg\ tracks the highest available unit for each resource, denoted as $z_i(t)$ for resource $i$ at arrival $t$. In fact, since \dpg\ is realization dependent $z_i(t)$ is more accurately written as $z_i(\nu,t)$. Given set of resources $S_t$ with an edge to $t$, \dpg\ matches according to the following simple rule,
\[\argmax_{i\in S_t} r_i \left(1-g\left(\frac{z_i(t)}{c_i}\right)\right). \]
Let $D(t)$ (technically $D(\nu,t)$) denote the resource matched to $t$ by \dpg\ on sample path $\nu$. Let $z_{D(t)}$ denote the highest available unit of resource $D(t)$ at $t$. Finally, recall that $\Delta g(k)=g(\frac{k-1}{c_i})-g(\frac{k}{c_i})=e^{-\frac{k}{c_i}}(e^{\frac{1}{c_i}}-1)$.

Recall from \eqref{theta}, 
we use the following candidate solution for the LP free certificate, 
\begin{eqnarray*}
\theta_i = r_i \, \ed\Big[\sum_{t\mid D(t)=i} g\Big(\frac{z_i(t)}{c_i}\Big) \Big],\, \forall i\in I 
\text{   and   }
\lambda_t= \ed \Big[ r_{D(t)} \Big(1-g\Big(\frac{z_{D(t)}}{c_{D(t)}}\Big)\Big)\Big], \, \forall t\in T. 
\end{eqnarray*}
Similar to the analysis of \galg\ in Section \ref{sec:solution}, we defined $\lambda_t$ and $\theta_i$ as deterministic quantities. Since $\lambda_t$ is independent of $\omega$ and $\nu$, it suffices to use the simplified system of inequalities given by \eqref{scert1} and \eqref{scert2}. The candidate solution above satisfies condition \eqref{scert1} by definition, with $\beta=1$.  To prove conditions \eqref{scert2} are satisfied with $\alpha_i=(1-1/e)$ for every $i\in I$, we start by 
lower bounding the term $\eo[ \sum_{t\in O(\omega,i)}\lambda_t]$. The following lemma brings out the key source of difficulty. A proof of the lemma is in Appendix \ref{appx:rbaproofs}. 


\begin{repeatlemma}[Lemma \ref{decompose1}.] Given $\lambda_t$ as defined by \eqref{lambda}, we have for every resource $i$,
\begin{equation}
\eo\Big[\sum_{t\in O(\omega,i)} \lambda_t\Big]\geq (1-1/e) \opt_i - r_i \eo\Big[ \sum_{t\in O(\omega,i)} \sum^{c_i}_{k=1} \Delta g (k) \pd \big[k> z_i(t)\big]\Big]. \label{lam4c}
\end{equation}
\end{repeatlemma}

Consider a resource $i\in I$ and let $k$ denote a unit of $i$. Contrast decomposition \eqref{lam4c} with its counterpart given in Lemma \ref{decompose}. 
Deterministic binary valued quantities $\onee(\neg k,t^+)$, that signified unavailability of a unit $k$ after arrival $t$ is matched in \galg, have been replaced by (non-binary) probabilities $\pd[k>z_i(t)]$. For any given arrival $t$, probability $\pd [k> z_i(t)]$ represents the likelihood that unit $k$ and all higher units of $i$ are unavailable in \dpg\ when $t$ arrives. This dependence on other units, combined with the sensitivity of $z_i(t)$ to small changes on the sample path, makes it challenging to bound these probabilities in a meaningful way. 
To address these challenges we introduce two new ingredients. 
The first ingredient simplifies the stochastic dependencies by introducing a conditional version of the probability $\pd[k>z_i(t)]$. This ingredient uses special structure that is only available in settings beyond online matching {\color{black} (see Appendix \ref{sec:rbafail} for further discussion). The second ingredient builds on the first one and addresses the non-binary nature of these probabilities.}

\subsubsection*{Ingredient 1: Conditioning.}
Recall that to analyze \galg, we captured the allocation of each individual unit 
via a natural $(F,\bsigma,\mb{p})$ random process. In an $(F,\bsigma)$ random process, at every $\sigma_t\in \bsigma$ the only event of relevance is the availability of the item at $\sigma_t$ (as probability $p_t=1$). If the item is available, then it is matched to $\sigma_t$ regardless of the usage durations realized before $\sigma_t$. 
In order to naturally capture the actions of \dpg\ on individual units through a random process, a similar property must hold in \dpg. 
To be more specific, consider a single unit $k$ of $i$ in \dpg\ and fix the randomness associated with all units and resources except $k$. Suppose we have an arrival $t$ with edge to $i$ and two distinct sample paths over usage durations of $k$ prior to the arrival of $t$. We are also given that $k$ is available at $t$ on both sample paths. Clearly, \dpg\ makes a deterministic decision on each sample path. Is it possible that \dpg\ matches $t$ to $k$ on one sample path but not the other? 

Interestingly, in the case of matching we can show that this is impossible, i.e., \dpg\ always makes the same decision on both sample paths in such instances. We show that this is not the case in general. In particular, when we include the aspect of customer choice or budgeted allocations, reusability can interact with these elements in an undesirable way (see example in Section \ref{sec:rbafail}).  
For online matching we leverage this property to show that conditioned on the randomness of all other resources and units, the allocation of $k$ in \dpg\ is characterized by an $(F_i,\bsigma)$ random process. 

Let us formalize the discussion above. First, since usage durations are independently sampled we generate $\nu$ as follows: For each edge, we independently sample a usage duration for every unit of the resource that the edge is incident on. \dpg\ only sees the samples that correspond to units that are matched and does not see the other samples. Thus, sample path $\nu$ includes a set of $c_i$ usage durations for every edge $(i,t)$. If \dpg\ matches $t$ to unit $k$ on $\nu$ then, the usage duration annotated for $k$ is realized out of the $c_i$ samples for edge $(i,t)$. 
Let $\nu_k$ be the collection of all samples annotated for unit $k$ on path $\nu$. Thus, $\nuk$ is a sample path of usage durations of $k$. 
Let $\nu_{-k}$ denote the sample path of all units and resources except unit $k$ of $i$. We use $\mathbb{P}_{\mb{\nu_{k}}}$ and $\mathbb{E}_{\mb{\nu_{k}}}$ to denote probability and expectation over randomness in sample paths $\nuk$. 
Define indicators,
\[
\onee(k,t) = \begin{cases}
1\quad  \text{ if unit $k$ is available when $t$ arrives},\\ 
0\quad  \text{ otherwise.}
\end{cases}\]
Now, consider the conditional version of probability $\pd [k> z_i(t)]$ in \eqref{lam4c},
\begin{eqnarray}
\mathbb{P}_{\mb{\nu_{k}}}[k> z_i(t) \mid \nu_{-k}]=\mathbb{E}_{\mb{\nu_{k}}}[1-\onee(k,t) \mid \nu_{-k}] \times \prod_{k'>k}^{c_i} [(1-\onee(k',t))\mid \nu_{-k}]. \label{probtoexp}
\end{eqnarray}
For every $k'\neq k$, the conditional indicator $\onee(k',t)\mid \nu_{-k}$ is deterministic. 
$\mathbb{E}_{\mb{\nu_{k}}}[\onee(k,t)\mid \nu_{-k} ]$ is the likelihood of $k$ being available at $t$ conditioned on fixing the usage durations of all other units. 
We show that it suffices to condition the expectation only on a certain subset of units as opposed to all units. This subset consists of all units that strictly precede $k$ in the ordering defined below. 

\emph{Order of Units}: {\color{black}Unit $k_i$ of resource $i$ precedes unit $k_j$ of resource $j$, denoted $k_i\succ k_j$, iff $r_i (1-g(k_i/c_i))>r_{j} (1-g(k_j/c_j))$, i.e., \dpg\ would prefer to match $k_i$ rather than $k_j$ if both were available.} In case of a tie we let the unit with the lower resource index precede and let \dpg\ follow the same tie breaking rule. Observe that this ordering is transitive.

{\color{black}	Let $D(t)$ denote the unit matched to $t$ by \dpg. Given edge $(i,t)$ and a unit $k$ of $i$, let $\onenu[k\succeq D(t)]$ be the indicator for the event $k\succeq D(t)$ in \dpg. Notice that if $\onenu[k\succeq D(t)]=1$, then on sample path $\nu$, all units with an edge to $t$ that precede $k$ are unavailable when $t$ arrives. If $k$ is available at $t$ and $k\succeq D(t)$, then we have $D(t)=k$. 

\begin{lemma}\label{indep}
For every edge $(i,t)$ and unit $k$ of $i$, the random variable $\onenu[k\succeq D(t)]$ is independent of the usage duration of $k$ and all units that succeed $k$. 
\end{lemma}
}

See Appendix \ref{appx:rbaproofs} for the proof.  Let $k^+$ denote all units that strictly precede $k$ and 
let $\nukp$ 
denote the sample path of usage durations for these units. 
As a consequence of Lemma \ref{indep}, we have the following properties for \dpg.
{\color{black}	\begin{corollary}\label{probsimp}
Given arrival $t$ with edge to $i$ and a unit $k\in[c_i]$,  we have 
$\mathbb{P}_{\mb{\nu_{k}}}[k> z_i(t) \mid \nu_{k^+}]=\mathbb{P}_{\mb{\nu_{k}}}[k> z_i(t) \mid \nu_{-k}].$ 
\end{corollary}
The corollary follows by using Lemma \ref{indep} to transform the RHS of equation \eqref{probtoexp}. See Appendix \ref{appx:rbaproofs} for the proof.
}
\begin{corollary}\label{sigma}
For every unit $k$, given a sample path $\nukp$ we have an ordered set of arrivals $\mb{\sigma}(\nukp)=\{\sigma_1,\cdots,\sigma_e\}$ with $\sigma_1<\cdots<\sigma_e$ such that for any arrival $t$,
\begin{itemize}
\item If $t\not\in \mb{\sigma}(\nukp)$, then conditioned on $\nukp$, the probability that $k$ is matched to $t$ is 0.
\item If $t\in\mb{\sigma}(\nukp)$, then conditioned on $\nukp$, $k$ is matched to $t$ w.p.\ 1 if it is available.
\end{itemize} 
Therefore, conditioned on $\nukp$, 
the $(F_i,\bsigma(\nukp))$ random process fully characterizes the matches of $k$ in \dpg. 
\end{corollary}
{\color{black} See Appendix \ref{appx:rbaproofs} for the proof.} As an immediate consequence of Corollary \ref{sigma}, $r(F_i,\mb{\sigma}(\nukp))$ gives the expected number of times $k$ is matched in \dpg\ conditioned on sample sample path $\nukp$. We remind the reader that technically the set $\bsigma$ in an $(F,\bsigma)$ random process is a set of arrival times. We are using arrivals and arrival times interchangeably for convenience.
{\color{black}	\begin{corollary}\label{indepcoro}
For every arrival $t$ with an edge to unit $k$ of $i$, conditioned on $\nukp$ and also on the availability of unit $k$ at $t$, the event $D(t)=k$ 
is independent of the (past) usage durations of $k$.
\end{corollary} 
The corollary follows from the fact that conditioned on $\nukp$, the $(F_i,\mb{\sigma}(\nukp))$ random process fully characterizes the matches of $k$ in \dpg. }
\subsubsection*{Ingredient 2: Covering.}
If the probabilities $\pd[k>z_i(t)\mid \nukp]$ were all either 1 or 0, 
the rest of the analysis could proceed along the same lines as the analysis of \galg. 
Of course, in general these probabilities have non-binary values and are hard to meaningfully bound. To address this challenge we introduce the ingredient of \emph{covering}. For intuition behind the idea, consider an instance where \dpg\ matches all $n$ units of a resource to arrivals in an infinitesimal interval $[0,\epsilon]$. Each unit returns after 1 unit of time with probability 0.5 and never returns with probability 0.5. 
Let $\onee(k,2)$ indicate that unit $k$ of the resource is available at time $2$. We have,
\begin{eqnarray*}
\pd[\onee(k,2)]=\frac{1}{2}
\text{ and therefore, } \pd[k>z_{ 1}(2)]<\frac{1}{2^k}.
\end{eqnarray*} 
This implies that $z_{1}(2)\geq n-O(\log n)$ w.h.p..\ Therefore, the fact that all units are available at $t$ with a constant probability implies that the probability $\pd[k>z_{1}(2)]$ decreases geometrically as $k$ increases, resulting in a small value ($O(\log n)$) for the negative term in \eqref{lam4c}.  Inspired by this, for every $k$ and sample path $\nukp$, we 
classify all arrivals $t$ into two groups. Roughly speaking, the first group consists of arrivals where $k$ is available with sufficiently high probability. These are called \emph{uncovered} arrivals. We show that the overall contribution to \eqref{lam4c} from this group of arrivals, summed appropriately over all units $k$, is well approximated by the expectation of a geometric random variable. This generalizes the bound that we observed in the example above.  
The remaining \emph{covered} arrivals, are the ones where the probability that $k$ is available is low. 
We show that the contribution from these terms to \eqref{lam4c} is effectively canceled out by $\theta_i$ 
, i.e., these arrivals are \emph{covered} by $\theta_i$. 

To describe this formally, for values $\epsilon_i\in(0,1]$,
let $\mathcal{X}_k(\nu_{k^+},\epsilon_i,t)\in\{0,1\}$ denote the \emph{covering function} that performs this classification.
\medskip

\fbox{\begin{minipage}{16.2cm}
Fix unit $k$, a value $\epsilon_i\in(0,1]$, and condition on sample path $\nukp$. Given an arrival $t$, we say that $t$ is uncovered and set 
$\mathcal{X}_k(\nu_{k^+},\epsilon_i,t)=0$ if,
\[\mathbb{P}_{\mb{\nu_{k}}}[\onee(k,t)=1 \mid \nu_{k^+}]\geq 
\epsilon_i, \]
i.e., at an uncovered arrival we have a lower bound of $\epsilon_i$ on the probability of $k$ being available. Equivalently, this condition imposes a lower bound of $\epsilon_i$ on the probability of $k$ being free at time $a(t)$ in an $(F_i,\mb{\sigma}(\nukp))$ random process. 
\end{minipage}}
\smallskip

Notice that as a direct consequence of the above definition, for every covered arrival $t$ we have $\mathcal{X}_k(\nu_{k^+},\epsilon_i,t)=1$ and $\mathbb{P}_{\mb{\nu_{k}}}[\onee(k,t)=1 \mid \nu_{k^+}]< \epsilon_i$.
For ease of notation we omit the input $\epsilon_i$ in covering functions and use the abbreviation,
\[\mathcal{X}_k(\nukp,t),\]
with the understanding that $\epsilon_i$ is present and will be set later to optimize the result. Note that a given arrival $t$ may be covered w.r.t.\ one unit $k$ of $i$ but uncovered w.r.t.\ another unit $k'$ of $i$. 

\subsubsection*{Combining the Ingredients:}	The following lemma uses the notions of covering and conditioning to further decompose and upper bound \eqref{lam4c}. Proof of the lemma is deferred to Appendix \ref{appx:rbaproofs}.
\begin{lemma} \label{decompose2}
For every resource $i$,
\[\eo\Big[ \sum_{t\in O(\omega,i)} \sum^{c_i}_{k=1} \Delta g (k) \pd \big[k> z_i(t)\big]\Big]\leq \mathbb{E}_{\mb{\omega}}\Big[\sum_{t\in O(\omega,i)} \Big(\sum^{c_i}_{k=1} \Delta g (k) \mathbb{E}_{\mb{\nu_{k^+}} } \big[\mathcal{X}_k(\nu_{k^+},t) \big]\Big)\Big] + \frac{2}{r_i\epsilon_ic_i} \opt_i . \]
\end{lemma}

\noindent \emph{Remarks:} The term $\frac{2}{r_i\epsilon_ic_i} \opt_i $ on the RHS is an upper bound on the contribution of the terms arising due to uncovered arrivals. As we discussed when motivating the second ingredient, this bound is a result of geometrically decreasing probability of multiple uncovered units being jointly unavailable at $t$. 
The formal proof of this bound relies crucially on the property that the usage durations of units $k'\prec k$ do not affect the matching decision of \dpg\ for unit $k$ (the first ingredient). 

The other term on the RHS of inequality in Lemma \ref{decompose2} captures the contribution from covered arrivals. 
As we informally stated earlier, we shall upper bound this term by $\theta_i$, i.e.,
\begin{equation}\label{tangled}
\eo\Big[\sum_{t\in O(\omega,i)} \Big(\sum^{c_i}_{k=1} \Delta g (k) \mathbb{E}_{\mb{\nu_{k^+}} } \big[\mathcal{X}_k(\nu_{k^+},t) \big]\Big)\Big]\leq \frac{1+o(1)}{r_i}\theta_i,
\end{equation}
where $o(1)$ represents a term that goes to 0 as $c_i\to +\infty$. 
An obstacle to proving this inequality is the dependence on set $O(\omega,i)$, which is determined by actions of \opt. At a high level, we tackle this difficulty in the same way as the proof of Lemma \ref{galgua}. We untangle the dependence on \opt\ in the LHS of \eqref{tangled} by upper bounding it as the expected reward of an $(F,\bsigma)$ random process. From Corollary \ref{sigma} we have a way of lower bounding $\theta_i$ by the expectation of a (different) random processes. Consequently, showing \eqref{tangled} boils down to proving a new perturbation property of $(F,\bsigma)$ random processes.  

\subsubsection*{Upper bounding the LHS in \eqref{tangled}:} Consider the 
set $\mb{s}(\nukp)$ of all arrivals $t$ that are covered i.e., \[\mb{s}(\nukp)=\{t\in T \mid \mathcal{X}_k(\nu_{k^+},t)=1\}.\]  
The following lemma shows that the quantity,
$r\big(F_i,\mb{s}(\nukp)\big)$ is an upper bound on the contribution from covered arrivals to \eqref{lam4c}. 

\begin{lemma}\label{transport}For any resource $i$, unit $k$, and path $\nukp$, we have,
\[\mathbb{E}_{\mb{\omega}}\Big[\sum_{t\in O(\omega,i)} \mathcal{X}_k(\nu_{k^+},t) \Big]	\leq c_i r\big(F_i,\mb{s}(\nukp)\big).\]
Consequently,
\[\eo\Big[\sum_{t\in O(\omega,i)} \Big(\sum^{c_i}_{k=1} \Delta g (k) \mathbb{E}_{\mb{\nu_{k^+}} } \big[\mathcal{X}_k(\nu_{k^+},t) \big]\Big)\Big]\leq \Big(1+\frac{2}{c_i}\Big) \sum^{c_i}_{k=1} \Big(g \Big(\frac{k}{c_i}\Big) \cdot \mathbb{E}_{\mb{\nu_{k^+}} } \big[r\big(F_i, \mb{s}(\nukp) \big) \big]\Big).\]		
\end{lemma}
The proof of this lemma is included in Appendix \ref{appx:rbaproofs} and closely mimics the two list coupling argument used to establish the inequality 
$	\eo\Big[\sum_{t\in O(\omega, k_O)} \onee(\neg k,t^+)\Big] \leq r(F_i,\mb{s}(k)),$
in the analysis of \galg. To finish the overall analysis of \dpg, we need the following new property for random processes.

\begin{proposition}[{Perturbation property}]\label{relatexp}For every resource $i$, every unit $k$ of $i$, every path $\nukp$, and parameter $\epsilon_i=\frac{1}{o(c_i)}$, we have
\begin{equation}\label{conjecture}
\boxed{	r\big(F_i, \mb{s}(\nukp) \big) \leq \big(1+\kappa_i( \epsilon_i,c_i)\big) r\big(F_i,\mb{\sigma}(\nukp) \big), }
\end{equation}
for some non-negative function $\kappa_i$ that approaches 0 for $c_i\to +\infty$.
\end{proposition}
\noindent Recall that given set of arrivals $T$, subset $\bsigma(\nukp)$, and value $\epsilon_i$, the set $\mb{s}(\nukp)$ of covered arrivals is fully determined by the $\big(F_i,\bsigma(\nukp)\big)$ random process. Thus, inequality \eqref{conjecture} is purely a statement about an explicitly defined random process. 
Note, the requirement that $\epsilon_i=\frac{1}{o(c_i)}$ or $\frac{1}{\epsilon_ic_i} \to 0$ for $c_i\to +\infty$, is clearly a necessity due to the term $\frac{2}{r_i\epsilon_ic_i} \opt_i$ in Lemma \ref{decompose2}. 
In general, the difficulty in proving \eqref{conjecture} is that we require $\epsilon_i$ to be suitably large and also desire $\kappa_i$ to be small. Larger $\epsilon_i$ leads to more arrivals $\mb{s}(\nukp)\backslash \bsigma(\nukp)$, leading to a (possibly) larger value of $r\big(F_i, \mb{s}(\nukp) \big)$. This tension between the two parameter values is made explicit through the competitive ratio guarantee in the next lemma.


\begin{lemma}\label{master}
Given, values $\theta_i$ and $\lambda_t$ set according to \eqref{theta}. For every usage distribution such that Proposition \ref{relatexp} holds, 
we have
\[\theta_i + \eo\Big[ \sum_{t\in O(\omega,i) }\lambda_t\Big] \geq \alpha_i \opt_i,\quad \forall i\in I,\]
with values $\alpha_i=\frac{\big(1-1/e-\frac{2}{\epsilon_i \, c_i}\big)}{ \big(1+\kappa_i( \epsilon_i,c_i)\big)(1+2/c_i)}$. If we can choose $\epsilon_i,\kappa_i$ such that, $\frac{1}{\epsilon_i\,c_i}\to 0$ and $\kappa_i( \epsilon_i,c_i)\to 0$ for $c_i\to +\infty$, then \dpg\ is asymptotically $(1-1/e)$--competitive, with the rate of convergence given by, \[O\Big(\min_i \Big\{\frac{1}{\epsilon_i\,c_i}+\kappa_i(\epsilon_i,c_i) \Big\}\Big).\] 
\end{lemma}
\begin{proof}{Proof.}
For convenience, we refer to $\kappa_i( \epsilon_i,c_i)$ simply as $\kappa_i$ respectively. By definition of $\theta_i$ in \eqref{lambda},
\begin{eqnarray*}
\frac{1}{r_i}\theta_i &= &\ed\Big[\sum_{t\mid D(t)=i} g\Big(\frac{z_i(t)}{c_i}\Big) \Big],\\
&= &\sum_{k=1}^{c_i} g\Big(\frac{k}{c_i}\Big)\sum_{t\mid D(t)=i,\, z_i(t)=k} 1,\\
&= &\sum_{k=1}^{c_i} g\Big(\frac{k}{c_i}\Big) \mathbb{E}_{\mb{\nu_{k^+}} } \big[r\big(F_i, \mb{\sigma}(\nukp) \big) \big],
\end{eqnarray*}
where the last equality follows from Corollary \eqref{sigma}. Combining this with Lemma \ref{decompose2}, Lemma \ref{transport} and Proposition \ref{relatexp}, we have,
\begin{eqnarray*}
\eo\Big[ \sum_{t\in O(\omega,i)} \sum^{c_i}_{k=1} \Delta g (k) \pd \big[k> z_i(t)\big]\Big]
&\leq & \Big(1+\frac{2}{c_i}\Big)(1+\kappa_i)\, \frac{1}{r_i}\theta_i + \frac{2}{r_i\epsilon_ic_i} \opt_i.\nonumber 
\end{eqnarray*}
Substituting this in Lemma \ref{decompose1}, we get,
\[	\eo\Big[ \sum_{t\in O(\omega,i) }\lambda_t\Big]\geq \Big(1-1/e -\frac{2}{\epsilon_ic_i}\Big)\opt_i - \Big(1+\frac{2}{c_i}\Big)(1+\kappa_i)\, \theta_i.\]
\hfill	\Halmos
\end{proof}
The final step in proving Theorem \ref{genmatch} is to establish the Perturbation property \eqref{conjecture} for the various families of usage distributions discussed in Section \ref{sec:rba}. In Appendix \ref{appx:prop1}, we prove \eqref{conjecture} for these families of distributions, starting with the simplest case of two-point distributions. 
These proofs rely on using the knowledge of distribution $F_i$ to further characterize the set $\mb{s}(\nukp)$, followed by carefully constructed coupling arguments that also critically use the structure of $F_i$.  

Note that to prove \eqref{conjecture}, we will typically show a stronger statement by considering a set $\mb{S}(\nukp)$ that is a superset of $\mb{s}(\nukp)$ and show that,
\[r\big(F_i,\mb{S}(\nukp) \big) \leq \big(1+\kappa_i( \epsilon_i,c_i)\big) r\big(F_i,\mb{\sigma}(\nukp) \big), \]
for a suitably small $\kappa_i$.
Inequality \eqref{conjecture} then follows by using the monotonicity property (Lemma \ref{monotone}). This will give us freedom to consider sets $\mb{S}(\nukp)$ that are easier to characterize than $\mb{s}(\nukp)$. In absence of any assumptions about the distribution $F_i$, it is not formally clear to us if \eqref{conjecture} still holds, though we conjecture that it does.

\subsection{Proofs of Individual Components}\label{appx:rbaproofs}


\begin{repeatlemma}[Lemma \ref{decompose1}.]
Given $\lambda_t$ as defined by \eqref{lambda}, we have for every resource $i$,
\[\eo\Big[\sum_{t\in O(\omega,i)} \lambda_t\Big]\geq (1-1/e) \opt_i - r_i \eo\Big[ \sum_{t\in O(\omega,i)} \sum^{c_i}_{k=1} \Delta g (k) \pd \big[k> z_i(t)\big]\Big]. \]
\end{repeatlemma}
\begin{proof}{Proof.}
We start with the LHS,
\begin{eqnarray}
\eo\Big[\sum_{t\in O(\omega,i)} \lambda_t\Big]&&= \eod \Big[ \sum_{t\in O(\omega,i)} r_{D(t)} [1-g\Big(\frac{z_{D(t)}}{c_{D(t)}}\Big)]\Big],\nonumber\\
&&\overset{(a)}{\geq} \eod \Big[ \sum_{t\in O(\omega,i)} r_i [1-g\Big(\frac{z_{i}(t)}{c_{i}}\Big)]\Big],\nonumber\\
&&\overset{(b)}{=}r_i \eod \Big[ \sum_{t\in O(\omega,i)} \Big(1-1/e- \sum_{k=z_i(t)+1}^{c_i} \Delta g (k)\Big)\Big],\nonumber\\
&&\overset{}{=}(1-1/e) \opt_i -r_i\eod \Big[ \sum_{t\in O(\omega,i)} \sum_{k=z_i(t)+1}^{c_i} \Delta g (k)\Big],\label{lam2'}
\end{eqnarray}
where $(a)$ follows by the fact that \dpg\ matches every arrival to the resource that has maximum reduced price. {\color{black} Equation $(b)$ follows from the inequality, 
\[1-g\Big(\frac{z_{i}(t)}{c_{i}}\Big)=(1-1/e)- \sum_{k=z_i(t)+1}^{c_i} \Delta g (k),\]
obtained by setting $k_i=z_i(t)$ in \eqref{delg2} (which states one of the properties of function $\Delta g(\cdot)$).
Let $\onee_\nu(k>z_i(t))$ indicate the event that $k>z_i(t)$. We rewrite the $\eod[\cdot]$ term in \eqref{lam2'} as follows, 
\begin{eqnarray}
\eod \Big[\sum_{t\in O(\omega,i)} \sum_{k=z_i(t)+1}^{c_i} \Delta g (k)\Big]
&& =\eo\Big[\sum_{t\in O(\omega,i)} \ed \, \Big[ \sum^{c_i}_{k=z_i(t)+1} \Delta g (k)\Big]\Big],\nonumber\\
&&=\eo\Big[\sum_{t\in O(\omega,i)} \ed \, \Big[ \sum^{c_i}_{k=1} \Delta g (k)\, \onee_\nu(k>z_i(t)) \Big]\Big],\nonumber\\
&&=\eo\Big[ \sum_{t\in O(\omega,i)} \sum^{c_i}_{k=1} \Delta g (k) \pd \big[k> z_i(t)\big]\Big], \label{lam2''}
\end{eqnarray}
here the final equality follows by linearity of expectation. Combining \eqref{lam2'} and \eqref{lam2''} completes the proof. }
\hfill\Halmos\end{proof}
\begin{repeatlemma}[Lemma \ref{indep} (restated).]
Let $D(t)$ denote the resource matched to $t$ by \dpg. Given arrival $t$ with edge to $i$ let $\onenu[k\succeq D(t)]$ denote the event $k\succeq D(t)$ in \dpg. Then, the event $\onenu[k\succeq D(t)]$ is independent of the usage durations of $k$ and all units that succeed $k$.
\end{repeatlemma}
{\color{black}	\begin{proof}{Proof.}
Suppose that the statement is false for some unit $k$. Then, there exists sample paths $\nu_1$ and $\nu_2$ such that $\onee_{\nu_1}[k\succeq D(t)]=1$ and $\onee_{\nu_2}[k\succeq D(t)]=0$ for some arrival $t$ and $\nu_1$ and $\nu_2$ agree on the usage durations of all units except unit $\ell$, for some $\ell\prec k$. For every unit $k'\succeq k$, the usage durations of $k'$ are the same on both $\nu_1$ and $\nu_2$. Without loss of generality,
\[\onee_{\nu_1}[k'\succeq D(\tau)]=\onee_{\nu_2}[k'\succeq D(\tau)],\] 
for every $k'\succeq k$ and every arrival $\tau$ prior to $t$. Thus, the matching decisions over units $k'\succeq k$ are identical prior to arrival $t$ and when $t$ arrives, the availability status of units $k'\succeq k$ is the same on both sample paths. Now by definition of \dpg, if $t$ is matched to a unit $k'\succeq k$ on one path then it will be matched to the same unit on the other path, contradiction. 
\hfill\Halmos\end{proof}}

{\color{black}	\begin{repeattheorem}[Corollary \ref{probsimp}.]
Given arrival $t$ with edge to $i$ and a unit $k\in[c_i]$,  we have 
$\mathbb{P}_{\mb{\nu_{k}}}[k> z_i(t) \mid \nu_{k^+}]=\mathbb{P}_{\mb{\nu_{k}}}[k> z_i(t) \mid \nu_{-k}].$ 
\end{repeattheorem}}
{\color{black}	\begin{proof}{Proof.} 
Equality \eqref{probtoexp} states that,
\[
\mathbb{P}_{\mb{\nu_{k}}}[k> z_i(t) \mid \nu_{-k}]=\mathbb{E}_{\mb{\nu_{k}}}[1-\onee(k,t) \mid \nu_{-k}] \times \prod_{k'>k}^{c_i} [(1-\onee(k',t))\mid \nu_{-k}],
\]
where $\onee(k',t)\mid \nu_{-k}$ is deterministic for every $k'\neq k$.  In fact, $\onee(k',t)\mid \nu_{-k}=\onee(k',t)\mid \nukp\,\, \forall k'>k$.
From Lemma \ref{indep}, we have, $\mathbb{E}_{\mb{\nu_{k}}}[\onee(k,t) \mid \nu_{-k}]=\mathbb{E}_{\mb{\nu_{k}}}[\onee(k,t) \mid \nu_{k^+}].$ Thus,
\[\mathbb{P}_{\mb{\nuk}}[k> z_i(t) \mid \nu_{-k}]=\mathbb{E}_{\mb{\nuk}}[1-\onee(k,t) \mid \nukp] \times \prod_{k'>k}^{c_i} [(1-\onee(k',t))\mid \nukp] =\mathbb{P}_{\mb{\nu_{k}}}[k> z_i(t) \mid \nu_{k^+}].\]

\hfill\Halmos		\end{proof}
}
\begin{repeattheorem}[Corollary \ref{sigma}.]
For every unit $k$, given a sample path $\nukp$ we have an ordered set of arrivals $\mb{\sigma}(\nukp)=\{\sigma_1,\cdots,\sigma_e\}$ with $\sigma_1<\cdots<\sigma_e$ such that for any arrival $t$,
\begin{itemize}
\item If $t\not\in \mb{\sigma}(\nukp)$, then conditioned on $\nukp$, the probability that $k$ is matched to $t$ is 0.
\item If $t\in\mb{\sigma}(\nukp)$, then conditioned on $\nukp$, $k$ is matched to $t$ w.p.\ 1 if it is available.
\end{itemize} 
Therefore, conditioned on $\nukp$, 
the $(F_i,\bsigma(\nukp))$ random process fully characterizes the matches of $k$ in \dpg. 
\end{repeattheorem}
{\color{black}\begin{proof}{Proof.}
From Lemma \ref{indep} we have that the matching decisions over units $k^+$ (that strictly precede $k$) do not depend on usage durations of unit $k$ and its successors. We condition on $\nukp$ and this fixes the state of units $k^+$ at every arrival.  
Let $\mb{\sigma}(\nukp)$ denote the ordered set of arrivals that have an edge to $k$ and are not matched to a unit preceding $k$. If arrival $t$ has an edge to $k$ and $t\not\in \mb{\sigma}(\nukp)$, then $D(t)\succ k$ and $k$ is not matched to $t$. If $t\in\mb{\sigma}(\nukp)$, then $k\succeq D(t)$ by definition and $k$ is matched to $t$ if it is available at $t$. The availability of $k$ at $t\in\mb{\sigma}(\nukp)$ depends only on the usage durations of $k$ prior to $t$, which are sampled independently according to distribution $F_i$. 
Thus, the $(F_i,\bsigma(\nukp))$ random process fully characterizes the matches of $k$ in \dpg.

\hfill\Halmos	\end{proof}}
\begin{repeatlemma}[Lemma \ref{decompose2}.]
For every resource $i$,
\[\eo\Big[ \sum_{t\in O(\omega,i)} \sum^{c_i}_{k=1} \Delta g (k) \pd \big[k> z_i(t)\big]\Big]\leq \mathbb{E}_{\mb{\omega}}\Big[\sum_{t\in O(\omega,i)} \Big(\sum^{c_i}_{k=1} \Delta g (k) \mathbb{E}_{\mb{\nu_{k^+}} } \big[\mathcal{X}_k(\nu_{k^+},t) \big]\Big)\Big] + \frac{2}{r_i\epsilon_ic_i} \opt_i . \]
\end{repeatlemma}
\begin{proof}{Proof.}
We start with the observation, $\pd \big[k> z_i(t)\big]=\mathbb{E}_{\mb{\nu_{k^+}} } \big[\mathbb{P}_{\mb{\nu_{k}}}[k> z_i(t) \mid \nu_{k^+}]\big]$. Using the classification of each arrival $t$ afforded by the covering function, we get,
\begin{eqnarray}
\mathbb{P}_{\mb{\nu_{k}}}[k> z_i(t) \mid \nu_{k^+}]&&=  \mathcal{X}_k(\nu_{k^+},t)\,\, \mathbb{P}_{\mb{\nu_k}}[k>z_i(t)\mid \nu_{k^+}]\nonumber\\
&&\quad + (1-\mathcal{X}_k(\nu_{k^+},t))\,\, \mathbb{P}_{\mb{\nu_k}}[k>z_i(t)\mid \nu_{k^+}]\nonumber\\
&&\leq \mathcal{X}_k(\nu_{k^+},t) + (1-\mathcal{X}_k(\nu_{k^+},t))\cdot \mathbb{P}_{\mb{\nu_k}}[k>z_i(t)\mid \nu_{k^+}],\label{probdecomp}
\end{eqnarray}
here we use the upper bound of 1 on probability $\mathbb{P}_{\mb{\nu_k}}[k>z_i(t)\mid \nu_{k^+}]$ for the scenario where $t$ is covered\footnote{As an aside from the proof, one might wonder if this is too loose an upper bound. Conditioned on sample path $\nukp$, in general it can be shown that if $t$ is covered in $k$ then it is also covered in all units of $i$ preceding $k$, conditioned on corresponding partial sample paths that are consistent with $\nukp$. So the bound is reasonably tight.}. 
Now, by substitution,
\begin{eqnarray*}
&&\eo\Big[ \sum_{t\in O(\omega,i)} \sum^{c_i}_{k=1} \Delta g (k) \pd \big[k> z_i(t)\big]\Big]\leq \\
&&\eo\Big[\sum_{t\in O(\omega,i)} \Big(\sum^{c_i}_{k=1} \Delta g (k) \mathbb{E}_{\mb{\nu_{k^+}} } \big[\mathcal{X}_k(\nu_{k^+},t) \big]\Big)\Big]+\\
&&\eo\Big[\sum_{t\in O(\omega,i)} \Big(\sum^{c_i}_{k=1} \Delta g (k) \mathbb{E}_{\mb{\nu_{k^+}} } \big[(1-\mathcal{X}_k(\nu_{k^+},t))\mathbb{P}_{\mb{\nu_k}}[k>z_i(t)\mid \nu_{k^+}] \big]\Big)\Big].
\end{eqnarray*}
The first term in the inequality is as desired. The next lemma (Lemma \ref{geom}) shows that the second term is further upper bounded by $\frac{1}{r_i\epsilon_ic_i}\opt_i$.
\hfill\Halmos\end{proof}

{\color{black} Let us recall the definition of the covering function $\mathcal{X}_k$. Fix a resource $i$ and a value $\epsilon_i\in(0,1]$. Consider unit $k$ of $i$ and condition on sample path $\nukp$. Given an arrival $t$, we say that $t$ is uncovered and set 
$\mathcal{X}_k(\nu_{k^+},\epsilon_i,t)=0$ if,
\[\mathbb{P}_{\mb{\nu_{k}}}[\onee(k,t)=1 \mid \nu_{k^+}]\geq 
\epsilon_i, \]
i.e., at an uncovered arrival we have a lower bound of $\epsilon_i$ on the probability of $k$ being available. Due to the importance of parameter $\epsilon_i$ in the next lemma, we use the full form of the covering function without omitting $\epsilon_i$. 

\begin{lemma}\label{geom}
For any resource $i$ and arrival $t$, 
we have, 
\begin{eqnarray*} 
\sum_{k=1}^{c_i}\Delta g(k) \mathbb{E}_{\mb{\nu_{k^+}}} \big[ (1-\mathcal{X}_k(\nu_{k^+},\epsilon_i,t))\,\, \mathbb{P}_{\mb{\nu_{k}}}[k> z_i(t) \mid \nu_{k^+}] \big]\leq \frac{2}{\epsilon_i c_i}
\end{eqnarray*}
Consequently,
\[	\eo\Big[\sum_{t\in O(\omega,i)} \Big(\sum^{c_i}_{k=1} \Delta g (k) \mathbb{E}_{\mb{\nu_{k^+}} } \big[(1-\mathcal{X}_k(\nu_{k^+},\epsilon_i,t))\mathbb{P}_{\mb{\nu_k}}[k>z_i(t)\mid \nu_{k^+}] \big]\Big)\Big]\leq \frac{2}{r_i\epsilon_i c_i}\opt_i.\]
The upper bound can be tightened to 0 for $\epsilon_i=1$.
\end{lemma}}
\begin{proof}{Proof.}
{\color{black}	Let us focus on the first part of the lemma. If $\epsilon_i=1$, then for every unit $k$ and sample path $\nukp$, conditioned on $\nukp$, we have,
\[1-\mathcal{X}_k(\nu_{k^+},1, t)=1\quad \Rightarrow \quad  \mathbb{P}_{\mb{\nu_{k}}}[\onee(k,t)=1 \mid \nu_{k^+}]=1\quad \Rightarrow \quad \mathbb{P}_{\mb{\nu_{k}}}[z_i(t)\geq k\mid \nukp]=1.\]

Thus, $(1-\mathcal{X}_k(\nu_{k^+},1,t))\times \mathbb{P}_{\mb{\nu_k}}[k>z_i(t)\mid \nu_{k^+}]=0$, and we are done. 
So, let $\epsilon_i<1$. 
Using $e^x\leq 1+x+x^2$ for $x\in[0,1]$, we have that, 
\[\Delta g(k) \leq (e^{\frac{1}{c_i}}-1)\leq 1/c_i+ 1/c_i^2\leq 2/c_i\text{, for $c_i\geq 1$.}\] Therefore, to prove the first part of the lemma, it suffices to show that for every $\epsilon_i<1$,
\begin{eqnarray}
\sum_{k=1}^{c_i}\mathbb{E}_{\mb{\nu_{k^+}}} \Big[ (1-\mathcal{X}_k(\nu_{k^+},\epsilon_i,t))\, \mathbb{P}_{\mb{\nu_{k}}}[k> z_i(t) \mid \nu_{k^+}] \Big]\leq \frac{1}{\epsilon_i}. \label{geominterim}
\end{eqnarray}
The proof relies on establishing that the LHS is upper bounded by the expectation of a geometric r.v.\ with success probability $\epsilon_i$. }

Fix an arbitrary arrival $t$ and consider the following iterative process for generating sample paths $\nu$.
We sample the usage durations of units one by one in decreasing order of precedence over units. Suppose that as we generate $\nu$, we also grow an array ${\bf b}(\nu)$ with binary values. ${\bf b}(\nu)$ is initially an empty array. 
At the step where usage durations of all units $k^+$ that strictly precede unit $k$ of $i$ have been sampled, 
the covering function $\mathcal{X}_k(\nukp,\epsilon_i,t)$ is well defined. If $t$ is not covered w.r.t.\ $k$ on the sample path $\nukp$ generated so far, then we add an entry to the array ${\bf b}$. Next, we sample the durations of $k$, i.e., generate $\nu_k$. This determines the availability of $k$ at arrival $t$. We set the new array entry to 1 if $k$ is available at $t$ and set the entry to 0 otherwise. Thereafter, we move to the next unit dictated by the precedence order. We repeat the same procedure whenever we reach a unit of resource $i$. For all other units, we simply generate a sample path of its usage durations and move to the next unit in order. At the end of the process (after visiting all units in this manner), we add an entry to the end of the array with value 1 so that there is always an entry with value 1 in the array.

{\color{black}	Let $M(\nu)$ denote the first entry in the array that has value 1. We show that the LHS of \eqref{geominterim} is upper bounded by the expected value of $M(\nu)$. Then, we prove that $\mathbb{E}_{\mb{\nu}}[M(\nu)]\leq \frac{1}{\epsilon_i}$. 
Let $b_k(\nu)$ denote the array entry (if any) generated at unit $k$. If no array entry is generated at $k$ then we set $b_k(\nu)=\emptyset$. If an entry is generated at $k$, then $b_k(\nu)$ is a boolean random variable which takes value 1 on sample paths where $k$ is available at $t$. 
Note that the value of $b_k(\cdot)$ depends only on $\nukp$ and $\nu_k$, and is independent of the usage durations of units that succeed $k$. Let $\onee(b_k(\nu)\neq \emptyset)$ indicate the event that a new entry is added to the array when we generate the usage durations of unit $k$. By definition, 
\[1-\mathcal{X}_k(\nu_{k^+},\epsilon_i,t)=\onee(b_k(\nu)\neq \emptyset).\]
Let $\ell(\nu)$ denote the first unit of $i$ where the corresponding array entry has value 1, i.e., \[b_{\ell(\nu)}(\nu)=1\quad \text{ and }\quad b_k(\nu)\in\{0,\emptyset\}\,\, \forall k\in\{\ell(\nu)+1,\cdots, c_i\}.\] 
If $\ell(\nu)\leq k$ then there exists a unit $k'$ of $i$ that is available at $t$ and that precedes $k$. In this scenario, we have, $k\leq z_i(t)$. Therefore, 
\[\mathbb{P}_{\mb{\nu_{k}}}[k> z_i(t) \mid \nu_{k^+}]\leq \mathbb{P}_{\mb{\nu_{k}}}[\ell(\nu)>k \mid \nu_{k^+}].\]
Using the observations above, we have the following upper bound on the LHS of \eqref{geominterim},
\begin{eqnarray*}	
\sum_{k=1}^{c_i}\mathbb{E}_{\mb{\nu_{k^+}}} \Big[ (1-\mathcal{X}_k(\nu_{k^+},\epsilon_i,t))\, \mathbb{P}_{\mb{\nu_{k}}}[k> z_i(t) \mid \nu_{k^+}] \Big]&\leq& 	\sum_{k=1}^{c_i}\mathbb{E}_{\mb{\nu_{k^+}}}\big[ \onee(b_k(\nu)\neq \emptyset)\,\,\mathbb{P}_{\mb{\nu_{k}}}[\ell(\nu)>k \mid \nukp] \big],\\
&=&\sum_{k=1}^{c_i}\mathbb{P}_{\mb{\nu}}\big[ (b_k(\nu)\neq \emptyset)\,\, \wedge\,\, (\ell(\nu)>k) \big],\\
&\leq &\sum_{k=1}^{c_i}\mathbb{P}_{\mb{\nu}}\big[ (b_k(\nu)\neq \emptyset)\,\, \wedge\,\, (\ell(\nu)\geq k) \big],\\
&=&	\mathbb{E}_{\mb{\nu}}\left[\sum_{k=1}^{c_i} \onee\left(b_k(\nu)\neq \emptyset\,\, \wedge\,\, \ell(\nu)\geq k\right) \right],\\
&=& \mathbb{E}_{\mb{\nu}}[M(\nu)].
\end{eqnarray*}
It remains to show that $\mathbb{E}_{\mb{\nu}}[M(\nu)]\leq \frac{1}{\epsilon_i}$. Since the last array entry is set to 1 we have $M(\nu)\in[c_i+1]$ on every sample path $\nu$. Recall that for any unit $k$ of $i$, conditioned on $\nukp$, 
a new entry is generated only if $t$ is not covered w.r.t.\ unit $k$, i.e., $\mathcal{X}_{k}(\nukp,\epsilon_i,t)=0$. By definition, $\mathbb{P}_{\mb{\nu_{k}}}[\onee(k,t)=1 \mid \nu_{k^+}]\geq \epsilon_i$. When all units of $i$ have been visited, we add a final entry to the array with value 1. Thus, conditioned on the previous array entries, any new entry in the array has value 1 w.p.\ at least $\epsilon_i$ and we get,
\[\mathbb{P}_{\mb{\nu}}[M(\nu)=j\mid M(\nu)> j-1]\,\,\geq\,
\begin{cases}
&\epsilon_i\quad \forall j\in[c_i],\\
&1 \quad \text{ for } j=c_i+1.
\end{cases} \] 
Using this we show that the expected number of entries until we see an entry with value 1 is at most $\frac{1}{\epsilon_i}$. 
\begin{eqnarray*}
\mathbb{E}_{\mb{\nu}}[M(\nu)]&=&\sum_{j= 0}^{c_i} \mathbb{P}_{\mb{\nu}}[M(\nu)> j]\\
&=&1+\sum_{j=1}^{c_i} \mathbb{P}_{\mb{\nu}}[M(\nu)>j-1]\,\, \mathbb{P}_{\mb{\nu}}[M(\nu)>j\mid M(\nu)> j-1]\\
&\leq &1+(1-\epsilon_i) \sum_{j= 1}^{c_i}\mathbb{P}_{\mb{\nu}}[M(\nu)>j-1],\\
&\leq & 1 + (1-\epsilon_i)\,\mathbb{E}_{\mb{\nu}}[M(\nu)]\\
&\leq &\frac{1}{\epsilon_i}.
\end{eqnarray*}

To show the second part of the lemma, we observe that \eqref{geominterim} holds for every arrival $t$. Thus, for any set $S$ of arrivals, we have, 
\[ \sum_{t \in S}\sum_{k=1}^{c_i}\Delta g(k) \mathbb{E}_{\mb{\nu_{k^+}}} \big[ (1-\mathcal{X}_k(\nu_{k^+},t))\,\, \mathbb{P}_{\mb{\nu_{k}}}[k> z_i(t) \mid \nu_{k^+}] \big]\leq \frac{2}{\epsilon_i c_i}|S|. \]
Plugging in $S=O((\omega,i)$ in the inequality above completes the proof.
}		\hfill\Halmos\end{proof}

\begin{repeatlemma}[Lemma \ref{transport}.] For any resource $i$, unit $k$, and path $\nukp$, we have,
\[\mathbb{E}_{\mb{\omega}}\Big[\sum_{t\in O(\omega,i)} \mathcal{X}_k(\nu_{k^+},t) \Big]	\leq c_i r\big(F_i,\mb{s}(\nukp)\big).\]
Consequently,
\[\eo\Big[\sum_{t\in O(\omega,i)} \Big(\sum^{c_i}_{k=1} \Delta g (k) \mathbb{E}_{\mb{\nu_{k^+}} } \big[\mathcal{X}_k(\nu_{k^+},t) \big]\Big)\Big]\leq \Big(1+\frac{2}{c_i}\Big) \sum^{c_i}_{k=1} \Big(g \Big(\frac{k}{c_i}\Big) \,\, \mathbb{E}_{\mb{\nu_{k^+}} } \big[r\big(F_i, \mb{s}(\nukp) \big) \big]\Big).\]		
\end{repeatlemma}
\begin{proof}{Proof.} 
Fix an arbitrary unit $k_O$ of $i$ in \opt\ and let $O(\omega,k_O)$ denote the set of arrivals matched to this unit on sample path $\omega$ in \opt. It suffices to show that,
\begin{equation}\label{single}
\mathbb{E}_{\mb{\omega}}\Big[\sum_{t\in O(\omega,k_O)} \mathcal{X}_k(\nu_{k^+},t) \Big]\leq r\big(F_i,\mb{s}(\nukp)\big).
\end{equation}
Now, recall that the set $\mb{s}(\nukp)$ is defined as the set of all covered arrivals on sample path $\nukp$, i.e., arrivals $t$ for which $\mathcal{X}_k(\nukp,t)=1$. Therefore, 
\[\mathbb{E}_{\mb{\omega}}\Big[\sum_{t\in O(\omega,k_O)} \mathcal{X}_k(\nu_{k^+},t) \Big]=\mathbb{E}_{\mb{\omega}}\Big[\big| O(\omega,k_O)\cap \mb{s}(\nukp)\big| \Big].\]
Now the proof of \eqref{single} follows by a straightforward application of the two list coupling argument used to prove \eqref{interim2} (Lemma \ref{galgua}). To prove the corollary statement, we change the order of summation in the LHS,
\begin{eqnarray*} 
\eo\Big[\sum_{t|O(t)=i} \Big(\sum^{c_i}_{k=1} \Delta g (k) \mathbb{E}_{\mb{\nu_{k^+}} } \big[\mathcal{X}_k(\nu_{k^+},t) \big]\Big)\Big]
=	\eo\Big[ \sum^{c_i}_{k=1} \Delta g (k)\sum_{t|O(t)=i}  \mathbb{E}_{\mb{\nu_{k^+}} } \big[\mathcal{X}_k(\nu_{k^+},t) \big]\Big].
\end{eqnarray*}
Using $\Big(1- \frac{1}{c_i}\Big)\Delta g(k)\leq \frac{1}{c_i}g(\frac{k}{c_i})$ completes the proof.
\hfill\Halmos\end{proof}

\subsection{Proof of Perturbation Property (Proposition \ref{relatexp}) for Various Families}\label{appx:prop1}
\subsubsection{$\{d_i,+\infty\}$ Distributions.}\label{diinf}
We consider the case where the usage duration for resource $i$ takes a finite value $d_i$ w.p.\ $p_i$ and with remaining probability $1-p_i$, the unit is never returned. This generalizes both non-reusable resources as well as deterministic reusability. 
We show inequality \eqref{conjecture} for this family of usage distributions. Then using Lemma \ref{master} we have that for any set of values $\{p_i \}_{i\in I}$, \dpg\ is  
$(1-1/e)$--competitive with convergence rate $O(\frac{1}{\sqrt{c_{min}}})$. 
Focusing on a single resource $i$ with two-point distribution, we show more strongly that regardless of the usage distributions of other resources, condition \eqref{cert2} holds for $i$ with $\alpha_i=1-1/e-O(\frac{1}{\sqrt{c_{i}}})$. We can also sharpen the convergence rate to $O\big(\frac{\log c_{i}}{c_{i}}\big)$, as promised in Table \ref{summary}. This requires a subtler analysis that will take us away from the outline developed previously. Thus, we include a proof with sharper convergence rate of $O\big(\frac{\log c_{i}}{c_{i}}\big)$ separately in Appendix \ref{refined2pnt}. 

Now, to prove \eqref{conjecture}, fix an arbitrary resource $i$, unit $k$ and path $\nukp$. For convenience we treat $\epsilon_i,\kappa_i$ as parameters instead of functions. Their relationship with each other and with parameter $c_i$ will be determined towards the end when we optimize the convergence rate. First, we claim the following.
\begin{lemma}\label{sigmasep}
Any two arrivals $t_1,t_2$ in the set $\mb{\sigma}(\nukp)$ are such that $|a(t_1)-a(t_2)|\geq d_i$.
\end{lemma}  
\begin{proof}{Proof.}
If we are given an ordered set $\mb{\sigma}(\nukp)$ where this is not true, then consider the earliest pair of contiguous arrivals in $\mb{\sigma}(\nukp)$ where this is false and note that the probability of $k$ being matched to the later of the two arrivals is 0. So we can remove this arrival from the set w.l.o.g.\ Repeating this gives a set with the desired property.
\Halmos\end{proof}
The next lemma gives a superset on the set $\mb{s}(\nukp)$. For a positive value $\epsilon_i\in(0,1]$, let $l_0\geq 1$ denote the largest integer such that $p^{l_0-1}_i\geq \epsilon_i$. The value of $\epsilon_i$, which is also the probability lower bound used in defining the covering function, will be suitably chosen later in order to optimize the convergence rate. If $\mb{\sigma}(\nukp)$ contains less than $l_0$ elements, let $\sigma_{l_0}=T$.  
\begin{lemma}
Define, $\mb{S}(\nukp)= \{t\mid \exists \sigma_j \text{ s.t. }  a(t)-a(\sigma_j)\in[0,d_i) \text{ or } t\geq \sigma_{l_0} \}. $ Then $\mb{s}(\nukp)\subseteq\mb{S}(\nukp)$.
\end{lemma}
\begin{proof}{Proof.}
For any arrival $t\not\in \mb{s}(\nukp)$, clearly $t< \sigma_{l_0}$. Now, if $k$ has usage duration $d_i$ for at least the first $l_0-1$ uses, then we claim that $k$ is available at $t$. This would imply that that the probability of $k$ being available at $t$ is at least $p_i^{l_0-1},$ as desired. To see the claim, note that since $k$ can be matched at most $l_0-1$ times prior to any $t<\sigma_{l_0}$, and a finite duration for the first $l_0-1$ uses of $k$ implies that $k$ is in use precisely during the intervals $\cup_{j=1}^{l_0-1}(a(\sigma_j),a(\sigma_j)+d_i)$, which do not contain $t$ by assumption.	Therefore, $k$ is available at $t$ when its first $l_0-1$ durations are finite. 
\hfill\Halmos\end{proof}
\noindent  It remains to show that,
$r\big(F_i, \mb{S}(\nukp)\big) \leq \big(1+\kappa_i\big) r\big(F_i,\mb{\sigma}(\nukp)\big)$. 
\begin{proof}{Proof of \eqref{conjecture}.}
The proof follows by considering the simple coupling where we sample the same number $l_f$ of finite usage durations, for both random processes in question. When the value $l_f\leq l_0-1$, the number of matches on set $\mb{S}(\nukp)$ equals the number of matches on $\mb{\sigma}(\nukp)$. For $l_f>l_0-1$ we have a higher reward on the LHS, however we also have the following,
\[r\big(F_i,\mb{S}(\nukp)\big)\leq r\big(F_i,\mb{\sigma}(\nukp)\big) + p_i^{l_0} r\big(F_i,\mb{S}(\nukp)\big).  \]
Applying Lemma \ref{monotone} then gives us \eqref{conjecture}, with $\kappa_i=\frac{\epsilon_i}{1-\epsilon_i}$.
\hfill\Halmos\end{proof}

\noindent  \textbf{Convergence rate}: The rate $O\big(\frac{1}{\epsilon_ic_i}+\frac{\epsilon_i}{1-\epsilon_i}\big)$ with $\epsilon_i\to0$ for large $c_i$, is optimal for $\epsilon_i=\frac{1}{\sqrt{c_i}}$. A sharper rate of $O\big(\frac{\log c_{i}}{c_{i}}\big)$ is shown in Appendix \ref{refined2pnt}.

\subsubsection{IFR Distributions.}\label{sec:ifr}

A major advantage in the $\{d_i,+\infty\}$ case was that if a unit returned after usage, the duration of usage was always $d_i$. In other words, we had the additional structure that a returning unit of $i$ in \dpg\ and in \opt\ was used for the exact same duration and the main question was whether the unit returned at all. More generally, it is not simply a matter of an item returning after use but also the duration of usage. In particular, the probability that an item is available for (potential) $l$-th use is not stated as simply as $p^{l-1}_i$. In this section, we address the new issues that arise for 
continuous IFR distributions. 

For convenience we treat $\epsilon_i,\kappa_i$ as parameters instead of functions. Their relationship with each other and with parameter $c_i$ will be determined towards the end when we optimize the convergence rate. 
We use $f_i$ to refer to the p.d.f.\ and $F_i$ to refer to the c.d.f..\ Recall that the function $L_i(\epsilon)$ is defined as the maximum possible value of the ratio $\big(F_i(x+F_i^{-1}(\epsilon))-F_i(x)\big)/\epsilon$. 
We will show that when the usage distribution of $i$ is IFR and bounded in the following sense, 
\[L_i(\epsilon_i)\epsilon_i\to 0 \text{ as }\epsilon_i\to 0,\]  then \dpg\ is asymptotically $(1-1/e)$--competitive. More specifically,  inequality \eqref{cert2} is satisfied with, 
\[\alpha_i=(1-1/e)-O\bigg(\frac{1}{\epsilon_i c_i}+L_i(\epsilon_i)\epsilon_i\log \Big(\frac{1}{L_i(\epsilon_i)\epsilon_i}\Big)\bigg).\] 
For distributions where $L_i(\epsilon_i)\epsilon_i=O(\epsilon_i^{\eta})$ for some $\eta>0$, the optimal rate is thus, $\tilde{O}\Big( c^{-\frac{\eta}{1+\eta}}_{i}\Big)$. In Appendix \ref{ifrexamples}, we evaluate this parameter $\eta$ for some commonly known families of IFR distributions.
Observe that for IFR distributions that have non-increasing densities, such as exponential, uniform etc., $L_i(\epsilon_i)=1$ and we have a resulting convergence rate of $O(\log c_i/\sqrt{c_i})$. In fact, for exponential distributions we will show a stronger convergence rate of $O(1/\sqrt{c_i})$. 

To prove these claims, let us fix an arbitrary resource $i$, unit $k$, path $\nukp$. 
This also fixes the set $\mb{\sigma}(\nukp)$ of arrivals that $k$ could be matched to in \dpg. Since we have fixed $i$, for simplicity we use the abbreviations,
\[\epsilon := \epsilon_i, \quad L := L_i(\epsilon_i), \quad F:=F_i  \]
We assume that the density $f_i$ is continuous and consequently, let $\delta_0$ and $\delta_L$ be such that, 
\[F(\delta_0)=\epsilon \quad \text{ and }\quad F(\delta_L)= L\epsilon, \]
where note that, $ L\epsilon\leq 1$ by definition of $L$. Now, we claim that for a $\mathcal{X}_k(\nukp,\epsilon,t)$ covering, the following set of arrivals is a superset of $\mb{s}(\nukp)$,
\[\mb{S}(\nukp)=\{t \mid  a(t)\in[a(\sigma_j), a(\sigma_j)+\delta_0) \text{ for some } \sigma_j\in\mb{\sigma}(\nukp) \}. \]
\begin{lemma}$\mb{s}(\nukp)\subseteq \mb{S}(\nukp).$
\end{lemma}
\begin{proof}{Proof.}
Consider an arrival $t\notin \mb{S}(\nukp)$. The closest arrival preceding $t$ in $\mb{\sigma}(\nukp)$, is at least $\delta_0$ time before $a(t)$. Using the IFR property, we have that the probability that $k$ switches from being in-use to free between $a(\sigma_j)$ and $a(\sigma_j)+\delta_0$ is at least $F(\delta_0)=\epsilon$.
\hfill\Halmos\end{proof}
\noindent Given this, in order to show \eqref{conjecture}, we now aim to to show the following lemma and then apply the monotonicity Lemma \ref{monotone}.
\begin{lemma}\label{ifrkey}
Given an IFR distribution $F$ and two sets or arrivals $\mb{S}(\nukp)$ and $\mb{\sigma}(\nukp)$, with the property that for any arrival $t$ in $\mb{S}(\nukp)$, there exists an arrival in $\mb{\sigma}(\nukp)$ that precedes $t$ by at most $\delta_0$ time. We have,	\[r\big(F,\mb{S}(\nukp) \big)\leq \Bigg(1+ O\bigg(F(\delta_L)\log\Big(\frac{1}{F(\delta_L)}\Big)\bigg)\Bigg)\, r\big(F,\mb{\sigma}(\nukp)\big).\]
\end{lemma}
\begin{proof}{Proof.}
We will show this in two main steps. First, we show for a certain modified distribution $F^m$ ($F^m_i$ to be precise),
\begin{eqnarray}
r\big(F,\mb{S}(\nukp) \big)\leq (1+2F(\delta_L)) r\big(F^m,\mb{\sigma}(\nukp)\big). \label{moddist}
\end{eqnarray}
In the second step, we will relate $r(F^m,A)$ to $r(F,A)$ for any set $A$. The modified distribution $F^m$ needs to satisfy a property that we describe next. The exact definition of the distribution will be introduced later in the proof.
\begin{proposition}\label{prop:moddist}
Given a c.d.f.\ $F$, we require a modified distribution $F^m$ such that, there exists a coupling that generates samples distributed according to $F^m$ by using i.i.d.\ samples from $F$, with the property that for every sample of $F$ that has value $d\geq \delta_L$, the coupled sample of $F^m$ has value at most $d-\delta_0$.
\end{proposition}
Suppose a distribution $F^m$ that satisfies the above property exists and let $\pi:[\delta_L,+\infty) \to \mathbb{R}^+$ be the mapping induced by the coupling, from samples of $F$ to those of $F^m$. Then, to show \eqref{moddist}, we use the following coupling between $(F,\mathcal{S}(\nukp))$ and $(F^m,\mb{\sigma}(\nukp))$ random processes. For the rest of the proof, let us refer to these as processes as $R$ and $R^m$ resp..\

Consider a long list of i.i.d.\ samples from distribution $F$. 
Starting with a pointer $P$ at the first sample in the list, every time we have a transition to in-use in process $R$, we draw the sample $P$ is pointing to and move $P$ to the next sample on the list. 
Thus, pointer $P$ draws from the list in order without skipping over samples.
To couple the two processes, we introduce another pointer on the list, $P^m$, for process $R^m$. 
Let us call any sample with value at least $\delta_L$, a \emph{large} sample. Other samples are called \emph{small} samples. Pointer $P^m$ starts at the first large sample in the list. Each time $R^m$ needs a new sample, 
we draw the sample pointed to by $P^m$ and pass it through the function $\pi(\cdot)$. The output is then passed to process $R^m$ and $P^m$ moves down the list to the next large sample. If $P$ and $P^m$ point to the sample sample and $P$ moves down the list, $P^m$ also moves down the list to the next large sample. At every arrival, we first provide a sample, if required, to $R^m$ and update $P^m$ before moving to $R$. Thus, pointer $P^m$ never lags behind $P$. Compared to the coupling used for Lemma \ref{transport}, there are two important differences. First, here we do rejection sampling by ignoring small values, as $P^m$ always skips to the next large sample. Second, we pass the value pointed to by $P^m$ through the function $\pi(\cdot)$ before giving the sample to $R^m$. 

We claim that $P^m$ never skips a large sample 
in the list. We shall prove this by contradiction. Consider the first large sample in the list which is skipped by $P^m$, let its position in the list be $q$. Clearly, $q$ is a sample that is used by $R$ otherwise $P^m$ would never skip it. Further, $q$ cannot be the first large sample in the list, as the earliest arrival in $\mb{\sigma}(\nukp)$ is also the first arrival in $\mb{S}(\nukp)$. Therefore, $R^m$ uses the first large sample at least as early as $R$ does. More generally, given that $P^m$ never lags $P$, we have that all large samples preceding $q$ were first used in $R^m$. Consider the arrival $\sigma_j\in\mb{\sigma}(\nukp)$, where $R^m$ uses the large sample immediately preceding $q$ in the list. Due to the fact that function $\pi(\cdot)$ reduces the sample value passed to $R^m$ by at least $\delta_0$, we claim that $R^m$ needs the next large sample, i.e., $q$, before $R$ does. To see the claim, let $s_j$ be the arrival in $\mb{S}(\nukp)$ where $R$ uses the large sample immediately preceding $q$. Observe that $\sigma_j$ occurs at least as early as $s_j$ due to $P^m$ leading $P$ and by the assumption on $q$ being the first sample skipped by $P^m$. Now, due to the decrease in sample value resulting from applying $\pi$, we have that in $R^m$, $k$ returns to free state after becoming matched to $\sigma_j$, at least $\delta_0$ time before $k$ returns to free state after getting matched to $s_j$ in $R$. Letting $t_{\sigma}$ and $t_s$ denote these return times resp., we have, $t_{\sigma}\leq t_s-\delta_0$. Suppose that the first arrival in $\mb{S}(\nukp)$ after time $t_s$ occurs at time $\tau_s$ and the first arrival in $\mb{\sigma}(\nukp)$ after time $t_{\sigma}$ occurs at time $\tau_{\sigma}$. To establish the claim that $R^m$ uses sample pointed by $q$ before $R$ does, it suffices to argue that $\tau_{\sigma}\leq \tau_s$. From the definition of set $\mb{S}(\nukp)$, we have that either $\tau_s \in \mb{\sigma}(\nukp)$ (in which case, we are done) or $\tau_s$ is preceded by an arrival in $\mb{\sigma}(\nukp)$ that occurs in the interval $[\tau_s-\delta_0,\tau_s)$. Therefore, 
$t_{\sigma}\leq t_s-\delta_0\leq \tau_s-\delta_0,$ and we have that $\tau_{\sigma}\leq \tau_s$. 
So $R^m$ requires sample $q$ before $R$, 
which contradicts that $q$ is the first large sample skipped by $P^m$. Hence, $P^m$ does not skip any large samples. 
Since $P^m$ never lags $P$, we also have that the total number of large samples passed to $R$ is upper bounded by the number of samples passed to $R^m$. Let $r^{lg}(F,\mb{S}(\nukp))$ denote the expected number of large sample transitions from available to in-use in process $R$. So far we have shown that,
\[r^{lg}(F,\mb{S}(\nukp))\leq r(F^m,\mb{\sigma}(\nukp)). \]

\noindent Next, we claim that, 
\[r(F,\mb{S}(\nukp))\leq (1+2F(\delta_L)))r^{lg}(F,\mb{S}(\nukp)).\]
Consider an arbitrary transition from available to in use in $R$. The probability that duration of this transition is large is $1-F(\delta_L)$. Thus, the expected contribution from this transition to $r^{lg}(F,\mb{S}(\nukp))$ is $1-F(\delta_L)$. Summing over all transitions, we have the desired. This completes the proof of \eqref{moddist}. To finish the main proof it remains to define a modified distribution $F^m$ such that Proposition \ref{prop:moddist} holds and compare the expected rewards of random processes
$(F^m,\mb{\sigma}(\nukp))$ with  $(F,\mb{\sigma}(\nukp))$. This is the focus of the next lemma.
\hfill\Halmos\end{proof}
\begin{lemma}
Given IFR distribution $F$ and value $\epsilon\in(0,1]$ such that $\epsilon L(\epsilon)\leq 1/2$, there exists a modified distribution $F^m$ that satisfies Proposition \ref{prop:moddist}, such that for any given set of arrivals $A$,
\begin{equation}
r(F^m,A)\leq \Bigg(1+ O\bigg(F(\delta_L)\log\Big(\frac{1}{F(\delta_L)}\Big)\bigg)\Bigg) r(F,A).\label{modtonormal}
\end{equation}
\end{lemma}
\begin{proof}{Proof.}
Let us start with the case of exponential distribution as a warm up. We claim that in fact, choosing $F^m=F$ suffices in this case. Inequality \eqref{modtonormal} follows directly for this choice and it remains show that Proposition \ref{prop:moddist} is satisfied, i.e., prove the existence of a coupling/mapping $\pi$. Consider the following modified density function,
\[f^m(x)= \frac{f(x+\delta_0)}{1-F(\delta_0)}.\]
$f^m=f$ owing to the memoryless property of exponentials. Now, consider the straightforward coupling that samples from $F$ and rejects all samples until the first large sample is obtained, which is reduced by $\delta_0$ before it is output. Observe that this process generates samples with distribution $F^m$, and 
satisfies Proposition \ref{prop:moddist}.

For other IFR distributions, it is not clear if $r(F^m,\mb{\sigma}(\nukp))$ and $r(F,\mb{\sigma}(\nukp))$ are comparable given the current definition of $F^m$. So we introduce a new modified distribution. Define $\delta_1$ so that,
\[F(\delta_1)= 1-L\epsilon.\] 
Since we started with $\epsilon L\leq 1/2$ for this lemma, we have that $\delta_L\leq \delta_1$. Now, the new distribution is,
\begin{equation*}\label{ifrmodf}
f^m(x)=
\begin{cases*}
\frac{1}{1-L\epsilon} f(x) & $x\in [0,\delta_{1}]$ \\
0       & $x > \delta_{1}$
\end{cases*}
\end{equation*}
So $f^m$ is a truncated version of $f$. First, we show that Proposition \ref{prop:moddist} holds by defining a mapping $\pi$ such that, for every $t\geq \delta_L$ we have $F(t)-F(\pi(t))=L\epsilon$. By definition of $L$ and $\epsilon$, this would imply $t-\pi(t)\geq \delta_0$. 

Consider values $x\in [0,L\epsilon)$ and non-negative $j\in \mathbb{Z}$. Let $t(x,j)=F^{-1}(x+(j+1)L\epsilon)$. Varying $x$ and $j$ we have that $t(x,j)$ takes all possible values in the range $[\delta_L,+\infty)$. Defining, 
$\pi(t(x))=F^{-1}(x+jL\epsilon)$, we have that $F(t)-F(\pi(t))=L\epsilon$ for every $t\geq \delta_L$, as desired. Finally, the distribution generated by sampling i.i.d.\ values from $F$ and applying mapping $\pi$ to all large samples (while ignoring all small samples), corresponds to the distribution $F^m$ as defined.
To show \eqref{modtonormal}, consider the random variable defined as the minimum number of values drawn i.i.d.\ with density $f^m$, such that the sum of the values is at least as large as a single random value that is drawn independently with	 density $f$. Let $\hat{n}$ denote the expected value of this random variable. 
Observe that in order to show the main claim it suffices to show that $\hat{n}\leq (1+\gamma)$, where $\gamma = O(L\epsilon \log (1/L\epsilon))$. To show this, let $\mu, \mu^m$ denote the 
mean of $f$ and $f^m$ respectively. Recall that $f^m$ is a truncated version of $f$. By definition of $f^m$ and the IFR property of $f$, we have that,
\begin{equation}
\mu\leq (1-L\epsilon)\cdot \mu^m + L\epsilon (\delta_1+\mu).\label{mupb}
\end{equation}
For the time being, we claim $\mu$ is lower bounded as follows,
\begin{eqnarray}
\mu \geq O(1)\cdot \frac{\delta_1}{\log (1/L\epsilon)}.\label{meanlb}
\end{eqnarray}
We proceed with the proof assuming \eqref{meanlb} holds and prove this claim later. Substituting this inequality in \eqref{mupb} and using the fact that $x\log (1/x)<1$ for $x\in(0,1]$, we get, 
\begin{eqnarray}
\mu^m\geq O(1)\cdot \frac{\delta_{1}}{\log (1/L\epsilon)} 
.\label{meanlb2} 
\end{eqnarray}  

\noindent Now, to upper bound $\hat{n}$ consider the following coupling: Given a random sample $t$ from $f$, for $t\leq \delta_1$ we consider an equivalent sample drawn from $f^m$. For durations $t>\delta_1$, we draw independent samples from $f^m$ until the sum of these samples is at least $t$. Thus, w.p.\ $(1-L\epsilon)$, we draw exactly one sample from $f^m$ and it suffices to show that in the remaining case that occurs w.p.\ $L\epsilon$, we draw in expectation $O(\log (1/L\epsilon))$ samples from $f^m$. Now, given sample $t\geq \delta_{1}$, note that from \eqref{meanlb2}, the expected number of samples drawn from $f^m$ before their sum exceeds $\delta_1$ is, $O(\log (1/L\epsilon))$. 
Second, 
using the IFR property of $f$, we have that the expected number of samples of $f^m$ such that the sum of the sample values is at least $t-\delta_{1}$, is upper bounded by $\hat{n}$. Overall, we have the recursive inequality, \[\hat{n}\leq (1-L\epsilon) + L\epsilon \big(O(\log (1/L\epsilon)) +\hat{n}\big).\] 
The desired upper bound on $\hat{n}$ now follows. 

It only remains to show \eqref{meanlb}. Let $x_1$ be such that $1-F(x_1)=1/2$. Then, let $x_2$ be such that $1-F(x_1+x_2)=1/4$. More generally, let $\{x_1,\cdots,x_s\}$ be the set of values such that $1-F(\sum_{j=1}^s x_j)=1/2^s$. We let $s$ be the smallest integer greater that equal to $\log(1/L\epsilon)$. Now, the mean $\mu\geq \sum_{j=1}^s x_j/2^j$.  Therefore, $\mu/\delta_1$ is lower bounded by $\frac{\sum_{j=1}^s x_j/2^j}{\sum_{j=1}^s x_j}$. While in general this ratio can be quite low, due to the IFR property we have $x_j\geq x_{j+1}$ for every $j\geq 1$. Consequently, the minimum value of this ratio is $\frac{O(1)}{s}=\frac{O(1)}{\log (1/L\epsilon)}$, and occurs when all values $x_j$ are equal (which incidentally, implies memoryless property).  
\hfill\Halmos\end{proof}
\textbf{Convergence rate}: We have $\kappa=  O\bigg(F(\delta_L)\log\Big(\frac{1}{F(\delta_L)}\Big)\bigg)=O\big(L\epsilon\log (1/L\epsilon)\big).$ Thus, over all rate of convergence for resource $i$ is, $O\bigg(\frac{1}{\epsilon_i c_i}+L_i(\epsilon_i)\epsilon_i\log \Big(\frac{1}{L_i(\epsilon_i)\epsilon_i}\Big)\bigg)$. Suppose $L_i(\epsilon_i)\epsilon_i=O(\epsilon_i^{\eta})$ for some $\eta>0$, then the optimal rate is, $\tilde{O}\Big( c^{-\frac{\eta}{1+\eta}}_{i}\Big)$, where the $\tilde{O}$ hides a $\log c_{i}$ factor.

\subsubsection{IFR with Arbitrary Mass at $+\infty$.}\label{massatinf}
In the previous section we, showed that for bounded IFR distributions \dpg\ is asymptotically $(1-1/e)$--competitive. In certain practical scenarios, it may be reasonable to expect that with some probability resource units under use may never return back to the system. 
We model this by allowing an arbitrary mass at infinity. 
Specifically, for resource $i$, let $p_i$ denote the probability that a usage duration takes value drawn from distribution with c.d.f.\ $F_i$ and with the remaining probability $1-p_i$, the duration takes value $+\infty$. In this section, we show that \dpg\ is still asymptotically $(1-1/e)$--competitive for such a mixture of non-increasing IFR 
usage distributions with arbitrary mass at infinity. 

Fix, $i,k,\nukp$ and define $L$,$\epsilon$, $\delta_0$ and $\delta_L$ as before. 
Given set $\mb{\sigma}(\nukp)$, consider the random process $(F,\mb{\sigma}(\nukp))$ and let $\onee(t,{\tt finite})$ indicate the event that $k$ has not hit a $+\infty$ duration at arrival $t$. For a chosen value $\gamma\in(0,1]$ (finalized later to optimize convergence rate), let $\sigma_{l_f}$ be the last arrival in $\mb{\sigma}(\nukp)$ where,
\[\mathbb{E}[\onee(\sigma_{l_f},{\tt finite})]\geq \gamma. \]
If no such arrival exists, let $\sigma_{l_f}=T$. This implies,
\begin{eqnarray}
&&\mathbb{E}[\onee(\sigma_j,{\tt finite})]\geq \gamma\quad \text{for every } \sigma_j<\sigma_{l_f} \text{ in } \mb{\sigma}(\nukp), \label{finite}\\
\text{ and } && \mathbb{E}[\onee(\sigma_j,{\tt finite})]< \gamma\quad \text{for every } t>\sigma_{l_f} \label{infinite}.
\end{eqnarray}
Using this define the set of arrivals,
\[\mb{S}(\nukp)=\{t\mid t>\sigma_{l_f} \text{ or } \exists \sigma_j\in\mb{\sigma}(\nukp) \text{ s.t. } a(t)\in [a(\sigma_j),a(\sigma_j)+\delta_0)  \}. \]
\begin{lemma}
For a $\mathcal{X}_k(\nukp,\gamma\epsilon,t)$ covering, we have that the set of covered arrivals satisfies the following relation,	$\mb{s}(\nukp)\subseteq \mb{S}(\nukp) .$
\end{lemma}
\begin{proof}{Proof.}
Consider an arrival $t\notin \mb{S}(\nukp)$. Clearly, $t<\sigma_{l_f}$ and the closest arrival preceding $t$ in $\mb{\sigma}(\nukp)$, call it $\sigma_t$, is at least $\delta_0$ time before $a(t)$. From \eqref{finite}, we have that $k$ has not hit a duration of $\infty$ by the time $\sigma_t$ arrives, w.p.\ at least $\gamma$. Conditioned on this, we have from the IFR property that $k$ switches from being in-use to free between $a(\sigma_j)$ and $a(\sigma_j)+\delta_0$ w.p.\ at least $F(\delta_0)=\epsilon$. Together, this implies that the probability that $k$ is free at $t$ is at least $\gamma\epsilon.$
\hfill\Halmos\end{proof}
\begin{lemma} 
$r\big(F,\mb{S}(\nukp) \big)\leq \Big(1+ O\big(L\epsilon\log (1/L\epsilon)\big)\Big)(1+2\gamma)\, r\big(F,\mb{\sigma}(\nukp)\big).$
\end{lemma}
\begin{proof}{Proof.}
The proof requires a combination of the two analyses so far. Consider the truncated set $\mb{\hat{S}}(\nukp)=\mb{S}(\nukp)\backslash \{t \mid t>\sigma_{l_f}\}.$ First, we have that,
\begin{eqnarray*}
r\big(F,\mb{S}(\nukp)\big)&&< r\big(F,\mb{\hat{S}}(\nukp)\big) + \gamma r\big(F,\mb{S}(\nukp)\big),\\
r\big(F,\mb{S}(\nukp)\big)&&\leq (1+2\gamma) r\big(F,\mb{\hat{S}}(\nukp)\big),
\end{eqnarray*}
where we used the fact that $\gamma\leq 1$.
It suffices to therefore show that, \[r\big(F,\mb{\hat{S}}(\nukp)\big)\leq  \Big(1+ O\big(L\epsilon\log (1/L\epsilon)\big)\Big) r\big(F,\mb{\sigma}(\nukp)\big).\]
Let us condition on the first $l$ transition in each process being finite and the $l+1$-th transition being $\infty$. Then, the resulting expected number of transitions can be compared in the same way as the case of IFR distributions given in Lemma \ref{ifrkey}, since we now draw independently from an IFR distribution for the first $l$ durations (having conditioned on these durations being finite). This holds for arbitrary $l$. Taking expectation over $l$ then gives the desired. 
\hfill\Halmos\end{proof}
\textbf{Convergence rate}: We have $\kappa= O\big(\gamma+L\epsilon\log (1/L\epsilon)\big)$. Moreover, the probability lower bound at uncovered arrivals is actually $\gamma\epsilon$, instead of $\epsilon$. Thus, over all convergence rate for resource $i$ is, $O\bigg(\frac{1}{\epsilon_i\gamma_i c_i}+\gamma_i+L_i(\epsilon_i)\epsilon_i\log \Big(\frac{1}{L_i(\epsilon_i)\epsilon_i}\Big)\bigg)$. Suppose $L_i(\epsilon_i)\epsilon_i=O(\epsilon_i^{\eta})$, then the optimal rate is, $\tilde{O}\Big( c^{-\frac{\eta}{1+2\eta}}_{i}\Big)$.

\subsection{Refined Bound for $\{d_i,+\infty\}$}\label{refined2pnt}
The main bottleneck in getting a stronger convergence factor is the need to ensure  large enough probability lower bound $\epsilon_i$, for uncovered arrivals. In particular, we chose the parameter $l_0$ to ensure that $p_i^{l_0-1}\geq 1/\sqrt{c_i}$, and consequently bound the error term $\frac{1}{p_i^{l_0-1}c_i}$ arising out of Lemma \ref{geom}. Choosing a larger $l_0$ 
would worsen the convergence rate since term designates the contribution from uncovered arrivals in \opt\ to \eqref{lam4c}. The separation of contributions from uncovered and covered arrivals makes the analysis tractable in general. It can also be, at least for $\{d_i,+\infty\}$ distributions, pessimistic from the point of view of convergence to the guarantee. As an extreme but illustrative example, consider a sample path in \opt\ where a unit $k_O$ of $i$ is matched only to arrivals $t$ that occur late and are all uncovered given some unit $k$ and path $\nukp$ in \dpg. Specifically, at each of the arrivals $k_O$ is matched to, let the probability of $k$ in \dpg\  being available (conditioned on $\nukp$), equal $p_i^{l_0-1}$. In this case, the contributions to \eqref{lam4c} from uncovered arrivals is significant and that part of the analysis is tight, but there is no contribution from covered arrivals in \opt. Thus, the $\theta_i$ term could be used to neutralize some of the negative terms arising out of the uncovered arrivals and this would in turn allow us to set a larger value of $l_0$. This observation is the key idea behind improving the analysis.

\begin{lemma}For every $i$, \eqref{cert2} is satisfied with $\alpha_i=(1-1/e)-O\Big(\frac{\log c_i}{c_i}\Big)$.
\end{lemma}
\begin{proof}{Proof.}
We start by fixing also an arbitrary unit $k_O$ in \opt\ and show the following inequality,
\begin{eqnarray} 
&&\frac{(1+2/c_i)}{r_ic_i}\theta_i + \eo\Big[ \sum_{t\mid O(t)=(i,k_O)} \Big((1-1/e)- \sum^{c_i}_{k=1} \Delta g (k) \pd \big[k> z_i(t)\big]\Big)\Big]\nonumber \\
&& \geq \Big[1-1/e -O\Big(\frac{\log c_i}{c_i}\Big)\Big] \eo\Big[ \big| \{t\mid O(t)=(i,k_O)\}\big|\Big].\label{unitopt}
\end{eqnarray}
The proof then follows by linearity of expectation.

To show \eqref{unitopt}, we start by conditioning on all randomness in \opt\ arising out of usage durations of resources and units other than unit $k_O$ of $i$, as well as any intrinsic randomness in \opt. Let this partial sample path be denoted as $\omega_{-k_O}$. This fixes the set of arrivals that $k_O$ is matched to if durations of $k_O$ are all finite. Let this set be $\Gamma (\omega_{-k_O})=\{t_1(\omega_{-k_O}),t_2(\omega_{-k_O}),\cdots,t_{l_f(\omega_{-k_O})}(\omega_{-k_O})\}$. 
Since usage durations of $k_O$ are sampled independently, we have that,
\[ \eo\Big[ \big| \{t\mid O(t)=(i,k_O)\}\big|\Big]= \mathbb{E}_{\mb{\omega_{-k_O}}}\Big[ \sum_{l\geq 1}p_i^{l-1} \onee\big(l\leq l_f(\omega_{-k_O})\big)  \Big] \]
Now, define $l_0$ as the largest integer such that $p_i^{l_0-1}\geq \frac{\log c_i}{c_i}$. On every path $\omega_{-k_O}$, it suffices to focus on at most $l_0-1$ finite durations for $k_O$ as,\
\[\sum_{l\geq 1}^{l_0} p_i^{l-1} \onee\big(l\leq l_f(\omega_{-k_O})\big)\geq \Big(1-\frac{\log c_i}{c_i}\Big) \sum_{l\geq 1}p_i^{l-1} \onee\big(l\leq l_f(\omega_{-k_O})\big). \]
Thus, in order to prove \eqref{unitopt} it suffices to show that,
\begin{eqnarray} 
&&\frac{(1+2/c_i)}{r_ic_i}\theta_i + \mathbb{E}_{\mb{\omega_{-k_O}}} \Big[ \sum_{l\geq 1}^{l_0} p_i^{l-1}\onee\big(l\leq l_f(\omega_{-k_O})\big) \Big((1-1/e)- \sum^{c_i}_{k=1} \Delta g (k) \pd \big[k> z_i(t_{l}(\omega_{-k_O}))\big]\Big)\Big]\nonumber\\
&& \geq \Big[1-1/e -O\Big(\frac{\log c_i}{c_i}\Big)\Big] \mathbb{E}_{\mb{\omega_{-k_O}}} \Big[ \sum_{l\geq 1}^{l_0} p_i^{l-1} \onee\big(l\leq l_f(\omega_{-k_O})\big)\Big].\nonumber
\end{eqnarray}
In fact, it suffices to show more strongly that for every ordered collection $\Gamma=\{t_1,\cdots,t_{l_f}\}$ of arrivals such that any two consecutive arrivals in the set are at least $d_i$ time apart and $l_f\leq l_0$, we have,
\begin{eqnarray} 
&&  \sum_{l\geq 1}^{l_f} p_i^{l-1} \Big(\sum^{c_i}_{k=1} \Delta g (k) \pd \big[k> z_i(t_{l})\big]\Big)- \frac{(1+2/c_i)}{r_ic_i}\theta_i\nonumber\\
&& = O\Big(\frac{\log c_i}{c_i}\Big)  \sum_{l\geq 1}^{l_f} p_i^{l-1}=O\Big(\frac{\log c_i}{c_i}\Big)  \frac{1-p_i^{l_f}}{1-p_i}.\label{final}
\end{eqnarray}
%
\noindent Now, fix unit $k$, sample path $\nukp$, and consequently, the set $\mb{\sigma}(\nukp)$. We classify arrivals $t_l\in \Gamma$ into covered and uncovered in a new way, given by a function $\mathcal{Y}_k$. Using Lemma \ref{sigmasep}, we let any two arrivals in $\mb{\sigma}(\nukp)$ be at least $d_i$ time apart. Then,

\emph{$t_l$ is covered and $\mathcal{Y}_k(\nukp,t_l)=1$ iff there are $l$ or more arrivals preceding $t_l$ in $\mb{\sigma}(\nukp)$. }

\noindent Using this definition and the decomposition we performed in Lemma \ref{decompose2}, we have,
\begin{eqnarray}
&&\sum_{l\geq 1}^{l_f} p_i^{l-1} \Big(\sum^{c_i}_{k=1} \Delta g (k) \pd \big[k> z_i(t_{l})\big]\Big)\leq \nonumber\\
&&\sum_{l\geq 1}^{l_f} p_i^{l-1} \Big(\sum^{c_i}_{k=1} \Delta g (k)  \mathbb{E}_{\mb{\nu_{k^+}} } \big[\mathcal{Y}_k(\nu_{k^+},t_{l}) \big]\Big)+\label{summand1}\\
&&\sum_{l\geq 1}^{l_f} p_i^{l-1} \Big(\sum^{c_i}_{k=1} \Delta g (k) \mathbb{E}_{\mb{\nu_{k^+}} } \big[(1-\mathcal{Y}_k(\nu_{k^+},t_{l}))\mathbb{P}_{\mb{\nu_k}}[k>z_i(t_{l})\mid \nu_{k^+}] \big]\Big)\nonumber.
\end{eqnarray}
Now, for any unit $k$ in \dpg, in \eqref{summand1} we may interpret probabilities $p_i^{l-1}$ as the probability that at least $l$ durations of $k$ are finite. Then, by definition of the coupling we have that for every path $\nukp$,
\[\sum_{l\geq 1}^{l_f} p_i^{l-1}  \mathcal{Y}_k(\nu_{k^+},t_{l})  \leq r(F_i,\mb{\sigma}(\nukp)).  \]
Using the same algebra as the corollary statements in Lemma \ref{transport} and Proposition \ref{prop:moddist} then gives us that,
\[\sum_{l\geq 1}^{l_f} p_i^{l-1} \Big(\sum^{c_i}_{k=1} \Delta g (k)  \mathbb{E}_{\mb{\nu_{k^+}} } \big[\mathcal{Y}_k(\nu_{k^+},t_{l}) \big]\Big)\leq  \frac{(1+2/c_i)}{r_ic_i}\theta_i. \]
So in order to prove \eqref{final}, it remains to show that,
\[\sum_{l\geq 1}^{l_f} p_i^{l-1} \Big(\sum^{c_i}_{k=1} \Delta g (k) \mathbb{E}_{\mb{\nu_{k^+}} } \big[(1-\mathcal{Y}_k(\nu_{k^+},t_{l}))\mathbb{P}_{\mb{\nu_k}}[k>z_i(t_{l})\mid \nu_{k^+}] \big]\Big)=O\Big(\frac{\log c_i}{c_i}\Big)  \frac{1-p_i^{l_f}}{1-p_i}. \]
To establish this, again fix a unit $k$ and path $\nukp$ and note that if some arrival $t_{l}$ is uncovered, we have at most $l-1$ arrivals in $\mb{\sigma}(\nukp)$ preceding $t_{l}$. Therefore, at any uncovered arrival $t_{l}$,
\[\mathbb{P}_{\mb{\nu_k}}[\onee(k,t_{l})=1\mid \nu_{k^+}]\geq p_i^{l-1}.\]
Then, applying Lemma \ref{geom} with $\epsilon_i=p_i^{l-1}$, we have that for any $p_i<1$,
\[	\sum_{k=1}^{c_i}\Delta g(k) \mathbb{E}_{\mb{\nu_{k^+}}} \big[ (1-\mathcal{Y}_k(\nu_{k^+},t_{l}))\cdot \mathbb{P}_{\mb{\nu_{k}}}[k> z_i(t_{l}) \mid \nu_{k^+}] \big]\leq \frac{2}{p_i^{l-1}c_i}, \]
and in case $p_i=1$, the RHS equals 0 (which gives us the desired). For $p_i<1$, observe that,
\[\sum_{l\geq 1}^{l_f} p_i^{l-1}\cdot \frac{2}{p_i^{l-1}c_i}=\frac{2l_f}{c_i}\quad \text{ and }\quad l_f \leq O(\log c_i)\frac{1-p_i^{l_f}}{1-p_i}\quad \forall l_f\leq l_0, \]
which completes the proof.
\hfill\Halmos\end{proof}
\subsection{Performance of \dpg\ Beyond Online Matching}\label{sec:rbafail}

Recall that to analyze \galg, we captured the allocation of each individual unit 
via a natural $(F,\bsigma,\mb{p})$ random process. In an $(F,\bsigma)$ random process, at every $\sigma_t\in \bsigma$ the only event of relevance is the availability of the item at $\sigma_t$ (as probability $p_t=1$). If the item is available, then it is matched to $\sigma_t$ regardless of the usage durations realized before $\sigma_t$. 
In order to naturally capture the actions of \dpg\ on individual units through a random process, a similar property must hold in \dpg. 
To be more specific, consider a single unit $k$ of $i$ in \dpg\ and fix the randomness associated with all units and resources except $k$. Suppose we have an arrival $t$ with edge to $i$ and two distinct sample paths over usage durations of $k$ prior to the arrival of $t$. We are also given that $k$ is available at $t$ on both sample paths. Clearly, \dpg\ makes a deterministic decision on each sample path. Is it possible that \dpg\ matches $t$ to $k$ on one sample path but not the other? 

Interestingly, in the case of matching we can show that this is impossible, i.e., \dpg\ always makes the same decision on both sample paths in such instances. We show that this is not the case in general. In particular, when we include the aspect of customer choice or budgeted allocations (even one of these), reusability can interact with these elements in an undesirable way.  

For the setting of online budgeted allocations (i.e., online matching with multi-unit demand), we generalize \dpg\ as follows: given arrival $t$ with demand $b_{it}\geq 0$ for resource $i\in I$ and set $S_t$ denoting the set of available resources with an edge to $t$, 
we match $t$ to,
\begin{equation}
\argmax_{i\in S_t}\, \sum_{k=1}^{b_{it}} \, r_i\left(1-g\left(\frac{z_i(t,k)}{c_i}\right)\right),\nonumber 
\end{equation}
where $z_i(t,k)$ denotes the $k-$th highest available unit of $i$ when $t$ arrives. In the general model of online assortments with multi-unit demand, each arrival has an associated choice model $\phi_t$ and we offer the following assortment to $t$,
\begin{equation}
\argmax_{S\subseteq S_t}\, \sum_{i\in S} \left[r_i \phi_t(S,i) \sum_{k=1}^{b_{it}}\left(1-g\left(\frac{z_i(t,k)}{c_i}\right)\right)\right].\nonumber 
\end{equation}
The key ingredients of our analysis for matching namely, conditioning and covering, are not helpful in case of budgeted allocations or assortments. To understand this in more detail, consider first the fundamental ordering property for online matching and Corollary \ref{indepcoro}. In case of budgeted allocations and assortments, for arbitrary arrival $t$ with an edge to $i$, whether unit $k$ of $i$ is matched to $t$ can depend on the past usage duration of unit $k$ itself. 
The root cause is that customers can now have different preference ordering among resources. A resource $j$ may be less preferred than resource $i$ by customer $1$ but more preferred by customer 2. This substantially complicates the stochastic dependencies that arise out of reusability. Consider the following concrete example for budgeted allocations. 
\begin{eg}
\emph{We are given two resources $\{1,2\}$, each with a reward of 1 per unit and capacity of 2 units (example can be generalized to a setting with arbitrary capacity). Let resource 2 be non-reusable and let the usage durations of resource 1 come from a two point distribution with support $\{0.5,1.5\}$ and probability $0.5$ of either possibility. Consider a sequence with three arrivals $t_1,t_2,t_3$ occurring at time $1,2,$ and $3$ respectively. The first two arrivals have a bid of $2$ for resource 1 and bid $1$ for resource 2 i.e., both arrivals prefer resource 1 if both its units are available. Arrival $t_3$ has the opposite preference and requires $1$ unit of resource 1 or 2 units of resource 2. }

\emph{Consider the actions of \dpg\ on this instance. Arrival $t_1$ is allocated 2 units of resource 1. 
If these units return before the second arrival then resource 1 is also matched to $t_2$. Overall, with non-zero probability resource 1 is matched to the first two arrivals and is still available when $t_3$ arrives. However, in this scenario resource 2 is matched to $t_3$. 
Now, consider a different sample path where the units of resource 1 allocated to arrival $t_1$ do not return by $t_2$. In this case, \dpg\ allocates a unit of the non-reusable resource 2 to arrival $t_2$. Subsequently, arrival $t_3$ is allocated a unit of resource 1. Therefore, we have two sample paths where resource 1 is available to allocate to arrival $t_3$, however, the actions of \dpg\ are different on these sample paths. This behavior is quite unlike the case of matching 
(for instance, recall Corollary \ref{indepcoro}). }
\end{eg}

Next, we construct a very similar instance for the setting of online assortments. 
\begin{eg}
\emph{ Consider three arrivals $t_1,t_2,t_3$ that come in that order and two resources $\{1,2\}$ with unit reward and unit capacity (example can be generalized to arbitrary capacity). Let resource $2$ be non-reusable. Each arrival has a multinomial logit (MNL) choice model, i.e., probability $\phi(S,i)=\frac{v_i}{v_0+\sum_{j\in S}v_j}$ for $i\in\{1,2\}$ and any set $S$ containing $i$. MNL parameters for arrivals $t_1$ and $t_2$ are as follows: $v_1=100$,  $v_2=1$ and $v_0=0.01$. 
Arrival $t_3$ has $v^3_1=1$, $v^3_2=100$ and $v^3_0=0.01$. Now, consider the actions of \dpg\ on this instance. Observe that \dpg\ offers set $\{1,2\}$ to arrival $t_1$ and with probability close to 1, resource $1$ is chosen by this arrival. Suppose that the probability of resource $1$ returning before arrival $t_2$ is $p\in(0,1)$ and resource $1$ returns before arrival $t_3$ w.p.\ 1. Then with probability $p$ we offer arrival $t_2$ the set $\{1,2\}$ and resource $1$ is chosen again w.h.p..\ Subsequently, arrival $t_3$ will choose resource $2$ w.h.p., even if resource $1$ returns and is available. In other words, resource $2$ is the most preferred available resource for arrival $t_3$ in this case. }

\emph{		On the other hand, consider the scenario where resource $1$ does not return before arrival $t_2$. Arrival $t_2$ takes resource $2$ w.h.p..\ Given that resource $2$ is non-reusable, arrival $t_3$ accepts resource 1 w.h.p..\ Therefore, whether arrival $t_3$ accepts resource $1$ depends not just on whether resource $1$ is available at arrival $t_3$, but also on the past usage duration of resource $1$ itself. If resource $1$ returns before arrival $t_2$ then arrival $t_3$ does not accept resource $1$, otherwise arrival $t_3$ accepts resource $1$ w.h.p..\ Note that on both sample paths, resource $1$ is available at arrival $t_3$. This violates key properties that enable the analysis of \dpg\ for matching (see Corollary \ref{indepcoro}).}
\end{eg}

Overall, the stochasticity in reusability interacts in a non-trivial way with the arrival dependent aspect of bids in budgeted allocations, and random choice rankings in case of assortments. Without reusability, and more specifically, without stochasticity in reusability, this interaction disappears. Indeed, as shown by \cite{feng}, the results from online assortments with non-reusable resources generalize naturally to the special case of reusable resources with deterministic usage durations.
In summary, the ingredients that enable us to derive a general framework of analysis for \dpg\ in case of online matching, do not apply more generally. Substituting for these ingredients in more general models appears to be challenging and is left open as an interesting technical question. 

\section{Faster Implementations of \galg}\label{appx:faster}

Executing \galg\ has two main bottlenecks. First, we need to update values $Y(k_i)$ by taking into account past partial matches of unit $k_i$. Second, we may fractionally match every arrival $t$ to many units and in the worst case, to all $\sum_{i\in I} c_i$ units. The first issue is somewhat less limiting, as one can in practice continue to update the states during the time between any two arrivals and update states of different units in parallel. 
The second issue is more important and has a direct impact on the time taken to decide the match for every arrival. 
\subsection{From Linear to $\log$ Dependence on Capacity}
The first approach to improve the runtime of \galg\ is based on the observation that 
we only need to find an estimate of the index of the highest available unit. In particular, we can geometrically quantize the priority index of units for every resource. The rank of each unit of resource $i$ now takes a value $\floor{(1+\epsilon)^j}$ for some $j\in\{0,1,\cdots,\floor{\log_{1+\epsilon} c_i}\}$, where parameter $\epsilon>0$ is a design choice that trades off the runtime with performance guarantee. Larger $\epsilon$ translates to smaller runtime and larger reduction in guarantee.

Since many units of a resource may now have the same index, for the sake of computation we treat all units with the same index as a single `unit'. The improvement in runtime is immediate. In each iteration of the while loop (except the last) we decrease the index of the highest available 'unit' of at least one resource. So for $\epsilon>0$, there are at most $O(\frac{1}{\epsilon}\sum_{i\in I} \log c_i)$ iterations to match each arrival. To understand the impact on the performance guarantee, notice that the reduced prices, $r_i\big(1-g(\frac{z_i(t)}{c_i})\big)$, computed for fractionally matching each arrival are off by a factor of at most $(1\pm\epsilon)$ for every resource. To address this in the analysis, we modify the candidate solution for the certificate as follows,
\[	\lambda_t=\frac{1}{1-\epsilon}\sum_{i\in I} r_i\sum_{k\in[c_i]} y(k_i,t) \bigg(1-g\Big(\frac{k}{c_i}\Big)\bigg).\]
The value of $\beta$ in condition \eqref{cert1} is now larger by a factor of $(1-\epsilon)^{-1}$. The rest of Lemma \ref{galgua} follows as is, resulting in an asymptotic guarantee of $(1-\epsilon)(1-1/e)$.
\subsection{Capacity Independent Implementation}
The second approach rests on the observation that if we could ensure that at each arrival, every unit either has at least a small $\epsilon>0$ fraction available or is completely unavailable, then the number of units any arrival is partially matched to is at most $\frac{1}{\epsilon}$. We implement this idea in \galg\ by treating each unit as \emph{unavailable unless at least $\epsilon$ fraction of it is available}. Here $\epsilon$ is a design choice; smaller the value, closer the competitive ratio guarantee to $(1-1/e)$ and larger the runtime. In particular, the time to match $t$ now reduces to $O\big(\frac{|I|}{\epsilon}\big)$, which is within a $\frac{1}{\epsilon}$ factor of the time taken by much simpler Balance and \dpg\ algorithms. 

The deterioration in competitive ratio guarantee is more challenging to unravel. Fix a resource $i$ and unit $k$, and recall the ordered set of arrivals $\mb{s}(k)$ from the analysis of Lemma \ref{galgua}. Roughly speaking, this is the set of arrivals where 
where unit $k$ is unavailable in \galg. Observe that as a consequence of treating available fraction less than $\epsilon$ as unavailable, the set $\mb{s}(k)$ is now larger. To see the implication of this, define the $(F_i,\mb{T},\mb{p}(k))$ process to capture actions of \galg, in the same way as Lemma \ref{galgua}. A key step in Lemma \ref{galgua} involves showing that the expectation $r(F_i,\mb{T},\mb{p}(k)\cup \mb{1}_{\mb{s}(k)})$, is the same as $r(F_i,\mb{T},\mb{p}(k))$. 
However, as the set $\mb{s}(k)$ is now larger, this equality does not hold. 
Therefore, for a suitably defined non-negative valued function $\kappa$, we aim to show the weaker statement,
\begin{equation}
r\big(F_i,\mb{T},\mb{p}(k)\cup \mb{1}_{\mb{s}(k)}\big)\leq \big(1+\kappa(\epsilon)\big)\, r\big(F_i,\mb{T},\mb{p}(k)\big), \label{compare}
\end{equation} 
If true, this would establish $(1-1/e)\big(1+\kappa(\epsilon)\big)^{-1}$--competitiveness (asymptotically) for \galg. Ideally, would like to show inequality \eqref{compare} with a function $\kappa$ that takes values as small as possible for every $\epsilon$. Recall that for $\epsilon=0$, we showed in Lemma \ref{galgua} that inequality \eqref{compare} holds with $\kappa(\epsilon)=0$. 
However, for non-zero but small $\epsilon$, it is not clear if inequality \eqref{compare} holds with a small value $\kappa(\epsilon)$ in general. 

Interestingly, this inequality has a remarkably strong connection to Proposition \ref{relatexp} in the analysis of \dpg. In some sense, it is equivalent to Proposition \ref{relatexp}. More concretely, one can show the inequality \eqref{compare} for the families of usage distributions where we establish validity of Proposition \ref{relatexp} in this paper. 
For instance, when distribution $F_i$ is exponential, inequality \eqref{compare} holds with the linear function $\kappa(\epsilon)=2\epsilon$, leading to a $O(n/\epsilon)$ algorithm with asymptotic guarantee $(1+\epsilon)^{-1}(1-1/e)$. More generally, we have $k(\epsilon)=\epsilon^{\eta}$ and a  $(1+\epsilon^{\eta})^{-1}(1-1/e)$ guarantee for bounded IFR distributions, where $\eta$ is as defined in case of \dpg\ (Appendix \ref{ifrexamples}) Note that while the guarantees for \dpg\ hold only for online matching, the guarantee for this modified version of \galg\ holds for budgeted allocation as well as assortment.


\section{Challenge with Small Inventory: Connection to Stochastic Rewards}\label{sec:stochrew}\label{apx:connection}

While our work provides algorithms with the best possible guarantee for reusable resources in the large inventory regime, finding an algorithm that outperforms greedy for small inventory remains open. It is worth noting that the case where all inventories are equal to 1 is the most general setting of the problem (see Proposition 1 in \cite{reuse}). In this section, we shed new light on the difficulty of this problem by establishing a connection between a very special case of reusability and the well studied problem of online matching with stochastic rewards.

Consider the setting where matched resources return immediately and can be re-matched to subsequent arrivals i.e., usage durations are deterministically 0. It is not surprising that the greedy solution is optimal for this instance, as the capacity of each resource is virtually unlimited. At the other end of the spectrum is the case of non-reusable resources where matched units never return. 
Now consider perhaps the simplest 
setting that captures both these extreme cases where every matched unit returns immediately with probability $p$, and never returns (usage duration $+\infty$) w.p. $1-p$. This setting isolates a key aspect of reusability -- 
the stochastic nature of the problem. At first glance one might expect this setting to be straightforward given the observations for the extreme cases where $p=0$ or $p=1$. As it turns out, the general case is more interesting. 

To make this formal, we consider the stochastic rewards problem of \cite{deb} for non-reusable resources. This problem generalizes online matching by associating a probability of success $p_{it}$ with every edge $(i,t)\in E$. 
When a match is made i.e., edge is chosen, it succeeds independently with this probability. 
If the match fails the arrival departs but the resource is available for future rematch. The goal is to maximize the expected number of successful matches. We show the following connection.  

\begin{lemma}\label{connection}
The problem of online matching with reusability where resources have identical two point usage distributions supported on $\{0,+\infty\}$, is equivalent in the competitive ratio sense to the problem of online matching with stochastic rewards and identical edge probabilities,  i.e., an $\alpha$--competitive online algorithm in one setting can be translated to an $\alpha$--competitive online algorithm in the other. 
\end{lemma}

Proof of the lemma is presented later in this section. For small inventory, the stochastic rewards problem is well known to be fundamentally different from classic online matching. 
For instance, \cite{deb} showed that when comparing against a natural LP benchmark, no online algorithm can have a guarantee better than $0.621<(1-1/e)$ for stochastic rewards with small inventory, even with identical probabilities. While a recent result for stochastic rewards shows that this barrier can be circumvented by comparing directly against offline algorithms instead of LP benchmark\footnote{\cite{stochrew} give a $(1-1/e)$ result for instances of stochastic rewards with decomposable edge probabilities i.e., when for every edge $(i,t)\in E$ the probability can be decomposed as a product $p_{it}=p_i\times p_t$.}, in general, the setting of \emph{reusable resources sharply diverges and becomes much harder than stochastic rewards}.  
In particular, the stochastic rewards problem with heterogeneous edge probabilities admits a $1/2$--competitive result for arbitrary inventory.\footnote{This is achieved with greedy algorithms \citep{mehta,negin}. Improving this is an open problem. A tight $(1-1/e)$ result is known for large capacity \citep{msvv,negin}.}
In contrast, the corresponding generalization for reusable resources, where 
the probability of immediate return is arrival/edge dependent, does not admit any non-trivial competitive ratio result even for large capacity (Theorem 2 in \cite{reuse}).

To prove Lemma \ref{connection}, we first introduce an equivalent form of stochastic rewards where the reward is deterministic and independent of the success/failure of matching. Formally, consider the stochastic rewards setting with identical success probability $p$ for every edge and resources with unit reward. We transform this to an instance of the following problem:

\textbf{Online matching with stochastic consumption:} Each edge has a probability $p$ of success. We assume that $p>0$. If an arrival is matched to some resource, we earn a unit reward. After each match, a unit of the matched resource is used forever w.p.\ $p$, independent of other outcomes. With the remaining probability $1-p$, we do not lose a unit of the resource. Recall, we earn a unit reward in either realization.

We can denote an instance of either of these problems simply as $(G,p)$, where $G$ is the graph and $p>0$ is the edge probability. Now, consider the family of \emph{non-anticipative algorithms} for these problems, i.e., algorithms (online or offline) that do not know the realization of any match beforehand. Every online algorithm is naturally non-anticipative. Offline algorithms such as the clairvoyant benchmark, as well the stronger fully offline benchmark that can match arrivals in an arbitrary sequence \citep{stochrew}, are both non-anticipative. Evaluating competitive ratios against non-anticipative offline algorithms, we have the following result. 
\begin{lemma}
The problem of online matching with stochastic consumption (with probability $p>0$) is equivalent in the competitive ratio sense to the problem of online matching with stochastic rewards and identical edge probabilities,  i.e., an $\alpha$--competitive online algorithm in one setting can be translated to an $\alpha$--competitive online algorithm in the other. 
\end{lemma}
\begin{proof}
Consider an instance of the stochastic rewards problem and a non-anticipative algorithm $\mathcal{A}$, that can be offline or online. 
Consider the alternate reward function where each time $\mathcal{A}$ makes a match we obtain a deterministic reward $p$ regardless of the outcome of the match. Due to non-anticipativity of $\mathcal{A}$ and using the linearity of expectation, the expected total \emph{alternative} reward of $\mathcal{A}$ is the same as its expected total reward. 

Now, consider an instance $(G,p)$ of the (online matching with stochastic) consumption problem with $p>0$.  Given algorithm $\mathcal{A}$ for the stochastic rewards problem, we can obtain an algorithm $\mathcal{A}'$ for the stochastic consumption problem by simulating $\mathcal{A}$ on a coupled instance $(G,p)$ of the stochastic reward problem. 
The total expected reward of $\mathcal{A}'$ for the consumption problem is $\frac{1}{p}$ times the alternative reward of $\mathcal{A}$ on instance $(G,p)$ for the stochastic rewards problem. 

Observe that we can proceed in the reverse direction with similar arguments, i.e., given an algorithm $\mathcal{B}'$ for the consumption problem, we can construct an algorithm $\mathcal{B}$ for the stochastic rewards problem via simulating $\mathcal{B}'$ on a coupled instance of the consumption problem. The expected reward of $\mathcal{B}$ on an instance $(G,p)$ of the stochastic rewards problem is $p$ times the expected reward of $\mathcal{B}'$ on the instance $(G,p)$ of the consumption problem. 

Since these arguments hold for both online and offline algorithms, the constant factor of $p$ (or $1/p$) cancels out for $p>0$ and we have the desired competitive ratio equivalence.
\hfill\Halmos\end{proof}

\begin{proof}{Proof of Lemma \ref{connection}.}
It is now easy to see that the stochastic consumption problem is equivalent to the setting of online matching with reusable resources when the usage distributions for every resource is supported on $\{0,+\infty\}$, with probability of return $1-p$. To make this connection, we simply interpret unsuccessful consumption of a resource in the stochastic consumption setting as the resource returning with duration $0$ in the reusable resources setting, and vice versa. 
Given this interpretation, we can now directly use an algorithm from one setting in the other setting with the same expected reward.
\hfill\Halmos\end{proof} 

More generally, consider the stochastic rewards problem with heterogeneous edge probabilities $p_{it}$. 
The greedy algorithm that matches each arrival to the resource with highest expected reward is $1/2$--competitive for this general problem. Now, consider the corresponding generalization in the reusable resource setting 
with two point usage distributions supported on $\{0,+\infty\}$ and return probability $1-p_{it}$ for edge $(i,t)$. This problem does not admit any constant factor competitive ratio result (Theorem 2 in \cite{reuse}). The proof of equivalence breaks down since the expected reward of a match $(i,t)$ in the stochastic rewards setting is now $p_{it}$. In contrast, the reward for a match in the reusable resource setting is 1, as before. As the ratio between these rewards is now arrival dependent (unlike the case of identical probabilities) it does not cancel out as a constant factor when evaluating the competitive ratios.

\section{Miscellaneous}
\subsection{Impossibility for Stronger Benchmark}\label{appx:stronger}
For online matching with reusable resources consider the offline benchmark that in addition to the arrival sequence also knows the realizations of all usage durations in advance. The following example illustrates that no non-trivial competitive ratio result is possible against this benchmark. 

Consider a setting with $n$ resources. Resources have identical reward and usage distribution. Consider a two point distribution uniformly supported on $\{0,\infty\}$. Suppose we see $n^2$ arrivals, each with an edge to all $n$ resources. The expected reward of any online algorithm is at most $2n$, whereas an offline algorithm that knows the realizations of all durations in advance can w.h.p.\ match all $n^2$ arrivals as $n\to\infty$.
\subsection{Clairvoyant is Deterministic}\label{appx:miscfrac}
\begin{lemma}
There exists a deterministic algorithm that is optimal among the class of all offline algorithms that know the entire arrival sequence but match (or decide assortments) in order of arrival and do not know realizations of stochastic elements (usage durations and customer choice) in advance. 
\end{lemma} 
\begin{proof}{Proof.}
Given an arrival sequence the optimal algorithm is given by a dynamic program with the state space given by the number of arrivals remaining and the availability status of the resources, i.e., for each unit, whether it is currently available or in-use and how long it has been in-use for. 
The decision space of clairvoyant is simply the assortment decision for the current arrival in the sequence. Let $V(t,S_t)$ denote the optimal value-to-go at arrival $t$ given that the state of resources is $S_t$. We use $\onee(i,S_t)$ to indicate if a unit of $i$ is available in state $S_t$ and recall that $\mathcal{F}_t$ denotes the set of feasible assortments at $t$. Let $R\left(A,t+1\mid S_t\right)$ denote a possible state of the resources at arrival $t+1$, given state $S_t$ and assortment $A$ at arrival $t$. Let $p\left(R\left(A,t+1\mid S_t\right)\right)$ denote the conditional probability of occurrence for this state and let $\Omega(A,t+1\mid S_t)$ denote the set of all possible states $R\left(A,t+1\mid S_t\right)$. Clearly,
\[V(t,S_t)=\max_{A\in \mathcal{F}_t \mid  \onee(i,S_t)=1, \forall i\in A} \left(R_t(A)+\sum_{R\left(A,t+1\mid S_t\right)\in \Omega(A,t+1\mid S_t)} p\left(R\left(A,t+1\mid S_t\right)\right)\, V\left(t+1,R(A,t+1\mid S_t)\right)\right), \]
where $R_t(A)=\sum_{i\in A} r_i \phi_t(A,i)$. At the last arrival $T$, the future value to go (second term in the above sum) is zero and for any given state $S_T$ of resources at $T$, the optimal decision at $T$ is simply the (deterministic) solution to a constrained assortment optimization problem. Performing a backward induction using the above equation, we have for any given set of values $V\left(t+1,R(A,t+1\mid S_t)\right)$, the optimal assortment decision at arrival $t$ is deterministic. 
\hfill\Halmos\end{proof}
\subsection{Clairvoyant Matches Fractional LP for Large Capacities}\label{appx:LP}
Consider the following natural LP upper bound for online matching with reusable resources \citep{dickerson,baek,feng},
\begin{eqnarray}
OPT(LP)= \max &&  \sum_{(i,t)\in E} r_i y_{it}\nonumber\\
s.t.  &&\sum_{t=1}^{\tau}[1-F_i(a(\tau)-a(t))]y_{it}\leq c_i \quad \forall \tau\in\{1,\dots,T\},\forall i\in I \nonumber\\
&& \sum_{i\in I} y_{it} \leq 1\quad \forall t\in T\nonumber\\
&& 0\leq y_{it}\leq 1 \quad \forall t\in T, \quad i\in I
\end{eqnarray}
Clearly, $\opt \leq \opt (LP)$ and the allocations generated by any algorithm (offline or online) can be converted into a feasible solution for the LP, regardless of $c_i$. Perhaps surprisingly, we show that for large $c_i$, the solution to this LP can be turned into a randomized clairvoyant algorithm (that does not know the realizations of usage in advance) with nearly the same expected reward, implying that the LP gives a tight asymptotic bound and moreover, all  asymptotic competitive ratios shown against the clairvoyant also hold against the LP. 

\begin{theorem}Let $c_{\min}=\min_{i\in I}c_i$. Then,
\[\opt(LP)\Big(1-O\Big(\sqrt{\frac{\log c_{\min}}{c_{\min}}}\Big)\Big)\leq \opt \leq \opt(LP).\] 
Hence, for $c_{\min}\to +\infty$, $\opt \to \opt (LP)$.
\end{theorem}
\begin{proof}{Proof.}
We focus on the lower bound and more strongly show that every feasible solution of the LP can be turned into an offline algorithm with nearly the same objective value. Let $\{y_{it}\}_{(i,t)\in E}$ be a feasible solution for the LP. Consider the offline algorithm that uses the LP solution as follows,
\[\text{ When $t$ arrives, sample a resource $i$ to offer, according to the distribution $\{y_{it}/(1+2\delta)\}_{i\in I}$,} \]
where  $\delta=\sqrt{\frac{\log c_{\min}}{c_{\min}}}$. Note that if the sampled resource is unavailable, the algorithm leaves $t$ unmatched. Also w.p., $1-\sum_{i\in I}y_{it}$, the algorithm rejects $t$. Since the LP does not use usage durations, the offline algorithm doesn't either. The critical element to be argued is that the expected reward of this algorithm is roughly the same as the objective value for the feasible solution. This holds due to concentration bounds and the argument closely mimics the proof of Lemma \ref{algvgalg}. Finally, since this offline algorithm makes matching decisions in order of the arrival sequence and does not know realizations of usage durations in advance, its performance gives a lower bound on the performance of clairvoyant, i.e., \opt. 
\hfill\Halmos\end{proof}
\emph{Remark:} This result generalizes naturally to the settings of online assortment and budgeted allocations. In case of assortments we use the Probability Matching algorithm from Appendix \ref{appx:asst} to use the concentration bounds.
\subsection{Sufficiency of Static Rewards}\label{appx:miscstatic}
\begin{lemma}
Given an algorithm (online or clairvoyant) for allocation that does not know the realizations of usage durations, the total expected reward of the algorithm is the same if we replace dynamic usage duration dependent rewards functions $r_i(\cdot)$ with their static (finite) expectations $r_i=\mathbb{E}_{d_i\sim F_i} [r_i(d_i)]$, for every resource. 
\end{lemma}
\begin{proof}{Proof.}
Consider arbitrary algorithm $\mathbb{A}$ as described in the lemma statement. Let $\mathbb{B}$ denote an algorithm that mimics the decisions of $\mathbb{A}$ but receives static rewards $r_i, \, \forall i\in I$ instead. Now, suppose $\mathbb{A}$ successfully allocates resource $i$ to arrival $t$ on some sample path $\omega(t)$ observed thus far. Then, conditioned on observing $\omega(t)$, the expected reward from this allocation is exactly $r_i$. More generally, using the linearity of expectation it follows that the total expected reward of $\mathbb{B}$ is the same as that of $\mathbb{A}$.  

Note that for algorithms that also know the realization of usage durations in advance, the decision of allocation can depend on usage durations and conditioned on observed sample path $\omega(t)$, the expected reward from successful allocation of $i$ to $t$ need not be $r_i$. 
\hfill\Halmos\end{proof}

\subsection{Upper Bound for Deterministic Arrival Dependent Usage} \label{appx:miscdet}
\begin{lemma}
Suppose the usage duration is allowed to depend on the arrivals such that a resource $i$ matched to arrival $t$ is used for deterministic duration $d_{it}$ (revealed when $t$ arrives). Then there is no online algorithm with a constant competitive ratio bound when comparing against offline algorithms that know all arrivals and durations in advance. 
\end{lemma}
\begin{proof}{Proof.}
Using Yao's minimax, it suffices to show the bound for deterministic online algorithms over a distribution of arrival sequences. For simplicity, suppose we have a single unit of a single resource and a family of arrival sequences $A(j)$ for $j\in[n]$ (the example can be naturally extended to the setting of large capacity). The arrivals in the sequences will be nested so that all arrivals in $A(j)$ also appear in $A(j+1)$. $A(1)$ consists of a single arrival with usage duration of 1. Suppose this vertex arrives at time $0$. $A(2)$ additionally consists of two more arrivals, each with usage duration of $1/2-\epsilon$, arriving at times $\epsilon$ and $1/2$ respectively. More generally, sequence $A(j)$ is best described using a balanced binary tree where every node represents an arrival and the depth of the node determines the usage duration. Each child node has less than half the usage duration ($d/2-\epsilon$) of its parent ($d$). If the parent arrives at time $t$, one child arrives at time $t+\epsilon$ and the other at $t+d/2$. The depth of the tree for sequence $A(j)$ is $j$ (where depth 1 means a single node). Note that the maximum number of arrivals that can be matched in $A(j)$ is $2^{j-1}$.

Let $Z=\sum_{j=1}^n 2^j$. Now, consider a probability distribution over $A(j)$, where probability $p_j$ of sequence $A(j)$ occurring is $\frac{2^{n-j+1}}{Z}$. Clearly, an offline algorithm that knows the full sequence in advance can match $2^{j-1}$ arrivals on sequence $A(j)$ and thus, has revenue $n2^n/Z=n/2$. It is not hard to see that the best deterministic algorithm can do no better (in expectation over the random arrival sequences) than trying to match all arrivals with a certain time duration. Any such deterministic algorithm has revenue at most $Z/Z=1$. Therefore, we have a competitive ratio upper bound of $n/2$. 
\hfill\Halmos\end{proof}
{\color{black}

\section{Numerical Experiments: Missing Details}\label{appx:exp}
\subsection{LP Benchmark}
Let $p_{it}$ denote the probability that arrival $t$ has an edge to $i$. Observe that, if $t$ is a bursty arrival in phase $k$, then 
\begin{equation*}
	p_{i,t}=\begin{cases}
		1 & i=2n-k+1,\\
		0 & \text{otherwise}.
	\end{cases}
\end{equation*}
If $t$ is a normal arrival in phase $k$, then 
\begin{equation*}
	p_{i,t}=\begin{cases}
		\frac{\sum_{\ell\geq i} e^{-\kappa|\ell-(n-k+1)|}}{\sum_{\ell\geq 1} e^{-\kappa|\ell-(n-k+1)|}} & \forall i\in[n],\\
		0 & \text{otherwise}.
	\end{cases}
\end{equation*}
Clairvoyant knows the arrival sequence for each random instance. The optimal solution of the following LP is an upper bound on the expected total reward of clairvoyant.
\begin{eqnarray}
	\textbf{LP benchmark}\quad  \max &&  \sum_{i\in [n], t\in [2cn]} y_{i,t}\nonumber\\
	s.t.  &&\sum_{t=1}^{\tau}[1-F(a(\tau)-a(t))]y_{i,t}\leq c \quad \forall \tau\in[2cn],\forall i\in [n] \nonumber\\
	&& \sum_{i\in [n]} y_{i,t} \leq 1\quad \forall t\in [2cn]\nonumber\\
	&& 0\leq y_{i,t}\leq p_{i,t} \quad \forall t\in T[2cn], \quad i\in [n].\label{ub}
\end{eqnarray}
This LP benchmark closely resembles the standard LP relaxation of clairvoyant (see Appendix \ref{appx:LP}). The key difference is that we impose the upper bounds $y_{i,t}\leq p_{i,t}$ (see \eqref{ub}).
Let $E_{i,t}$ denote the event that there is an edge between resource $i$ and arrival $t$. Let $O_{i,t}$ denote the event that the clairvoyant matches $t$ to $i$ conditioned on the event $E_{i,t}$. To see that this LP is an upper bound on the expected performance of clairvoyant, observe that, if, for every $i\in[n]$ and $t\in[2cn]$, we set the decision variable $y_{i,t}$ as the probability that both $E_{i,t}$ and $O_{i,t}$ occur then all constraints in the LP will be satisfied and the LP objective represents the expected total reward of clairvoyant. 

\subsection{Other Scenarios}

For Tables \ref{suptab1} -- \ref{suptab3}, the performance of each algorithm (\alg) in the table is reported as the ratio of the empirical average performance of \alg\ (based on $20\times100$ trials) and the optimal value of the LP benchmark. Note that the standard deviation of the reported ratios is less than $0.0001$ for all algorithms except \salg, for which it is less than $0.01$. 
	\begin{table}[h]
			\centering
			\begin{tabular}{|c|c|c|c|c|c|}
				\hline
				\textbf{$c$} &\textbf{$F$}&\text{$\quad$\textbf{\dpg}$\quad$} & \text{$\quad$\textbf{Balance}$\quad$} & \text{$\quad$\textbf{Greedy}$\quad$} &\textbf{\salg} \\ \hline
				\multirow{3}{*}{5} & Two-point & 0.72485 & 0.72485 & 0.58894 & 0.5129  \\ \cline{2-6}
				& Exponential  & 0.9901 & 0.9932 & 0.9648 & 0.98907 \\ \cline{2-6}
				& Weibull & 0.9972 & 0.9968 & 0.9732 & 0.91239 \\ \hline
				\multirow{3}{*}{15} & Two-point   & 0.79477 & 0.79477 & 0.61379 & 0.6192  \\ \cline{2-6}
				& Exponential & 0.99907 & 0.9996 & 0.97867 & 0.99907\\ \cline{2-6}
				& Weibull & 0.9999 & 0.9998 & 0.9844 & 0.9999 \\ \hline
			\end{tabular}
			\caption{Average performance of online algorithms in comparison to the LP benchmark when the arrival sequence only includes the sequence of normal arrivals ($T_2$), i.e., the bursty arrivals are excluded. The results are for \( n = 5 \) and \( \kappa = 1 \).  Observe that \dpg, Balance, and greedy are all close to optimal in most scenarios and \dpg\ and Balance have very comparable performance. For two-point usage distribution, we believe that the overall instance is relatively closer to the worst case instance for non-reusable resources and this may be why greedy performs relatively poorly.}\label{suptab1}
	\end{table}
	
	\begin{table}[h]
	\centering
	\begin{tabular}{|c|c|c|c|c|c|}
		\hline
		\textbf{$c$} &\textbf{$F$}&\text{$\quad$\textbf{\dpg}$\quad$} & \text{$\quad$\textbf{Balance}$\quad$} & \text{$\quad$\textbf{Greedy}$\quad$} &\textbf{\salg} \\ \hline
		\multirow{3}{*}{5} & Two-point & 0.87617 & 0.85191 & 0.81745  & 0.75748 \\ \cline{2-6}
		& Exponential    & 0.92708 & 0.92266 & 0.91062  & 0.79072  \\ \cline{2-6}
		& Weibull  & 0.94267 & 0.94025 & 0.93380 & 0.81921 \\ \hline
		\multirow{3}{*}{15} & Two-point & 0.90468 & 0.86780 & 0.81277  & 0.82592  \\ \cline{2-6}
		& Exponential & 0.91876 & 0.90630 & 0.90978 & 0.8595\\ \cline{2-6}
		& Weibull & 0.94677 & 0.93587 & 0.93102 & 0.884263 \\ \hline
		\multirow{3}{*}{25} & Two-point & 0.91813 & 0.87634 & 0.82596 & 0.90428 \\ \cline{2-6}
		& Exponential & 0.91950 & 0.90333 & 0.90711 & 0.91942 \\ \cline{2-6}
		& Weibull & 0.94641 & 0.93414 & 0.93050 &0.945711 \\ \hline
	\end{tabular}
	\caption{Average performance of algorithms for $n=5$ and $\kappa=0$ in comparison to LP benchmark. Comparing the numbers with Table \ref{maintab} for $n=5$ and $\kappa=1$, one can see that all algorithms perform better for $\kappa=0$. Greedy now slightly outperforms the other algorithms in some scenarios. \dpg\ continues to dominate Balance but the gap between them is much smaller as compared to the scenario where $\kappa=1$.} \label{suptab2}
\end{table}

\begin{table}[h]
	\centering
	\begin{tabular}{|c|c|c|c|c|c|}
		\hline
		\textbf{$c$} &\textbf{$F$}&\text{$\quad$\textbf{\dpg}$\quad$} & \text{$\quad$\textbf{Balance}$\quad$} & \text{$\quad$\textbf{Greedy}$\quad$} &\textbf{\salg} \\ \hline
		\multirow{3}{*}{5} & Two-point  & 0.73806 & 0.72307 & 0.68854 & 0.6382  \\ \cline{2-6}
		& Exponential   & 0.9675 & 0.9665 & 0.9751 & 0.91201  \\ \cline{2-6}
		& Weibull  & 0.97312 & 0.97202 & 0.97742 & 0.95657 \\ \hline
		\multirow{3}{*}{15} & Two-point  & 0.74498 & 0.72151 & 0.68462 & 0.73217  \\ \cline{2-6}
		& Exponential & 0.95493 & 0.94877 & 0.9732 & 0.9441\\ \cline{2-6}
		& Weibull & 0.96362 & 0.95792 & 0.97379 & 0.96991 \\ \hline
		\multirow{3}{*}{25} & Two-point& 0.76864 & 0.74169 & 0.69747 & 0.75311 \\ \cline{2-6}
		& Exponential & 0.94867 & 0.93771 & 0.96965  & 0.94631 \\ \cline{2-6}
		& Weibull & 0.95822 & 0.94865 & 0.96916 &0.95812 \\ \hline
	\end{tabular}
	\caption{Average performance of algorithms for \( n=20 \) and \( \kappa=1 \) in comparison to LP benchmark. The overall trend in this case is similar to the case of $n=5$ and $\kappa=0$ (see Table \ref{suptab2} above). For two-point usage distribution, we believe that for a larger value of $n$ the overall instance is relatively closer to the worst case instance for non-reusable resources and this may be why greedy performs relatively poorly.}\label{suptab3} 
\end{table}

\begin{table}[h]
	\centering
	\begin{tabular}{|c|c|c|c|c|c|}
		\hline
		\textbf{$c$} &\textbf{$F$}&\text{$\quad$\textbf{\dpg}$\quad$} & \text{$\quad$\textbf{Balance}$\quad$} & \text{$\quad$\textbf{Greedy}$\quad$} &\textbf{\salg} \\ \hline
		\multirow{3}{*}{5} & Two-point  & 0.84941 & 0.82502 & 0.77134 & 0.68801  \\ \cline{2-6}
		& Exponential    & 0.9785  & 0.9776  & 0.9757 & 0.91239  \\ \cline{2-6}
		& Weibull  & 0.97984 & 0.97944 & 0.97654 & 0.96112 \\ \hline
		\multirow{3}{*}{15} & Two-point  & 0.88137 & 0.84527 & 0.77988 & 0.75819  \\ \cline{2-6}
		& Exponential & 0.97667 & 0.97393 & 0.9735 & 0.95613\\ \cline{2-6}
		& Weibull & 0.9784  & 0.97583 & 0.97583 & 0.97014 \\ \hline
		\multirow{3}{*}{25} & Two-point& 0.89713 & 0.86034 & 0.82436 & 0.76183 \\ \cline{2-6}
		& Exponential & 0.92088 & 0.90234 & 0.91402 & 0.91015 \\ \cline{2-6}
		& Weibull & 0.94488 & 0.93114 & 0.93002 &0.92137 \\ \hline
	\end{tabular}
	\caption{Average performance of algorithms for \( n=20 \) and \( \kappa=0 \) in comparison to LP benchmark. Comparing the numbers with Table \ref{suptab3} for $n=20$ and $\kappa=1$, one can see that all algorithms perform better for $\kappa=0$. \dpg\ continues to dominate all algorithms by a small margin. Greedy now performs comparably with Balance in some scenarios.  For two-point usage distribution, we believe that for a larger value of $n$ the overall instance is relatively closer to the worst case instance for non-reusable resources and this may be why greedy performs relatively poorly.}\label{suptab4} 
\end{table}
}
\end{APPENDICES}
\end{document}